\documentclass[12pt,a4paper,twoside]{article}

\usepackage{fancyhdr,graphicx}

\usepackage{amsmath,amsfonts,amsthm}
\usepackage[numbers,sort&compress,merge]{natbib}
\DeclareMathOperator{\tr}{Tr} 

\DeclareMathOperator{\ipart}{Im}

\DeclareMathOperator*{\res}{Res}

\numberwithin{equation}{section}

\newcommand{\p}{\partial}

\newcommand{\abs}[1]{\left \lvert #1 \right \rvert}

\theoremstyle{plain}
\newtheorem{theorem}{Theorem}[section]
\newtheorem{proposition}{Proposition}[section]
\newtheorem{lemma}{Lemma}[section]
\newtheorem{corollary}{Corollary}[section]

\theoremstyle{definition}
\newtheorem{remark}{Remark}[section]

\begin{document}
\title{A matrix model with a singular weight and Painlev\'e
  III\footnote{The authors acknowledge financial support by the EPSRC
    grant EP/G019843/1.}}  \author{L. Brightmore, F. Mezzadri and
  M. Y. Mo.}  \date{}
\maketitle
\begin{abstract}
We investigate the matrix model with weight
\begin{equation*}
  w(x) := \exp\left(-\frac{z^2}{2x^2}+\frac{t}{x} 
  -\frac{x^2}{2} \right)
\end{equation*}
and unitary symmetry. In particular we study the double scaling limit
as $N \to \infty$ and $(\sqrt{N} t, Nz^2 ) \to (u_1,u_2)$, where $N$
is the matrix dimension and the parameters $(u_1,u_2)$ remain
finite. Using the Deift-Zhou steepest descent method we compute the
asymptotics of the partition function when $z$ and $t$ are of order
$O\bigl(N^{-1/2}\bigr)$. In this regime we discover a phase transition
in the $(z,N)$-plane characterised by the Painlev\'e III equation.
This is the first time that Painlev\'e III appears in studies of
double scaling limits in Random Matrix Theory and is associated to the
emergence of an essential singularity in the weighting function.  The
asymptotics of the partition function is expressed in terms of a
particular solution of the Painlev\'e III equation. We derive
explicitly the initial conditions in the limit
$Nz^2\rightarrow u_2$ of this solution.
\end{abstract}
\vspace{.25cm}

\hspace{.26cm} 2010 MSC: 15B52, 35Q15.
\maketitle
\tableofcontents
\section{Introduction}
\label{introduction}
\subsection{Background}
\label{background}
The purpose of this article is to study the asymptotics as $N \to
\infty$ of the  partition function
\begin{equation}
  \label{eq:main_average}
E_N(z,t):=  \frac{1}{N!} \int_{\mathbb{R}^N} \prod_{j=1}^N
 \exp\left(-\frac{z^2}{2x_j^2} + \frac{t}{x_j} -\frac{x^2_j}{2}\right)
   \prod_{1 \le j < k \le N}\abs{x_k - x_j}^2 d^Nx,  
\end{equation}
where $z \in \mathbb{R} \setminus \left \{0 \right \}$ and $0\le t
<\infty$; the particular case $z=t=0$ corresponds to the partition
function of the Gaussian Unitary Ensemble (GUE).  This multiple
integral %~\eqref{eq:main_average}
belongs to a general class 
\begin{equation}
  \label{eq:partition_functions}
 \frac{1}{N!} \int_{J^N} \prod_{j=1}^N \exp\bigl(-V(x)\bigr)
\prod_{1 \le j < k \le N}\abs{x_k -    x_j}^2 d^Nx, \quad J \subseteq
\mathbb{R}, 
\end{equation}
where $w(x)=\exp\bigl(-V(x)\bigr)$ is the \textit{weighting function}
and $V(x)$ is known as the \textit{potential.}  These integrals have
been the subject of extensive investigations in Random Matrix Theory
(RMT), because they contain all the information on the correlations of
the eigenvalues and are the starting point to study their linear
statistics as well as global fluctuations of the spectra.  

The potential of the GUE is $V(x) = - x^2/2$; therefore, $E_N(z,t)$
can be thought of as the partition function of a matrix model obtained
by perturbing the GUE potential with a first and a second order pole,
with $t$ and $z$ measuring the strength of the perturbations. Then, a
natural question arises: what happens as the
average~\eqref{eq:main_average} approaches the GUE partition function?
The main result of this paper is that as $N \to \infty$ while
$N^{\frac{1}{2}}t $ and $N^{\frac{1}{2}}z$ converge to finite
constants, a phase transition emerges where the asymptotics of
$E_N(z,t)$ is characterized by a solution of the Painlev\'e III (PIII)
equation.

In most applications the potential $V(x)$ is required to have some
regularity properties; for example, imposing that $V(x)$ should be
real analytic together with appropriate boundary conditions suffices
in many cases, as it guarantees that the limiting mean density of the
eigenvalues, known as \emph{equilibrium measure}, is supported on a
finite union of intervals.  However, recently matrix models whose
weight $w(x)$ has an essential singularity have appeared in several
area of mathematics and physics, like number theory, quantum transport
and finite-temperature field theory (see,
\textit{e.g.},~\cite{BFB97,BS08,CI10,Luk01,MS12b}).  For such
singular potentials the asymptotic analysis of the partition function,
or of any statistics of the spectra, becomes substantially involved
and no studies of the double scaling limits are available.

The eigenvalues of a matrix ensemble with partition
function~\eqref{eq:partition_functions} form a determinantal point
process.  It is remarkable that with an appropriate choice of the
scaling limit, as $N \to \infty$ the kernel of this process becomes
universal and depends only on the local properties of the equilibrium measure
(see, \emph{e.g.},~\cite{Dei99book} and references therein).  Finding
the universal kernel can be reduced to the asymptotics analysis in the
same scaling limit of the polynomials orthogonal with respect to
$w(x)$.  By standard theory of orthogonal polynomials~\cite[Chap. 2.2]{Sze39} and
by the results in~\cite{BEH06}, the asymptotics of $E_N(z,t)$ can
expressed in terms of the same system of orthogonal polynomials.

In order for the integral~\eqref{eq:main_average} to converge either
$z\neq 0$ or $z=t=0$, when it is simply the partition function of the
GUE whose equilibrium measure is the \textit{semicircle law,} which is
supported in an interval symmetric with respect to the origin. In the
limit $N \to \infty$ the parameter $t$ does not contribute to the
equilibrium measure because with the correct scaling it becomes
asymptotically negligible.  Thus, the limiting values $z=0$ is
bifurcation point: away from it, the second order pole in the
potential splits the support of the equilibrium measure in two
intervals symmetric with respect to the origin~\cite{MM09}. Previous
studies of double scaling limits in matrix models have concerned
various types of critical points of the equilibrium measure. When it
vanishes quadratically inside the support, the universality of the
spectra correlations is characterized in the double scaling limit by
the second Painlev\'e equation~\cite{BI99,BI03,CK06,CKV08}.  In this
case too, two intervals in the support of the equilibrium measure
coalesce into one at the critical point; however, the main difference
from the model of the present paper is that the potential is real
analytic and has not any singularity at the critical point.  This
changes the nature of the problem entirely.  When the equilibrium
measure vanishes as the power $5/2$ at an endpoint of its support,
then in the double scaling limit universality is identified by a
fourth order analogue of Painlev\'e one~\cite{CV07}.

This is the first time that a critical phenomenon associated to an
essential singularity in the weighting function has been studied. It
seems that the characterization of the double scaling limit by PIII is
specific to a pole emerging in the background of a smooth potential
and appears to be a new universal feature of the spectra of unitary
matrix ensembles.  This property is likely to be shared by other
matrix models whose weighting function has the same type of
singularity.

 %and
% its applications go beyond the direct and we expect that this paper
% will be the first step in obtaining a new universality class relating
% matrix models with this type of singular behaviour to Painlev\'e III.

%Indeed, although usually double scaling limits in matrix models are described by Painlev\'e equations (see, \textit{e.g.},~\cite{BI99,BI03,CK06,CKV08,CV07}), the appearance of PIII seems to be a new feature specific to 

The average~\eqref{eq:main_average} was introduced by Berry and
Shukla~\cite{BS08} in their study of the random function 
\begin{equation}
  \label{eq:t_fun}
  Q_N(x) := \frac{\Lambda_N^{\prime 2}(x)}{\Lambda_N^{\prime 2}(x) -
    \Lambda_N(x)\Lambda_N^{\prime \prime }(x)},
\end{equation}
where $\Lambda_N(x) := \prod_{j=1}^N(x - x_j)$ and $x_1,\ldots,x_N$ are
a set of random variables --- in general not independent.  The value
distribution $P(Q_N)$ is important mainly for two reasons: firstly, it
is a sensitive indicator of the degree of the repulsion of two
neighbouring $x_j$'s, in the sense that the rate of decay of $P(Q_N)$
is a measure of the rigidity of $x_1,\ldots,x_N$; secondly, if
$\Lambda_N(x)$ is replaced by the Riemann zeta function (or more
precisely the Hardy function, which is real on the critical line),
then the Riemann hypothesis implies that $Q_N(x) >0$. Therefore,
$P(Q_N)$ provides valuable information on the statistics of the
zeros of the Riemann zeta function.

The connection between RMT and the theory of the Riemann zeta function
suggests to replace $\Lambda_N(x)$ with the characteristic polynomial
of a random matrix from the GUE.  In this case an explicit formula for
the probability density $P(Q_N)$ is very difficult to find.  However,
all the information on $P(Q_N)$ is contained in the
integral~\eqref{eq:main_average}.  It is straightforward to see from
the definition that $E_N(z,t)$ is a real analytic function of $z$ and $t$
and is even.  This symmetry suggests the introduction of the power
expansion
\begin{equation}
  \label{eq:p_ex_p}
  E_{N}(z,t)= \sum_{m=0}^\infty E_{N\, 2m}(z)t^{2m},
\end{equation}
which for $z\neq 0$ converges uniformly in $t$ in any closed subset of
$\mathbb{R}_+$ including the origin. Therefore, $E_N(z,t)$ can be
interpreted as the generating function of the coefficients
$E_{N\,2m}(z)$. Berry and Shukla~\cite{BS08} showed that the moments
of $P(Q_N)$ are given by
\begin{equation}
  \label{eq:moments}
  M_{m}= 2^{1-m}\left(\prod_{j=m}^{2m}j\right) \int_0^\infty z^{2m-1}E_{N\,2m}(z)
  dz.   
\end{equation}

In a previous article~\cite{MM09} we showed that in the range 
$c_1 N^{-\frac12} < \abs{z} < c_2N^{\frac14}$, where $c_1,c_2 >0$ 
are independent of $N$, we have   
\begin{equation}
\label{eq:asymdet}
\begin{split}
  E_{N}(z,t)&=B_{N}\exp\left(\frac{z^2}{4}-
    \frac{9}{2^{\frac{10}{3}}}\left(N^{\frac{2}{3}}z^{\frac{4}{3}}-1\right)
    +\frac{t^2N^{\frac{1}{3}}}{2^{\frac{5}{3}}z^{\frac{4}{3}}}\right)
  \\
  &\quad  \times\bigl(1+o(1)\bigr), \quad N \to \infty
\end{split}
\end{equation}
and 
\begin{equation}
  \label{eq:main}
    E_{N \, 2m}(z) \sim B_{N} \exp\left(\frac{z^2}{4}
-\frac{9}{2^{\frac{10}{3}}}\left(N^{\frac{2}{3}}z^{\frac{4}{3}}-1\right)\right)
\frac{N^{\frac{m}{3}}}{2^{\frac{5m}{3}}m!z^{\frac{4m}{3}}}, \quad N
\to \infty,
\end{equation}
% where $E_{N \, 2m}(z)$ are the coefficients of the power expansion
% \begin{equation}
%   \label{eq:p_ex_p_old}
%   E_{N}(z,t)= \sum_{m=0}^\infty E_{N\, 2m}(z)t^{2m}
% \end{equation}
where  $B_{N}$ is the ensemble average
\begin{equation}
B_{N}:= \frac{1}{\left(2\pi\right)^{N/2}\prod_{j=1}^N j!}
\int_{\mathbb{R}^N}\prod_{j=1}^N 
\exp\left(-\frac{1}{2Nx_j^2} - \frac{x_j^2}{2}\right) \prod_{1 \le j
< k \le N}\abs{x_k - x_j}^2 d^Nx.
\end{equation}
These limits, however, are not uniform in $z$: in the regime
where $zN^{\frac{1}{2}}$ remains finite they are determined by PIII and
change drastically. Since the integral~\eqref{eq:moments}  extends to
zero, the fact that the asymptotics are not uniform cannot be ignored
when computing the moments of $P(Q_N)$.

As already mentioned, recently ensembles whose weighting function have
essential singularities have appeared in several
applications.  Usually, such models depend on external parameters
beside the matrix dimension $N$; therefore, double scaling limits play
a fundamental role in their asymptotic behaviour.  Since the universal
properties of the spectrum and the associated Painlev\'e transcendents
are identified by the singular points of the equilibrium measure,
which in turn are determined by the analytical properties of $V(x)$,
we would expect the same critical behaviour whenever $V(x)$ has the
same poles as the potential appearing in $E_N(z,t)$. Therefore, the
results in this paper have implications for other matrix models,
beside its direct application to the theory of the Riemann zeta
function suggested by Berry and Shukla~\cite{BS08}, and provides a
pathway to tackle similar asymptotic problems.

As an example consider the weight introduced by Chen and
Its~\cite{CI10},
\begin{equation}
  \label{eq:its_Lag}
  W_{\alpha}(x) = x^\alpha e^{-x - s/x}, \quad 0\le x < \infty,
  \quad \alpha > -1, \quad s>0.
\end{equation}
This is a singular perturbation of the Laguerre Unitary Ensemble. For
the unperturbed weight, \textit{i.e.} when $s=0$, Forrester and
Witte~\cite{FW02,FW06} discovered that the generating function of the
probability that an interval contains $k$ eigenvalues in the hard edge
scaling limit can be evaluated in terms of a P$\mathrm{III}^{\prime}$
transcendent in $\sigma$-form.  Chen and Its~\cite{CI10} studied the
polynomials orthogonal with respect to the weight~\eqref{eq:its_Lag}
and showed that for finite $N$ the partition function can be written
as an integral involving PIII.  In unpublished
work~\cite{CI09}\footnote{We are grateful to Professors Chen and Its
  for making their manuscript available to us.}  they also
investigated the asymptotics of these orthogonal polynomials.  When
$t=0$ in Eq.~\eqref{eq:main_average} the system of monic polynomials
orthogonal with respect to the weight of the partition function
$E_N(z,0)$ can be mapped into that orthogonal with respect to $W_{\pm
  \frac{1}{2}} (x)$ by a change of variables --- the respective
partition functions, however, would still be different.  The pole of
order one in the exponent of the weighting function of $E_N(z,t)$ does
not contribute to the equilibrium measure, as it is asymptotically
negligible~\cite{MM09}.  Therefore, in the general setting we would
expect that the matrix model with weighting $W_\alpha(x)$ should
manifest the same critical behaviour discovered in this article.

% Although for general $\alpha$ the logarithmic
% singularity in the potential would make the details of the
% calculations more involved.

There exist other matrix models whose weighting function has essential
singularities and whose understanding could shed light on important
unsolved problems.  By setting $\alpha = 3n$ and $s=n$, the partition
function associated to $W_{3n}(x)$ becomes the moment generating
function of the probability density of the~\emph{Wigner delay
  time}~\cite{MS12b}.  This is the average time that an electron
spends when scattered by an open cavity and plays a fundamental role
in the theory of mesoscopic quantum dots. The distribution of the
Wigner delay time is far from being
understood~\cite{MS12b,TM13}. 

An other example of a similar problem is provided by the distribution
of the roots of the derivative of the characteristic polynomials of a
random unitary matrix. It has an integral representation that reduces
to a matrix average over the unitary group; the weighting function of
this average has essential singularities analogous to those in the
partition function $E_N(z,t)$~\cite{DFFHMP10,Mez03}.  This
distribution is very elusive; the calculation of more refined formulae
than those known at the moment would improve present results on the
percentage of the number of zeros of the Riemann zeta function on the
critical line.

\subsection{The Riemann-Hilbert Problem and the 
 Isomonodromic Deformations}
\label{se:RH_iso_def}
The asymptotics of the system of polynomials  orthogonal with respect
to the weighting function $\exp\left(-V_{z,t}(x)\right)$, where
\begin{equation}
\label{tzpotential}
V_{z,t}(x)=\frac{z^2}{2x^2} - \frac{t}{x} +\frac{x^2}{2},
\end{equation}
is characterized by two distinct regimes. It is well known that such
asymptotics is determined, after appropriate rescaling, by the
behaviour of the equilibrium measure, which is the solution of
a particular variational problem.  If $z=t=0$, then the equilibrium
measure is the semicircle law, which is supported in $[-2,2]$ and is
defined by
\begin{equation}
\label{eq:eqmeas1cut}
d\mu(y) := \frac{1}{2\pi}\sqrt{4 - y^2}\,dy, \qquad y := \frac{x}{\sqrt{N}}.
\end{equation}
In the limit $N \to \infty$ the simple pole in the potential
~\eqref{tzpotential} does not contribute to the equilibrium measure,
which depends only on $z$; in~\cite{MM09} we showed that in the
interval $c_1N^{-\frac12} < \abs{z} < c_2N^{\frac14}$, where $c_1$ and
$c_2$ are two positive constants, it is supported on two disjoint
intervals symmetric with respect to the origin. As $z \to 0$ the
quadratic singularity vanishes and the gap between the intervals
closes.  If this migration is fast enough, \textit{i.e.}  when
$zN^{1/2}$ remains finite as $N \to \infty$,
formulae~\eqref{eq:asymdet} and~\eqref{eq:main} cease to be valid and
a phase transition emerges. The main result of this paper is that in
this region of the phase space $(z,N)$ the asymptotics of $E_N(z,t)$
is expressed in terms of a solution of the PIII equation.

As we shall see in Sec.~\ref{se:RH}, for technical reasons we will use
the equilibrium measure~\eqref{eq:eqmeas1cut} and then deform the
relevant Riemann-Hilbert problem (RHP) by incorporating the essential
singularities in the jump matrix.

Let us introduce the scaling   
\begin{equation}
  \label{eq:u1u2}
   \left(u_{1,N},u_{2,N}\right) := (\sqrt{N}t,N z^2),
\end{equation}
with 
\begin{equation}
  \label{eq:limu1u2}
  (u_{1,N},u_{2,N}) = \left(u_1,u_2\right)\bigl(1 +
  O\left(1/N\right)\bigr), \quad N \to \infty,
\end{equation}
where $u_1\ge 0$ and $u_2 >0$ are finite. The rate of
convergence of $(u_{1,N},u_{2,N})$ is chosen so that it does not
affect the leading order asymptotics of $E_N(z,t)$ (see
Remark~\ref{re:error_term}).  Consider the rescaled weight
\begin{equation}
  \label{eq:rescaled_weight}
  w_N(y):=  \exp\Biggl(-N\left(\frac{u_{2,N}}{2N^3y^2} +
          \frac{y^2}{2}\right) +
    \frac{u_{1,N}}{Ny} \Biggr),
\end{equation}
where $y$ is rescaled as in~\eqref{eq:eqmeas1cut}.  Let $\pi_j(y)$
denote the monic polynomials orthogonal with respect to $w_N(y)$,
\textit{i.e.}
\begin{equation}
   \label{eq:orthogonality_2}
  \int_{-\infty}^\infty w_N(y) \pi_j(y)\pi_k(y)dy = h_j \delta_{jk},
  \quad j,k=1, 2,\dotsc,
\end{equation}
where the subscript in $\pi_j(y)$ refers to the degree of the
polynomial. This change of variables turns the
integral~\eqref{eq:main_average} into
\begin{equation}
  \label{eq:E_F}
  E_N(z,t) = N^{\frac{N^2}{2}}G_N(u_{1,N},u_{2,N}),
\end{equation}
where 
\begin{equation}
  \label{eq:tildeB}
  G_N(u_{1,N},u_{2,N})  := \frac{1}{N}\int_{\mathbb{R}^N}\prod_{j=1}^N
  w_N(y_j)\prod_{1\le j < k \le N}\abs{y_k  - y_j}^2d^Ny.
\end{equation}
Standard theory of orthogonal polynomials (see,
\textit{e.g.},~\cite[Chap.~2.2]{Sze39})  gives
\begin{equation}
  \label{eq:hankel_det_2}
  G_N(u_{1,N},u_{2,N}) = \det\left(\mu_{j + k}\right)_{j,k=0}^{N-1}
  = \prod_{j=0}^{N-1} h_j,
\end{equation}
where
\begin{equation}
  \label{eq:moments_2}
 \mu_j := \int_{-\infty}^\infty w_N(y) y^j dy.
  \quad j=0,1,\dotsc
\end{equation}

Our approach to the analysis of the
integral~\eqref{eq:main_average} is based on the steepest descent
method to compute the asymptotics of orthogonal polynomials introduced
by Deift and Zhou~\cite{DZ93} and further developed by Deift
\textit{et al.}~\cite{DKMVZ99b,DKMVZ99a} (see also~\cite{BI99} in
connection to the double scaling limit).  The starting point of this
technique is the characterization of the orthogonal polynomials in
terms of the solution of a RHP due
to Fokas \textit{et al.}~\cite{FIK91,FIK92}.  Define the matrix valued
function
\begin{equation}
 \label{eq:RH_sol}
   Y(y) :=
   \begin{pmatrix}
   \pi_N(y) & \frac{1}{2\pi i}
    \int_{-\infty}^\infty\frac{\pi_N(q)w_N(q)}{q-y}dq  \\
     \kappa_{N-1}\pi_{N-1}(y) &
     \frac{\kappa_{N-1}}{2\pi i}
      \int_{-\infty}^\infty\frac{\pi_{N-1}(q)w_N(q)}{q-y}dq
\end{pmatrix},
\end{equation}
where $\kappa_{N-1}=-2\pi i/h_{N-1}$. It solves the following RHP:
\begin{equation}
  \label{eq:RHP}
\begin{aligned}
1. \quad  & \text{$Y(y)$ is analytic in $\mathbb{C}/\mathbb{R}$},\\
2. \quad & Y_+(y)=Y_-(y)\begin{pmatrix} 1 & w_N(y)  \\ 0 & 1
\end{pmatrix}, \quad y\in\mathbb{R},\\
3. \quad  & Y(y)=\left(I+O(y^{-1})\right)\begin{pmatrix} y^N & 0
\\ 0 & y^{-N} \end{pmatrix}, \quad  y\rightarrow\infty,
\end{aligned}
\end{equation}
where $Y_+(y)$ and $Y_-(y)$ denotes the limiting value of $Y(y)$ as it
approaches the left- and right-hand side of the real axis.  The jump
matrix in the second condition is continuous, as the
weight~\eqref{eq:rescaled_weight} can be uniformly bounded on
$\mathbb{R}$; therefore, there is no need to specify any special
behaviour of $Y(y)$ near the origin in the definition~\eqref{eq:RHP}.

Bertola \textit{et al.}~\cite{BEH06} showed that the logarithmic
derivatives of $G_N(u_{1,N},u_{2,N})$ admit integral representations
involving $Y(y)$, which in our context can be phrased as
follows~\cite{MM09}:
\begin{lemma}[Bertola, Eynard and Harnad~\cite{BEH06}]
\label{thm:beh} 
The following differential identities
hold:
\begin{subequations}
\label{eq:diffid}
\begin{align}
\label{eq:diffidv1} \frac{\p \log G_N}{\p u_{1,N}}&=-\frac{1}{4\pi
iN}\oint_{y=0}\frac{1}{y}\tr\left(Y^{-1}(y)Y^{\prime}(y)\sigma_3\right)dy,\\
\label{eq:diffidv2} \frac{\p \log G_N}{\p u_{2,N}}&=\frac{1}{8\pi
iN^2}\oint_{y=0}\frac{1}{y^2}
\tr\left(Y^{-1}(y)Y^{\prime}(y)\sigma_3\right)dy,
\end{align}
\end{subequations}
where the contour of integration is a small loop around $y=0$
oriented counter-clockwise and $\sigma_3 = \left ( \begin{smallmatrix}
    1 & 0 \\ 0 & -1\end{smallmatrix} \right)$.
\end{lemma}
Note that, although the function $Y(y)$ has a jump discontinuity on
the real axis, it has a uniform asymptotic expansion near the origin,
which is used to compute the residues in the right hand sides of
\eqref{eq:diffid}.  Thus, the differential
identities~\eqref{eq:diffid} allow us to compute the asymptotics of
$G_N(u_{1,N},u_{2,N})$ in terms of that of $Y(y)$.

A standard technique in the Deift-Zhou steepest descent is the \textit{opening
  of the lens}. The $g$-function is used to construct a sequence of
transformations that turn~\eqref{eq:RHP} into a RHP which is amenable to an
asymptotic analysis.  As $N \to \infty$ the solution of the modified RHP outside
small neighbours of the critical points of the equilibrium measure can be
approximated in the same way as in~\cite[Chap.~7.3]{Dei99book}
and~\cite{DKMVZ99b}.  In our problem such critical points are the edges of
support of the equilibrium measure  and the origin, where the essential singularity appears.  These
aspects of the Riemann-Hilbert analysis are discussed in Sec~\ref{se:RH}.

The challenge for the weight~\eqref{eq:rescaled_weight} comes from the essential
singularity. Inside a small disc around the origin the asymptotic approximation
of the RHP breaks down; thus, it is replaced by an exactly solvable
``model'' RHP, whose solution $\Phi$, known as the local parametrix, matches the
asymptotic solution outside such a disc.  A widely used
technique in investigations of double scaling limits in RMT (see,
\textit{e.g.,}~\cite{BI03,CK06,CKV08,CV07}) is to express $\Phi$ in terms of a
special solution of a nonlinear differential equation, which in turn gives the
compatibility conditions for the following system of linear ODEs:
\begin{equation}
\label{eq:comp}
\p_{\zeta}\Phi(\zeta,u)=A(\zeta,u)\Phi(\zeta,u),
\quad\p_{u}\Phi(\zeta,u)=B(\zeta,u)\Phi(\zeta,u),
\end{equation}
where $\zeta \in \mathbb{C}$ and $u$ is a parameter measuring the
strength of the perturbation from the critical point. The matrices
$A(\zeta,u)$ and $B(\zeta,u)$ are rational functions of $\zeta$. In
our case the parameter $u$ is replaced by the vector $(u_1,u_2)$.
% where
% \begin{equation}
%   \label{eq:limu1u2_old}
%   u_1 = \lim_{N\to \infty} u_{1,N} \quad \text{and} \quad u_2 =
%   \lim_{N \to \infty} u_{2,N}.
% \end{equation}

The compatibility conditions for \eqref{eq:comp}
are among isomonodromic deformations which include
the Painlev\'e equations. In previous studies of
double scaling limits in one matrix models, one
parameter was sufficient to represent the whole
family of perturbations from the critical
point. In the RHP that we
study, however, the dependence of the local
parametrix on two perturbation parameters instead
of one increases the technical difficulties
considerably.  Indeed, the compatibility
conditions for the linear system analogous to
\eqref{eq:comp} give us a nonlinear PDE and not an
ODE. In Sec.~\ref{se:local} we use the Hamiltonian
theory of isomonodromic deformations~\cite{JMU81}
to express such a PDE in terms of a time-dependent
Hamiltonian system ODEs in the variables $u_{1}$
and $u_{2}$. This simplifies the
problem substantially from the computational point
of view.

The differential identities~\eqref{eq:diffid} give
a link between the partition function $E_N(z,t)$
and the RHP~\eqref{eq:RHP}; in
Secs.~\ref{se:final} and~\ref{se:asy_Hank} we
compute their asymptotics in terms of the solution
of the Hamiltonian system. This completes the proof of
Theorem~\ref{thm:main1}.  The boundary conditions
of the system of ODEs are provided by the
asymptotics of the local parametrix as $u_{2}\to
0$, which we compute in Sec.~\ref{se:asymu2}.

Sec.~\ref{se:red} concerns the reduction of the
Hamiltonian system to the PIII equation. In order
to achieve this goal we need to make two
observations.  Firstly, when $t=0$, and therefore
by Eqs.~\eqref{eq:u1u2} and~\eqref{eq:limu1u2}
$u_1=0$ too, the symmetries of the partition
function $E_N(t,z)$ and the initial conditions
computed in Sec.~\ref{se:asymu2} simplify the
system of ODEs, reducing it to a nonlinear second
order ODE.  The Taylor coefficients in
Eq.~\eqref{eq:p_ex_p}, and in turn the partition
function~\eqref{eq:main_average}, can be expressed
in terms of the solution of this nonlinear ODE.
Secondly, Chen and Its~\cite{CI10} showed that the
solution of RHP for the orthogonal polynomials
associated to the weighting
function~\eqref{eq:its_Lag} is described by the
PIII equation. When $\alpha=\pm 1/2$ in the
weight~\eqref{eq:its_Lag} and $t=0$ in the
potential~\eqref{tzpotential}, the two respective
systems of orthogonal polynomials are related by a
change of variables.  Thus, the corresponding RHPs
can be mapped into each other. These two facts and
a lengthy calculation allow us to identify the
second order ODE that we discovered with a scaling
limit of the PIII equation in~\cite{CI10}.  The
ODE associated to the problem in this paper
belongs to the PIII family too, but it is
different, as the parameters that define it are
not those of the ODE in~\cite{CI10}.

\section*{Acknowledgements}
We are very grateful to Professor Sir Michael Berry for first
suggesting us to study this problem and to Professor Alexander Its
for many invaluable discussions and suggestions on the material in this
article.

\section{Statement of Results}
\label{se:st_res}

The differential identities~\eqref{eq:diffid} are the starting point
to compute an asymptotic formula for the partition function
$E_N(z,t)$. When the perturbation parameters $z$ and $t$ are
$O(N^{-1/2})$, Eqs.~\eqref{eq:diffidv1} and~\eqref{eq:diffidv2} can be
expressed at leading order in terms of a particular solution of
non-linear Hamiltonian equations, which are isomonodromic deformations
of the compatibility conditions for the ODEs~\eqref{eq:comp}.  The
phase space of this Hamiltonian system is four-dimensional and the
time variables are the parameters $(u_1,u_2)$ defined in
Eq.~\eqref{eq:limu1u2}. Denote by $\mathbf{Z}(u_1,u_2) =
(P_1,Q_1,P_2,Q_2)$ such a solution.  Then, $\mathbf{Z}(u_1,u_2)$ is
determined uniquely by the monodromy data of the local parametrix at
the origin of the RHP~\eqref{eq:RHP}, which in turn gives the
behaviour of $\mathbf{Z}(u_1,u_2)$ as $u_2\rightarrow 0$, thereby
providing the initial conditions.  The asymptotic formulae of
$\mathbf{Z}(u_1,u_2)$ in the limit as $u_2 \to 0$ are rather
involved and are given by Theorem~\ref{thm:init}.

% The precise statement is
% given in Theorem~\eqref{thm:main1}. The solution
% of the Hamiltonian system is uniquely determined
% by the monodromy data of the local parametrix at
% the origin, whose asymptotics as $u_2 \to 0$
% provide the initial conditions. The explicit
% formulae of such initial conditions are rather
% involved and are given by Theorem~\ref{thm:init}.

\begin{theorem}
\label{thm:main1}
Let $u_1$ and $u_2$ be the parameters defined in
Eq.~\eqref{eq:limu1u2}. At leading order as $N\to \infty$ the
differential identities~\eqref{eq:diffidv1} and~\eqref{eq:diffidv2}
are
\begin{subequations}
\label{eq:logder}
\begin{align}
\label{eq:logder_1}
\frac{\p\log G_N}{\p u_{1,N}}&
= - H_1(\mathbf{Z},u_1,u_2)+O\left(N^{-1}\right),\\
\label{eq:logder_2}
\frac{\p\log G_N}{\p u_{2,N}}&=-H_2(\mathbf{Z},u_1,u_2) +O(N^{-1}),
\end{align}
\end{subequations}
uniformly for $u_1 \in E_1$ and $u_2 \in E_2$, where $E_1$ and $E_1$
are two closed interval in $\mathbb{R}_+$.  The vector
$\mathbf{Z}(u_1,u_2) = (P_1,Q_1,P_2,Q_2)$ is a particular solution of
the time-dependent Hamilton equations
\begin{equation}
\label{eq:hameq}
\frac{\p P_k}{\p u_j}=-\frac{\p \left(H_j+h_j\right)}{\p Q_k},
\qquad \frac{\p Q_k}{\p u_j}=\frac{\p \left(H_j+h_j\right)}%
{\p P_k}, \qquad j,k=1,2
\end{equation}
with Hamiltonian functions
\begin{subequations}
\label{eq:hamcan}
\begin{align}
H_1&=2i\frac{P_2}{u_2}+\frac{1}{2}u_1Q_2-\frac{1}{2}u_2Q_1Q_2
-\frac{1}{4u_2}u_1^2Q_1+\frac{1}{2}u_1Q_1^2
-\frac{1}{4}u_2Q_1^3 \notag\\
\label{eq:H1}
 & \quad +\frac{2}{u_2}P_1P_2Q_2+\frac{1}{u_2}P_1^2Q_1,\\
H_2&=-\frac{iP_1}{u_2}+\frac{iP_2Q_1}{u_2}+\frac{1}{8}
     u_2Q_1^2Q_2+\frac{1}{8}u_2Q_2^2-\frac{1}{2u_2}P_1^2Q_2\notag \\
&\quad -\frac{1}{2u_2}P_2^2Q_2^2-\frac{iP_2u_1}{u_2^2}
  +\frac{1}{8u_2^2}u_1^3Q_1-\frac{1}{4u_2}
u_1^2Q_1^2\notag \\
\label{eq:H2}
&\quad +\frac{1}{8}u_1Q_1^3-\frac{u_1}{u_2^2}P_1P_2Q_2
 -\frac{u_1}{2u_2^2}P_1^2Q_1-\frac{u_1^2Q_2}{8u_2},
 \end{align}
 as well as
 \begin{equation}
\label{eq:h1h2}
h_1=\frac{P_1}{u_2}, \qquad h_2=-\frac{P_2Q_2}{u_2}-\frac{P_1Q_1}{2u_2}-\frac{u_1}{2u_2^2}P_1.
\end{equation}
\end{subequations}
\end{theorem}

% \begin{remark}
%   In Sec.~\ref{ss:ptode}, we argue that by
%   replacing $(u_{1,N},u_{2,N})$ in the local parametrix with its
%   limit~\eqref{eq:limu1u2} we make an error of order $O(1/N^2)$, \ie Eq.~\eqref{eq:gen_th}.  Thus,
%   Proposition~\ref{pro:BHHP} implies that the rate of convergence of
%   $(u_{1,N},u_{2,N})$ does not affect the error terms
%   Eqs.~\eqref{eq:logder_1} and~\eqref{eq:logder_2}.
 
% \end{remark}

\begin{remark}
\label{re:error_term}
The definitions of $(u_{1,N},u_{2,N})$ and of their limits $(u_{1},u_{2})$ (see Eqs.~\eqref{eq:u1u2} and~\eqref{eq:limu1u2}) imply that the error terms in Eqs.~\eqref{eq:logder_1} and~\eqref{eq:logder_2} are not affected by the rate of convergence of $(u_{1,N},u_{2,N})$.  
\end{remark}

All the information on the partition function $E_N(z,t)$ is contained
in the coefficients of the generating function~\eqref{eq:p_ex_p},
which can be studied by looking at the projection
\begin{equation}
\label{eq:hameq2}
\frac{\p P_k}{\p u_2}=-\frac{\p \left(H_2+h_2\right)}{\p Q_k},
\qquad \frac{\p Q_k}{\p u_2}=\frac{\p \left(H_2+h_2\right)}%
{\p P_k}, \qquad k=1,2,
\end{equation}
at $u_1=0$. The expressions of the initial conditions of this
subsystem simplifies considerably  compared to the formulae in
Theorem~\ref{thm:init}.
\begin{corollary}
\label{co:asym_u2}
Let $u_1=0$.  As $u_2\to 0$  we have
\begin{subequations}
\label{eq:asympcan}
\begin{align}
\label{eq:asympcan_a}
P_1&=i\Biggl(\frac{\sqrt{u_2}}{\sqrt{2\pi}}+\frac{2u_2}{\pi}
+\left(\frac{2^{\frac{5}{2}}}{\pi^{\frac{3}{2}}}-\frac{\sqrt{2}}{\pi^{\frac{3}{2}}}
-\frac{3}{\sqrt{2\pi}}\right)u_2^{\frac{3}{2}}\notag \\
& \quad
+\left(\frac{2}{3}+\frac{8}{\pi^2}-\frac{32}{9\pi}\right)u_2^2\Bigr)
+O\left(u_2^{\frac{5}{2}}\right),\\
\label{eq:asympcan_b}
Q_1&=i\Biggl(\sqrt{\frac{2}{u_2\pi}}-2+\frac{4}{\pi}
+\left(\frac{2^{\frac{3}{2}}}{\sqrt{\pi}}+\frac{2^{\frac{7}{2}}}{\pi^{\frac{3}{2}}}
-\frac{3\sqrt{2}}{\sqrt{\pi}}-\frac{2^{\frac{3}{2}}}{\pi^{\frac{3}{2}}}
\right)u_2^{\frac{1}{2}} \notag \\
& \quad +\left(\frac{16}{\pi^2}-\frac{64}{9\pi}\right)u_2\Biggr)
+O\left(u_2^{\frac{3}{2}}\right),
\end{align}
\begin{align}
P_2&=O\left(u_2^3\right),\\
Q_2&=-\frac{2^{\frac{3}{2}}}{\sqrt{\pi
u_2}}+2-\frac{4}{\pi}+\left(\frac{3\cdot
2^{\frac{5}{2}}}{\pi^{\frac{3}{2}}}+\frac{7\cdot
2^{\frac{3}{2}}}{3\sqrt{\pi}}-\frac{2^{\frac{5}{2}}}{\sqrt{\pi}}
-\frac{2^{\frac{9}{2}}}{\pi^{\frac{3}{2}}}\right)u_2^{\frac{1}{2}}\notag \\
& \quad +\left(\frac{32}{9\pi}-\frac{8}{\pi^2}\right)u_2
+O\left(u_2^{\frac{3}{2}}\right).
\end{align}
\end{subequations}
\end{corollary}
The subsystem~\eqref{eq:hameq2} is equivalent to a fourth order ODE.
However, from the integral~\eqref{eq:tildeB} we see that the average
$G_N(u_{1,N},u_{2,N})$ is an even function in $u_{1,N}$, and hence the
derivatives $\p_{u_{1,N}}^{2k+1}\log G_N$ are zero when
$u_{1,N}=0$. Then, Eq.~\eqref{eq:logder_1} gives
\begin{equation}
\label{eq:reduction}
\p_{u_1}^{2k}H_1\bigr \rvert_{u_1=0}=0,\qquad k\in\{0\}\cup\mathbb{N}.
\end{equation} 
The first two equations in \eqref{eq:reduction} provide two possible
sets of relations between the variables $P_1$, $Q_1$, $P_2$ and $Q_2$,
namely
\begin{subequations}
\begin{equation}
\label{eq:red1}
P_2 =0, \qquad   Q_2 =-\frac{u_2^2Q_1^2-4P_1^2}{2u_2^2},
\end{equation}
or
\begin{equation}
\label{eq:red2}
P_2 =-\frac{iu_2^4Q_1^4-8iu_2^2P_1^2Q_1^2+16iP_1^4-8u_2^2P_1}{8u_2^2Q_1},
\qquad Q_2 =\frac{4iP_1}{u_2^2Q_1^2-4P_1^2}.
\end{equation}
\end{subequations}
The condition $\p_{u_1}^4H_1\bigr \rvert_{u_1=0}=0$ does not yield any
new constraints; when $k>2$ Eqs.~\eqref{eq:reduction}
involve high powers of $P_1$ and $Q_1$ and cannot be solved
analytically. Equation~\eqref{eq:red1} is compatible with the
behaviour of the coordinates $P_1$ and $Q_1$ as $u_2\to 0$ in
Corollary~\ref{co:asym_u2}, while Eq.~\eqref{eq:red2} is
not. Therefore, at $u_1=0$, the correct algebraic relations among the
canonical coordinates are~\eqref{eq:red1}. 

For convenience let us make the change of variable $r:=
\sqrt{u_2}$. Then,  Eqs.~\eqref{eq:red1} reduce the Hamilton
equations~\eqref{eq:hameq} in Theorems~\ref{thm:main1} to the system
\begin{subequations}
\label{eq:2ndorder}
\begin{align}
  \frac{d P_1}{d r} &= - \frac{\partial
    \left(H+h\right)}{\partial Q_1}=
  \frac{Q_1(r^4Q_1^2-4P_1^2)}{4r^5}+\frac{P_1}{r},\\
  \frac{ d Q_1}{d r} &= \frac{\partial
    \left(H+h\right)}{\partial P_1}=
  \frac{P_1(r^4Q_1^2-4P_1^2)}{r^5}-\frac{Q_1}{r}-\frac{2i}{r},
\end{align}
\end{subequations}
where 
\begin{equation*}
\begin{split}
H(P_1,Q_1,r)=-\frac{(r^4Q_1^2-4P_1^2)^2}{16r^5}
- \frac{2iP_1}{r}, \qquad h = - \frac{P_1Q_1}{r}.
\end{split}
\end{equation*}

The following theorem is the main result of this article.
\begin{theorem}
\label{thm:main2}
Consider the PIII equation
\begin{equation}
\label{eq:PIIIsp}
v_Y^{\prime\prime} =  \frac{(v_Y^\prime)^2}{v_Y} -
\frac{v_Y^\prime}{r} +(-1)^N\frac{v_Y^2}{r}-\frac{2}{r}+ v_Y^3,
\end{equation}
with initial conditions 
\begin{equation}
\label{eq:vasymeven}
v_Y(r)=\sqrt{\frac{\pi}{2}}
+\frac{\pi-4}{2}r+\frac{\pi^2-4\pi+4}{2^{\frac{3}{2}}
\sqrt{\pi}}r^2+O\left(r^2\right),\quad
r\rightarrow 0,
\end{equation}
when $N$ is even, and
\begin{equation}
\label{eq:vasymodd}
v_Y(r)=\frac{1}{r}+
\sqrt{\frac{2}{\pi}}-\frac{2\pi-6}{3\pi}r+O(r^2),\quad
r\rightarrow 0,
\end{equation}
when $N$ is odd. Then, the trajectory
\begin{subequations}
\label{eq:can_coor_mon}
\begin{align}
P_1& =  \frac{ir}{ 2v_Y(r)} +
\frac{ir^2v_Y^2(r)}{4} -(-1)^N \frac{ir^2 v_Y^\prime(r)}{4},  \\
Q_1& = \frac{(-1)^Ni}{rv_Y(r)} -(-1)^N
\frac{iv_Y^2(r)}{2} + \frac{iv_Y^\prime(r)}{2},
\end{align}
\end{subequations}
solves the system of ODEs~\eqref{eq:2ndorder}, where the solution
$v_{Y}(r)$ of~\eqref{eq:PIIIsp} is specified uniquely by the monodromy
data of the local parametrix at the origin of the RHP~\eqref{eq:RHP}.
Finally, there exist two closed intervals $E_1,E_2\in \mathbb{R}_+$,
such that for $u_1 \in E_1$ and $u_2 \in E_2$ the asymptotics of the
partition function~\eqref{eq:tildeB}  is given by
\begin{equation*}
  \ln \left[\frac{G_N(u_{1,N},u_{2,N})}{G_N(0,0)}\right]=
  F(u_1,u_2)\Bigl(1+O\bigl(N^{-1/2}\bigr)\Bigr), \quad N \to \infty,
\end{equation*}
where
\begin{equation}
\label{eq:exponent}
F(u_1,u_2) := - \int_{0}^{\sqrt{u_2}}H(r')dr' -\sum_{j=1}^{\infty}\frac{\p^{2j-1}
 H_1}{\p u_1^{2j-1}}\biggr\rvert_{u_1=0}\frac{u_1^{2j}}{(2j)!},
\end{equation}
and $H(r)$ is computed along the trajectory~\eqref{eq:can_coor_mon}.
\end{theorem}

\begin{remark}
Recall that PIII is the following second
order nonlinear ODE
\begin{equation}
\label{eq:PIII}
v^{\prime\prime}=\frac{\left(v^{\prime}\right)^2}{v}
-\frac{v^{\prime}}{r}+\frac{1}{r}\left(a
v^2+b\right)+c v^3+\frac{d}{v}.
\end{equation}
Hence, Eq.~\eqref{eq:PIIIsp} is a special case of~\eqref{eq:PIII} with
parameters $a=(-1)^{N}$, $b=-1$ $c=1$ and $d=0$.
\end{remark}

\begin{remark}
  Since $H_1(u_1,u_2)$ is an odd function of $u_1$, the series in the
  right-hand side of Eq.~\eqref{eq:exponent} can be expressed as
  \begin{equation}
  \label{obj}
    \sum_{j=1}^{\infty}\frac{\p^{2j-1}
      H_1}{\p u_1^{2j-1}}\biggr\rvert_{u_1=0}\frac{u_1^{2j}}{(2j)!}=     
    \int_0^{u_1} H_1(u_1',u_2)du_1', 
  \end{equation}
  as it is also apparent from Eq.~\eqref{eq:logder_1}.  The series
  expansion, however, provides a direct way of evaluating the integral in the right-hand side of~\eqref{obj},
  as the derivatives $\partial_{u_1}^{2j-1} H_1 \bigr \rvert_{u_1=0}$
  can be computed explicitly by combining the Hamilton
  equations~\eqref{eq:hameq} with the algebraic relations
  \eqref{eq:red1}. For example, we have
\begin{equation*}
\begin{split}
\frac{\p H_1}{\p
u_1}\biggr \rvert_{u_1=0}&=\frac{P_1^2}{u_2^2}-\frac{Q_1^2}{4},\\
\frac{\p^3 H_1}{\p
u_1^3}\biggr \rvert_{u_1=0}&=\left(\frac{3}{u_2^2}
-\frac{1}{8u_2^4}\right)\left(u_2^2Q_1^2-4P_1^2\right)^2+\frac{2iP_1}{u_2^2}.
\end{split}
\end{equation*}
% This then gives the leading order asymptotics of the
% average~\eqref{eq:main_average}.rProf.
\end{remark}

\begin{remark}
  A priori the boundary conditions~\eqref{eq:vasymeven}
  and~\eqref{eq:vasymodd} may not identify the solution of
  Eq.~\eqref{eq:PIIIsp} uniquely.  However, the PIII transcendent
  $v_Y(r)$ that enters into the formulae for the canonical
  coordinates~\eqref{eq:can_coor_mon} is specified unambiguously by
  the parametrix $\hat{\Psi}_0(\zeta)$ introduced in
  Eq.~\eqref{Yasym0}.  The connection between $v_Y(r)$ and
  $\hat{\Psi}_0(\zeta)$ is discussed in Sec.~\ref{sse:rel_PIII}.
	%the proof of Proposition~\ref{pro:PIII}, more specifically in
  
\end{remark}

\section{Riemann-Hilbert Analysis}
\label{se:RH}
In this section we apply the Deift-Zhou steepest descent
analysis to the RHP (\ref{eq:RHP}).

\subsection{First Transformation of the RHP}
Define the $g$-function
\begin{equation}\label{eq:1cutg}
g(y):=\int_{\mathbb{R}}\log(y-q)d\mu(q),
\end{equation}
where $d\mu(q)$ is the equilibrium measure for the potential
$V(y)= y^2/2$, defined in Eq.~\eqref{eq:eqmeas1cut}.
% \begin{equation}
% \label{eq:eqmeas1cut}
% d\mu(q):=\frac{1}{2\pi}\sqrt{4-q^2}\,dq.
% \end{equation}
It satisfies the constraints
\begin{subequations}
\label{eq:ineq}
\begin{align}
&2\int_{-\infty}^\infty\log \left \lvert y-q \right \rvert d\mu(q)-\frac{y^2}{2}= l, \quad
y\in[-2,2],\\
&2\int_{-\infty}^\infty\log \left \lvert y-q\right \rvert d\mu(q)-\frac{y^2}{2}< l, \quad
 y\in\mathbb{R}\setminus[-2,2],
\end{align}
\end{subequations}
for some constant $l$. 

Define
\begin{subequations}
\label{eq:Tg}
\begin{align}
\label{eq:Tg1}
T(y)&:=e^{\frac{-Nl\sigma_3}{2}}Y(y)e^{-Ng(y)\sigma_3}e^{\frac{Nl\sigma_3}{2}},\\
\tilde{g}(y)& :=\frac{y^2}{4}-g(y)+\frac{l}{2},\\
F(y)&:=-\frac{1}{2}\left(\frac{u_{2,N}}{2N^2y^2}-\frac{u_{1,N}}{Ny}\right),
\end{align}
\end{subequations}
where $Y(y)$ is given in Eq.~\eqref{eq:RH_sol}. Then,  $T(y)$ solves
the RHP
\begin{equation*}
\begin{split}
&1. \quad \text{$T(y)$ is analytic in $\mathbb{C}\setminus\mathbb{R}$},\\
&2. \quad T_+(y)=T_-(y)J_T(y),\quad y\in\mathbb{R},\\
&3. \quad T(y)=I+O(y^{-1}),\quad y\rightarrow\infty.
\end{split}
\end{equation*}
The jump matrix $J_T(y)$ is given by
\begin{equation*}
J_T(y):=\begin{pmatrix} 
e^{-N\left(g_+(y)-g_-(y)\right)} &
e^{-N\left(\tilde{g}_+(y)+\tilde{g}_
-(y)\right)+2F(y)}  \\
0 &e^{N(g_+(y)-g_-(y))}
\end{pmatrix},\quad y\in\mathbb{R}.
\end{equation*}
Note that since the essential singularities can be controlled on the
real line, the jump matrix $J_T(y)$ is continuous on $\mathbb{R}$.
\subsection{Opening of the Lenses}
\label{se:lens} 
We now perform a standard technique in the steepest
decent method (see, \textit{e.g.},~\cite{BI99,DKMVZ99b,DKMVZ99a}). 
\begin{figure}[h]
\centering
\includegraphics[width=5.5in]{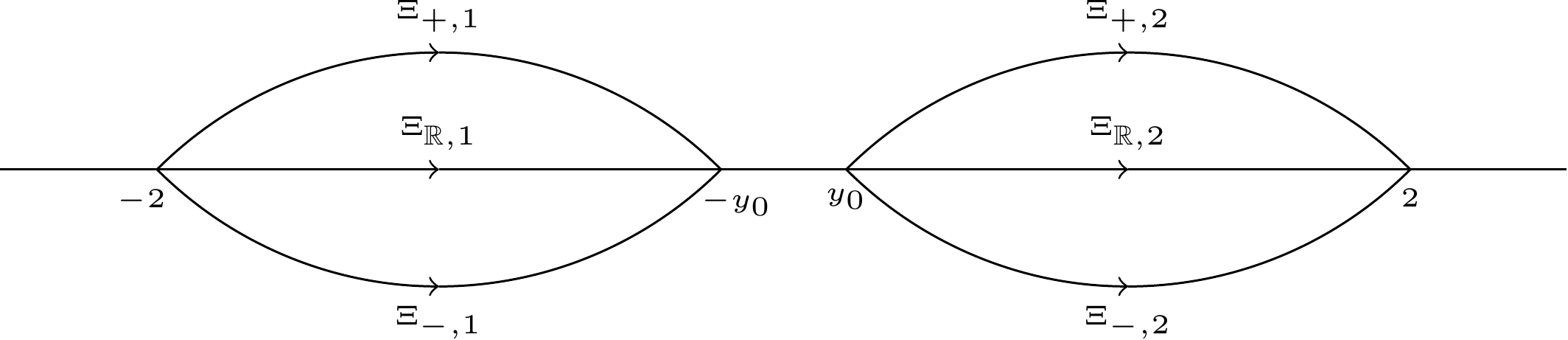}
\caption{The opening of the lenses in the interval $[-2,2]$.}
\label{fig:lens}
\end{figure}
Take a point $y_0$ sufficiently close to the
origin and define the lens contours as in
Fig.~\ref{fig:lens}.  The precise location of
$y_0$ will be specified in
Sec.~\ref{sse:conformal_map}, where we study the
local parametrix near zero.  The boundaries of the
lens regions $L_{\pm,j}$ are $\p L_{\pm,j}
=\Xi_{\pm,j} \cup \Xi_{\mathbb{R},j}$ and we denote by
$\Xi$ the union of the lens contours in
Fig.~\ref{fig:lens}. Write
\begin{equation}
\label{eq:S}
S(y):= \begin{cases}
         T(y), & \text{$y$ outside the lenses,} \\
         T(y)\begin{pmatrix} 1 & 0  \\
-e^{2N\tilde{g}(y)-2F(y)} &1
\end{pmatrix}, & \text{$y\in L_{+,j}, \quad j=1,2,$} \\
         T(y)\begin{pmatrix} 1 & 0  \\
e^{2N\tilde{g}(y)-2F(y)} &1
\end{pmatrix}, & \text{$y\in L_{-,j}, \quad j=1,2$.}
       \end{cases}
\end{equation}
 Then, $S(y)$ satisfies the RHP
\begin{equation}\label{eq:RHS}
\begin{split}
&1. \quad \text{$S(y)$ is analytic in
  $\mathbb{C}\setminus\left(\mathbb{R} \cup \Xi\right)$},\\
&2. \quad S_+(y)=S_-(y)J_S(y),\quad y\in
\mathbb{R} \cup \Xi ,\\
&3. \quad S(y)=I+O(y^{-1}),\quad y\rightarrow\infty.
\end{split}
\end{equation}
The jump matrix $J_S(y)$ is defined piecewise: 
\begin{equation}
\label{eq:Sjump1}
J_{S}(y):=\begin{cases}
\begin{pmatrix} 1 & 0  \\
e^{2N\tilde{g}(y)-2F(y)} &1
\end{pmatrix}, & y\in \Xi_{\pm,j},\quad j=1,2,\\
\begin{pmatrix}0&e^{-\frac{u_{2,N}}{2N^2y^2}+\frac{u_{1,N}}{Ny}}\\
-e^{\frac{u_{2,N}}{2N^2y^2}-\frac{u_{1,N}}{Ny}}&0\end{pmatrix},
& y\in \Xi_{\mathbb{R},1}\cup \Xi_{\mathbb{R},2},\\
\begin{pmatrix} e^{N(\tilde{g}_+(y)-\tilde{g}_-(y))} & e^{-2N\tilde{g}(y)+2F(y)}  \\
0 &e^{-N(\tilde{g}_+(y)-\tilde{g}_-(y))}\end{pmatrix}, & y\in
\mathbb{R}\setminus\left(\Xi_{\mathbb{R},1} \cup \Xi_{\mathbb{R},2}\right).
\end{cases}
\end{equation}

Equation~\eqref{eq:Sjump1} combined with the inequalities
\eqref{eq:ineq} imply that outside of some discs $D_{\pm 2}$ and $D_0$
of sufficiently small radius $\delta$ centred at the points $\pm 2$
and zero, $J_S(y)$ can be approximated by
\begin{equation}
\label{eq:jump_S}
J_S(y)=\begin{pmatrix}0&e^{-\frac{u_{2,N}}{2N^2y^2}+\frac{u_{1,N}}{Ny}}\\
-e^{\frac{u_{2,N}}{2N^2y^2}-\frac{u_{1,N}}{Ny}}&0\end{pmatrix}
=\Bigl(I+O\left(N^{-1}\right)\Bigr)
\begin{pmatrix}0&1\\
-1&0\end{pmatrix}
\end{equation}
for $y \in [-2,2]$ and the identity on the rest of the contour. This
suggests the following approximation to $S(y)$ outside of $D_0$ and
$D_{\pm 2}$:
\begin{equation}\label{eq:sinf}
\begin{split}
&1. \quad \text{$S^{\infty}(y)$ is analytic in $\mathbb{C}\setminus[-2,2]$};\\
&2. \quad S_+^{\infty}(y)=S_-^{\infty}(y)\begin{pmatrix}0&1\\
-1&0\end{pmatrix},\quad y\in[-2,2];\\
&3. \quad S^{\infty}(y)=I+O(y^{-1}),\quad y\rightarrow\infty.
\end{split}
\end{equation}
The outer parametrix that solves \eqref{eq:sinf} can be constructed as
in \cite[Chap.~7.6]{Dei99book} and~\cite{DKMVZ99b}:
\begin{equation}
\label{eq:para}
S^{\infty}(y)=\begin{pmatrix}\frac{\gamma+\gamma^{-1}}{2}
&\frac{\gamma-\gamma^{-1}}{2i}\\
-\frac{\gamma-\gamma^{-1}}{2i}&
\frac{\gamma+\gamma^{-1}}{2}\end{pmatrix},\quad\gamma=\left(\frac{y-2}{y+2}\right)^{\frac{1}{4}},
\end{equation}
where the branch cut is chosen to be on $[-2,2]$ and $\gamma\sim 1$
as $y\rightarrow\infty$.

\subsection{Local Parametrices Near the Points $\pm 2$}
The approximation of $S(y)$ by
$S^{\infty}(y)$ fails near the points $\pm 2$. Therefore, we must find exact solutions to the
RHP for $S(y)$ and match them with $S^{\infty}(y)$ as $y$ moves away
from the edges of the interval $[-2,2]$. These local parametrices can
be constructed using Airy functions as in~\cite[Chap.~7.6]{Dei99book}
and in~\cite{BI99,DKMVZ99b,DKMVZ99a} and we shall not repeat the
derivation here.

\section{Local Parametrix Near the Origin}
\label{se:local}
The essential singularity in the weight~\eqref{eq:rescaled_weight}
means that the approximation of $S(y)$ by $S^{\infty}(y)$ breaks down
near the origin too, as becomes apparent from Eq.~\eqref{eq:jump_S}.
Therefore, we must solve the RHP for $S(y)$ in a neighbourhood of the
origin exactly and match it with $S^{\infty}(y)$ away from the
singularity.

Take a small neighborhood of the origin $D_0$ containing the interval
$(-y_0,y_0)$. We
want to solve the following RHP inside $D_0$:
\begin{equation}
\label{eq:localpara0}
\begin{split}
1.\quad &\text{$S^{(0)}(y)$ is analytic in
$D_0\setminus \left(D_{0}\cap\Xi\right)$},\\
2.\quad &S_+^{(0)}(y)=S_-^{(0)}(y)J_S(y),\qquad y\in D_{0}\cap \Xi,\\
3.\quad &S^{(0)}(y)=\left(I+O(N^{-1})\right)S^{\infty}(y),\qquad
z\in\p D_{0},
\end{split}
\end{equation}
where the jump matrix is 
\begin{equation}
\label{eq:Sjump0}
J_S(y):=\begin{cases}
\begin{pmatrix} 1 & 0  \\
e^{2N\tilde{g}(y)+\frac{u_{2,N}}{2N^2y^2}-\frac{u_{1,N}}{Ny}} &1
\end{pmatrix},& y\in \Xi_{\pm,j}\cap D_0,\quad j=1,2,\\
\begin{pmatrix}0&e^{-\frac{u_{2,N}}{2N^2y^2}+\frac{u_{1,N}}{Ny}}\\
-e^{\frac{u_{2,N}}{2N^2y^2}-\frac{u_{1,N}}{Ny}}&0\end{pmatrix},
& y\in (-\delta,-y_0)\cup(y_0,\delta),\\
\begin{pmatrix} e^{2N\tilde{g}_+(y)} 
& e^{-\frac{u_{2,N}}{2N^2y^2}+\frac{u_{1,N}}{Ny}}  \\
0 &e^{-2N\tilde{g}_+(y)}\end{pmatrix}, & y\in
(-y_0,y_0).
\end{cases}
\end{equation}
\subsection{Conformal Maps Inside $D_0$}
\label{sse:conformal_map}
In order to study the model
problem~\eqref{eq:localpara0}, we 
introduce a conformal map $\zeta$ from $D_0$ into
an open neighbourhood of the origin that maps the
lens contours $\Xi$ as in Fig.~\ref{fig:conf} and
sends $\p D_0$ to $\zeta=\infty$ as $N \to \infty$.  

Consider the map
\begin{equation}
\label{eq:zeta0}
y \mapsto  -iN\bigl(\tilde{g}_+(y)-\tilde{g}_+(0)\bigr)=N\left(y+O(y^3)\right),
\quad y \in \mathbb{R} \cap D_0.
\end{equation}
Note that $\tilde{g}_+(y)\in i\mathbb{R}$ for $y\in \mathbb{R} \cap
D_0 $. For small but finite $\delta$ we denote by $\zeta$ the analytic
continuation of~\eqref{eq:zeta0} to the complex plane for all $y \in
D_0$. In the limit $N \to \infty$,  $\zeta$ maps the disc $D_0$ to the
whole complex plane.  Equation~\eqref{eq:zeta0} implies that, in order
for the images $\zeta(\pm y_0)$ to remain at a finite distance from
the essential singularity as $N \to \infty$, we need to choose $y_0 =
O(1/N)$.
\begin{figure}[h]
\centering 
\includegraphics[width=6in]{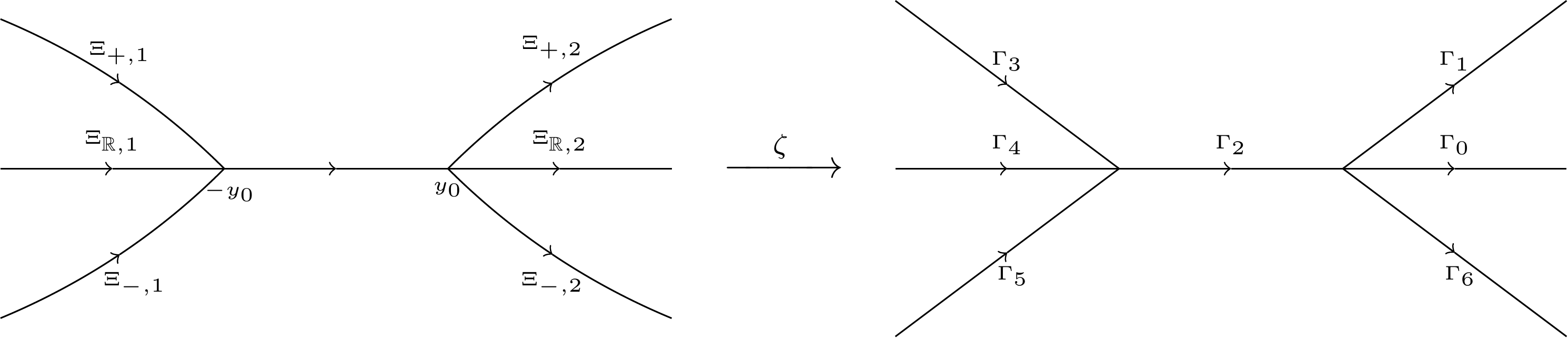}
\caption{The mapping of the contours in the disc $D_0$ under the conformal map
$\zeta$ in the limit $N \to \infty$}\label{fig:conf}
\end{figure}

Define $U_1(y)$ and $U_2(y)$ in $D_0$ as
\begin{equation}
\label{eq:etavarsigma}
U_1(y):=u_{1,N}\frac{\zeta}{Ny}=u_{1,N}\bigl(1+O(y^2)\bigr), \qquad
U_2(y):=u_{2,N}\frac{\zeta^2}{N^{2}y^2}=u_{2,N}\bigl(1+O(y^2)\bigr).
\end{equation}
For small enough $\delta$, $U_1(y)$ and $U_2(y)$ are conformal inside
$D_0$ and 
\begin{equation}\label{eq:potzeta}
2N\tilde{g}(y)+\frac{u_{2,N}}{2N^2y^2}-\frac{u_{1,N}}{Ny}=\pm
2i\zeta+\frac{U_2(y)}{2\zeta^2}-\frac{U_1(y)}{\zeta}\pm2N\tilde{g}_+(0),\quad
\pm \ipart(\zeta)>0.
\end{equation}

Let $\Gamma$ be the union of the contours in the right-hand side of
Fig.~\ref{fig:conf}.  As $N\to \infty$ let us introduce the following
RHP in the $\zeta$-plane:
\begin{equation}
\label{eq:localzeta}
\begin{split}
&1.\quad \text{$P(\zeta)$ is analytic in
$\mathbb{C}\setminus\Gamma$},\\
&2.\quad P_+(\zeta)=P_-(\zeta)J_P(\zeta),\qquad \zeta\in \Gamma,\\
&3.\quad P(\zeta)=\left(I+O(\zeta^{-1})\right)\begin{pmatrix}0&1\\
-1&0\end{pmatrix},\qquad
\ipart(\zeta)>0,\quad \zeta\rightarrow\infty, \\
&4. \quad P(\zeta)=I+O(\zeta^{-1}),\qquad \ipart(\zeta)<0,\quad
\zeta\rightarrow\infty.
\end{split}
\end{equation}
The jump matrices are 
\begin{equation}
\label{eq:JP}
J_P(\zeta) :=
\begin{cases}
\begin{pmatrix} 1 & 0  \\
e^{2i\zeta+\frac{U_2(y)}{2\zeta^2}-\frac{U_1(y)}{\zeta}} &1
\end{pmatrix}, &  \zeta\in \Gamma_1\cup\Gamma_3,\\
\begin{pmatrix} 1 & 0  \\
e^{-2i\zeta+\frac{U_2(y)}{2\zeta^2}-\frac{U_1(y)}{\zeta}}
&1\end{pmatrix}, & \zeta\in\Gamma_5\cup\Gamma_6,\\
\begin{pmatrix}0&e^{-\frac{U_2(y)}{2\zeta^2}+\frac{U_1(y)}{\zeta}}\\
-e^{\frac{U_2(y)}{2\zeta^2}-\frac{U_1(y)}{\zeta}}&0\end{pmatrix}, &
\zeta\in \Gamma_0\cup\Gamma_4,\\
\begin{pmatrix} e^{2i\zeta} & e^{-\frac{U_2(y)}{2\zeta^2}+\frac{U_1(y)}{\zeta}}  \\
0 &e^{-2i\zeta}\end{pmatrix}, & \zeta\in \Gamma_2.
\end{cases}
\end{equation}

The local parametrix that solves the RHP~\eqref{eq:localpara0} can be
obtained from the solution of \eqref{eq:localzeta} through the relations
\begin{subequations}
\label{eq:S0P}
\begin{align}
S^{(0)}(y)&=S^{\infty}(y)e^{-N\tilde{g}_+(0)\sigma_3}\begin{pmatrix}0&-1\\
1&0\end{pmatrix}P(\zeta)e^{N\tilde{g}_+(0)\sigma_3},\quad
\ipart(y)>0,\\
S^{(0)}(y)&=S^{\infty}(y)e^{N\tilde{g}_+(0)\sigma_3}
P(\zeta)e^{-N\tilde{g}_+(0)\sigma_3},\quad\ipart(y)<0,
\end{align}
\end{subequations}
with $N\tilde{g}_+(0)\in i\mathbb{R}$.

\subsection{Existence of the Local Parametrix}
Before proceeding we need to prove the following existence theorem.
\begin{theorem}
The solution to the RHP~\eqref{eq:localzeta} exists and is unique. 
\end{theorem}
\begin{proof}
Define
\begin{equation*}
\hat{P}(\zeta): =\begin{cases}
                   P(\zeta)\begin{pmatrix}0&-1\\1&0\end{pmatrix}, & 
                  \ipart(\zeta)>0,\\
                   P(\zeta), & \ipart(\zeta)<0.
                 \end{cases}
\end{equation*}
Then, $\hat{P}(\zeta)$ solves the RHP
\begin{equation}\label{eq:hatP}
\begin{split}
&1.\quad \text{$\hat{P}(\zeta)$ is analytic in
$\mathbb{C}\setminus\Gamma$},\\
&2.\quad \hat{P}_+(\zeta)=\hat{P}_-(\zeta)J_{\hat{P}}(\zeta),\quad \zeta\in \Gamma,\\
&3.\quad \hat{P}(\zeta)=I+O(\zeta^{-1}),\quad\zeta\rightarrow\infty,
\end{split}
\end{equation}
where the jump matrices $J_{\hat{P}}(\zeta)$ are given by
\begin{equation}
\label{eq:JhatP}
J_{\hat{P}}(\zeta):=
\begin{cases}
\begin{pmatrix} 1 & -e^{2i\zeta+\frac{U_2(y)}{2\zeta^2}-\frac{U_1(y)}{\zeta}}  \\
0 &1
\end{pmatrix}, &  \zeta\in \Gamma_1\cup\Gamma_3,\\
\begin{pmatrix} 1 & 0  \\
e^{-2i\zeta+\frac{U_2(y)}{2\zeta^2}-\frac{U_1(y)}{\zeta}}
&1\end{pmatrix}, &\zeta\in\Gamma_5\cup\Gamma_6,\\
\begin{pmatrix}e^{-\frac{U_2(y)}{2\zeta^2}+\frac{U_1(y)}{\zeta}}&0\\
0&e^{\frac{U_2(y)}{2\zeta^2}-\frac{U_1(y)}{\zeta}}\end{pmatrix},
& \zeta\in \Gamma_0\cup\Gamma_4,\\
\begin{pmatrix} e^{-\frac{U_2(y)}{2\zeta^2}+\frac{U_1(y)}{\zeta}} & -e^{2i\zeta}  \\
e^{-2i\zeta} &0\end{pmatrix}, & \zeta\in \Gamma_2.
\end{cases}
\end{equation}

Let $J_{\hat P,k}(\zeta)$ denote the jump matrix on $\Gamma_k$. In
order to prove the existence and uniqueness of $\hat P(\zeta)$ we
refer to the theory developed by Zhou~\cite{Zho89}, which for our
purposes is equivalent to the following statements:
\begin{enumerate}
\item the RHP~\eqref{eq:hatP} and the jump matrices~\eqref{eq:JhatP}
  do not have singularities at the self-intersection points $\pm y_0$
  of the contours;
\item the jump matrices $J_{\hat P,k}(\zeta)$ satisfy the cyclic
  conditions
\begin{equation*}
J_{\hat{P},2}^{-1}J_{\hat{P},6}J_{\hat{P},0}J_{\hat{P},1}=I\quad
\text{and} \quad 
J_{\hat{P},2}^{-1}J_{\hat{P},5}J_{\hat{P},4}J_{\hat{P},3}=I
\end{equation*}
at the points $\pm y_0$;
\item there is no non-trivial function $\hat{P}_0(\zeta)$ that
  satisfies properties 1. and 2. of the RHP~\eqref{eq:hatP} and
  behaves as $O(\zeta^{-1})$ as $\zeta\rightarrow\infty$.
\end{enumerate}

\begin{figure}[h]
\centering
\includegraphics[width=4.0in]{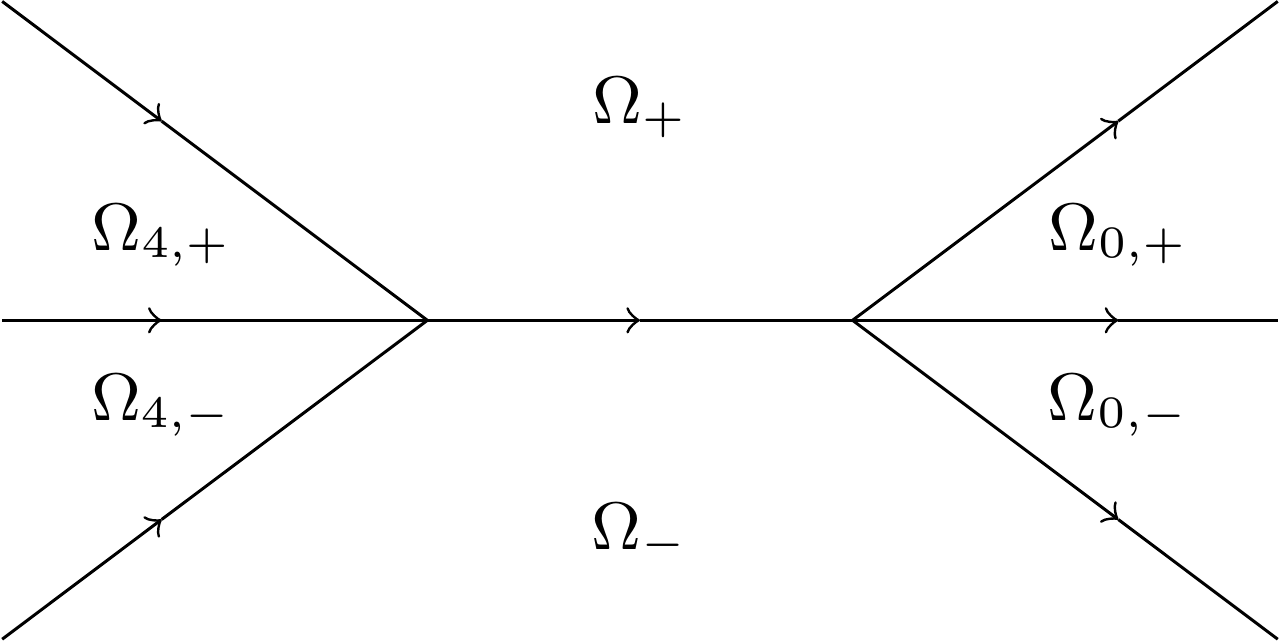}
\caption{The regions in the $\zeta$-plane}
\label{fig:zp_reg}
\end{figure}

The conditions a. and b. can be easily verified by direct
inspection. In order to prove c., suppose that such a $\hat
P_0(\zeta)$ exists and denote the region between $\Gamma_1$ ($\Gamma_3$)
and $\Gamma_0$ ($\Gamma_4$) by $\Omega_{0,+}$ ($\Omega_{4,+}$) and the
region between $\Gamma_0$ ($\Gamma_4$) and $\Gamma_6$ ($\Gamma_5$) by
$\Omega_{0,-}$ ($\Omega_{4,-}$) (see Fig.~\ref{fig:zp_reg}).  Then, define 
\begin{equation*}
X(\zeta) := 
\begin{cases}
\hat{P}_0(\zeta),&
\zeta\in\mathbb{C}\setminus\left(\Omega_{+} \cup \Omega_{-}\right),\\
\hat{P}_0(\zeta)\begin{pmatrix} 1 &
  -e^{2i\zeta+\frac{U_2(y)}{2\zeta^2}
-\frac{U_1(y)}{\zeta}}  \\
0 &1 \end{pmatrix}, & \zeta\in\Omega_{0,+}\cup\Omega_{4,+},\\
\hat{P}_0(\zeta)\begin{pmatrix} 1 & 0  \\
-e^{-2i\zeta+\frac{U_2(y)}{2\zeta^2}-\frac{U_1(y)}{\zeta}}
&1\end{pmatrix},& \zeta\in\Omega_{0,-}\cup\Omega_{4,-}.
\end{cases}
\end{equation*}
The matrix $X(\zeta)$ has the following properties~\cite{DKMVZ99b}:
\begin{equation}
\label{eq:Xi}
\begin{split}
&1.\quad \text{$X(\zeta)$ is analytic in
$\mathbb{C}\setminus\mathbb{R}$ and continuous down to $\mathbb{R}$},\\
&2.\quad X_+(\zeta)=X_-(\zeta)\begin{pmatrix} e^{-\frac{U_2(y)}{2\zeta^2}+\frac{U_1(y)}{\zeta}} & -e^{2i\zeta}  \\
e^{-2i\zeta} &0\end{pmatrix},\quad \zeta\in \mathbb{R},\\
&3.\quad X(\zeta)=O(\zeta^{-1}),\quad\zeta\rightarrow\infty.
\end{split}
\end{equation}

Let $\mathcal{C}$ be a close contour in the upper-half plane
consisting of the real line and a large semicircle
$\mathcal{S}$. Consider the integral
\begin{equation}
\label{eq:cauchy}
  \oint_{\mathcal{C}} X(\zeta) \overline{X(\overline{\zeta})}d\zeta = 
\int_{\mathcal{S}} X(\zeta)\overline{X(\overline{\zeta})} d\zeta
+ \int_{\mathbb{R}} X_{+}(\zeta)\overline{X_{-} (\overline{\zeta})}d\zeta. 
\end{equation}
Since $X(\zeta)\overline{X(\overline{\zeta})}$ is analytic in
$\mathbb{C}_+$ and of order $O(\zeta^{-2})$ as
$\zeta\rightarrow\infty$, we have 
\[
\int_{\mathcal{S}} X(\zeta)\overline{X(\overline{\zeta})} d\zeta
\rightarrow 0. 
\]
The second condition in Eq.~\eqref{eq:Xi} gives
\begin{equation}
\label{intcauchy}
\int_{\mathbb{R}}X_+(\zeta)\overline{X_-(\overline{\zeta})}d\zeta=\int_{\mathbb{R}}X_-(\zeta)\begin{pmatrix} e^{-\frac{U_2(y)}{2\zeta^2}+\frac{U_1(y)}{\zeta}} & -e^{2i\zeta}  \\
e^{-2i\zeta}
&0\end{pmatrix}\overline{X_-(\overline{\zeta})}d\zeta.
\end{equation} 
The essential singularity in the jump matrix does not affect this
integral, since the exponential 
\[
\exp\left(
  - \frac{U_2(y)}{2\zeta^2}+ \frac{U_1(y)}{\zeta}\right)
\] 
can be uniformly bounded in a small interval on the real line
containing the origin; therefore, Cauchy's theorem implies that
integral~\eqref{intcauchy} is zero.
% \begin{equation*}
% \int_{\mathbb{R}}X_-(\zeta)\begin{pmatrix} e^{-\frac{U_2(y)}{2\zeta^2}+\frac{U_1(y)}{\zeta}} & -e^{2i\zeta}  \\
% e^{-2i\zeta}
% &0\end{pmatrix}\overline{X_-(\overline{\zeta})}d\zeta=0.
% \end{equation*}
By adding this matrix  to its Hermitian conjugate, we see that
\begin{equation*}
\int_{\mathbb{R}}X_-(\zeta)\begin{pmatrix} 2e^{-\frac{U_2(y)}{2\zeta^2}+\frac{U_1(y)}{\zeta}} & 0  \\
0 &0\end{pmatrix}\overline{X_-(\overline{\zeta})}d\zeta=0.
\end{equation*}
This implies that the first column of $X_-(\zeta)$ is identically
zero. From the jump conditions in (\ref{eq:Xi}), it follows that the
second column of $X_+(\zeta)$ is identically zero too. 

Now, let the first column of $X$ be $X_{1}$ and the second column of
$X$ be $X_{2}$ and define the vector 
\begin{equation*}
\begin{split}
f(\zeta):=\left\{
           \begin{array}{ll}
             X_{2}(\zeta), & \hbox{$\ipart(\zeta)<0$,} \\
             X_{1}e^{2i\zeta}, & \hbox{$\ipart(\zeta)>0$.}
           \end{array}
         \right.
\end{split}
\end{equation*}
From Eq.~\eqref{eq:Xi} we see that $f(\zeta)$ is analytic in the
whole complex plane and behaves as $\zeta^{-1}$ as
$\zeta\rightarrow\infty$. Hence, by Liouville's theorem, we have
$f(\zeta)=0$. This shows that there is no non-trivial function
$\hat{P}_0$ that satisfies 1. and 2. in~ \eqref{eq:hatP} and such that
$\hat{P}_0=O(\zeta^{-1})$ as $\zeta\rightarrow\infty$. Therefore,
the RHP~\eqref{eq:localzeta} is uniquely solvable.
\end{proof}

\subsection{Painlev\'e Type Differential Equations}
\label{ss:ptode}
We will now transform the RHP~\eqref{eq:localzeta} into one with
constant jumps but with essential singularities at $\zeta=0$ and
$\zeta=\infty$. Then, theory of isomonodromy deformations developed by
Jimbo, Miwa, and Ueno~\cite{JMU81} can be applied to derive differential
equations that give the solution to \eqref{eq:localzeta}.

Let us deform the jump contours of \eqref{eq:localzeta} as in
Fig.~\ref{fig:local2} and let $\Gamma_{\pm}$ be the semicircles in the uppe/lower half planes.  Note
that $\partial \Omega_{2,\pm} = \Gamma_{\pm} \cup \Gamma_2$.
%Then, let the solutions of (\ref{eq:localzeta})
%be $P_I(\zeta)$ in the region $I$, $P_{II}(\zeta)$ in the region $II$,
%and so on. 
Then, define the function
\begin{equation}
\label{eq:phi}
\Phi(\zeta) :=
\begin{cases}
P(\zeta)\begin{pmatrix}0&-e^{-\frac{U_2(y)}{2\zeta^2}+\frac{U_1(y)}{\zeta}}\\
e^{\frac{U_2(y)}{2\zeta^2}-\frac{U_1(y)}{\zeta}}&0\end{pmatrix},
 & \zeta\in \Omega_{0,+}\cup \Omega_{4,+}\\
P(\zeta)\begin{pmatrix}0&-e^{-\frac{U_2(y)}{2\zeta^2}+\frac{U_1(y)}{\zeta}}\\
e^{\frac{U_2(y)}{2\zeta^2}-\frac{U_1(y)}{\zeta}}&e^{2i\zeta}\end{pmatrix}, &
\zeta\in \Omega_1,\\
P(\zeta)e^{-2i\zeta\sigma_3}, & \zeta\in  \Omega_{2,+},\\
P(\zeta), & \zeta\in  \Omega_{4,-}\cup \Omega_{2,-} \cup \Omega_{0,-},\\
P(\zeta)\begin{pmatrix} 1 & 0  \\
e^{-2i\zeta+\frac{U_2(y)}{2\zeta^2}-\frac{U_1(y)}{\zeta}}
&1\end{pmatrix}, &\zeta\in \Omega_3.
\end{cases}
\end{equation}

\begin{figure}[h]
\centering
\includegraphics[width=4.0in]{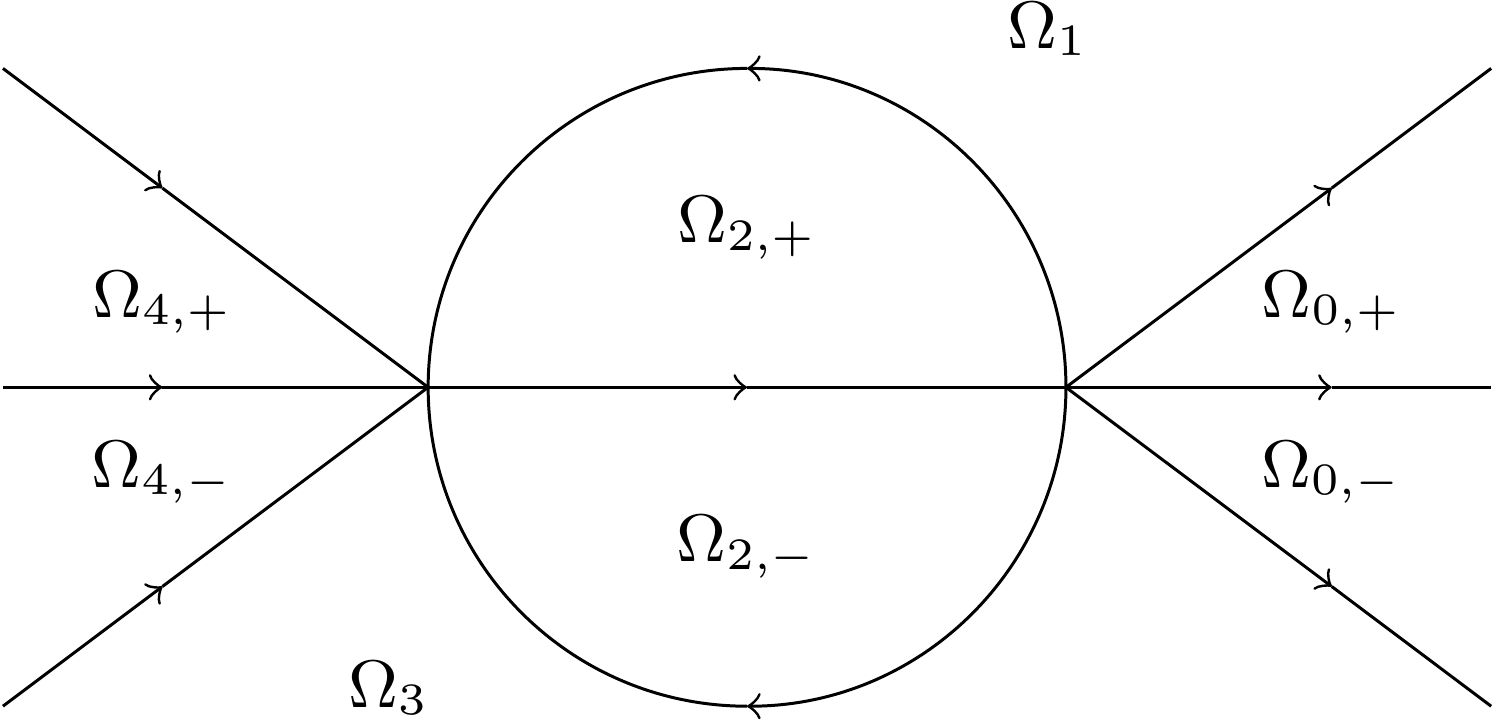}
\caption{The contour deformation for the RHP~\eqref{eq:RHPphi}}
\label{fig:local2}
\end{figure}

By direct inspection we can see that the matrix
$\Phi(\zeta)$ satisfies the RHP
\begin{equation}\label{eq:RHPphi}
\begin{split}
&1.\quad \text{$\Phi(\zeta)$ is analytic in
$\mathbb{C}\setminus\left(\Gamma_+\cup\Gamma_-\cup\Gamma_2\right)$},\\
&2.\quad \Phi_+(\zeta)=\Phi_-(\zeta)J_{\Phi}(\zeta),\quad \zeta\in \Gamma_+\cup\Gamma_-\cup\Gamma_2,\\
&3.\quad \Phi(\zeta)=I+O(\zeta^{-1}),\quad \zeta\rightarrow\infty,
\end{split}
\end{equation}
where the jump matrices $J_{\Phi}(\zeta)$  are given by
\begin{equation}
\label{eq:JPhi}
J_{\Phi}(\zeta):=
\begin{cases}
\begin{pmatrix} 1 & 0  \\
e^{-2i\zeta+\frac{U_2(y)}{2\zeta^2}-\frac{U_1(y)}{\zeta}}
&1\end{pmatrix}, & \zeta\in\Gamma_-,\\
\begin{pmatrix}1&e^{2i\zeta-\frac{U_2(y)}{2\zeta^2}+\frac{U_1(y)}{\zeta}}\\
-e^{-2i\zeta+\frac{U_2(y)}{2\zeta^2}-\frac{U_1(y)}{\zeta}}&0\end{pmatrix},
& \zeta\in \Gamma_+,\\
\begin{pmatrix} 1 & e^{2i\zeta-\frac{U_2(y)}{2\zeta^2}+\frac{U_1(y)}{\zeta}}  \\
0 &1\end{pmatrix},& \zeta\in \Gamma_2.
\end{cases}
\end{equation}

The function 
\begin{equation}
\label{eq:yzeta}
\Psi(\zeta):=\Phi(\zeta)e^{\left(i\zeta-\frac{U_2(y)}{4\zeta^2}
+\frac{U_1(y)}{2\zeta}\right)\sigma_3}
\end{equation}
satisfies a RHP with constant jumps on
$\Gamma_+\cup\Gamma_-\cup\Gamma_2$:
\begin{equation}\label{eq:JY}
J_{\Psi}(\zeta) :=
\begin{cases}
\begin{pmatrix} 1 & 0  \\
1 &1\end{pmatrix}, &  \zeta\in\Gamma_-,\\
\begin{pmatrix}1&1\\
-1&0\end{pmatrix},&  \zeta\in \Gamma_+,\\
\begin{pmatrix} 1 & 1  \\
0 &1\end{pmatrix}, &  \zeta\in \Gamma_2.
\end{cases}
\end{equation}
Furthermore, it has essential singularities at the points $\zeta=0$ and
$\zeta=\infty$:
\begin{subequations}
\label{eq:Yasym}
\begin{align}
\label{Yasym0}
\Psi(\zeta)&=\hat{\Psi}_0(\zeta)e^{\left(-\frac{U_2(y)}{4\zeta^2}
+\frac{U_1(y)}{2\zeta}\right)\sigma_3},\quad
\zeta\rightarrow 0,\\
\label{YasymInf}
\Psi(\zeta)&=\hat{\Psi}_{\infty}(\zeta)e^{i\zeta\sigma_3},\quad
\zeta\rightarrow \infty,
\end{align}
\end{subequations}
where $\hat{\Psi}_0(\zeta)$ and $\hat{\Psi}_{\infty}(\zeta)$ have
the asymptotic expansions
\begin{subequations}
\label{eq:psihat}
\begin{align}
\label{eq:psihat0}
\hat{\Psi}_0(\zeta)& =\Psi_0^{(0)}+\Psi_1^{(0)}\zeta+\dotsb, \quad
\zeta \to 0,\\
\hat{\Psi}_{\infty}(\zeta)&=I+\Psi_1^{(\infty)}\zeta^{-1}+\dotsb,
\quad \zeta \to \infty.
\end{align}
\end{subequations}
Note that the determinants of  $\hat{\Psi}_0$ and $\hat{\Psi}_{\infty}$
are both equal to one.

%\subsection{Differential equations}
From \eqref{eq:etavarsigma} it follows that $U_j(y)$, $j=1,2$,
behaves as 
\begin{equation*}
U_j(y)=u_{j,N}\left(1+O\left(\frac{\zeta^2}{N^2}\right)\right)=
u_j\left(1+O\left(\frac{\zeta^2}{N^2}\right)\right),\quad N \to \infty,
\end{equation*}
where $u_{j,N}$ and $u_j$ were defined in \eqref{eq:u1u2}
and~\eqref{eq:limu1u2}, respectively. Let $\Phi_{\mathrm{c}}(\zeta)$
and $\Psi_{\mathrm{c}}(\zeta)$ denote the functions \eqref{eq:phi}
and~\eqref{eq:yzeta} with $U_j(y)$ replaced by
$u_j$. Since the jump matrices $J_{\Phi}(\zeta)$ are all bounded, the
matrix $\Phi(\zeta)\Phi_{\mathrm{c}}^{-1}(\zeta)$ will have jump
discontinuities on the contours $\Gamma_2$ and $\Gamma_{\pm}$ of order
$I+O(N^{-2})$ and behaves as $I+O(\zeta^{-1})$ as
$\zeta\rightarrow\infty$. By the general theory of the RHP (see
\textit{e.g.}, \cite[Sec.~7]{DKMVZ99b} and also 
Eq.~\eqref{eq:Rest}),
we have
\begin{equation}
\label{eq:gen_th}
\Phi(\zeta)\Phi_{\mathrm{c}}^{-1}(\zeta)
=I+O\left(\frac{1}{N^2(\lvert \zeta \rvert+1)}\right).
\end{equation}
Hence, the matrix $\Phi_{\mathrm{c}}(\zeta)$ approximates
$\Phi(\zeta)$ up to $O(N^{-2})$ as $N\rightarrow\infty$. In what
follows we shall replace the functions $U_j(y)$ by the constants $u_j$
and identify $\Phi_{\mathrm c}(\zeta)$ and $\Psi_{\mathrm c}(\zeta)$ with
$\Phi(\zeta)$ and $\Psi(\zeta)$.

Let us consider the RHP with jumps~\eqref{eq:JY} and
asymptotic behaviour~\eqref{eq:Yasym}. Since the matrix $\Psi(\zeta)$
has constant jumps on the complex plane, the matrices
\begin{subequations}
\label{eq:lax}
\begin{align}
\label{eq:lax_A}
A(\zeta)& :=\frac{\p \Psi(\zeta)}{\p\zeta}\Psi^{-1}(\zeta) \\
B_j(\zeta)& :=\frac{\p \Psi(\zeta)}{\p
  u_j}\Psi^{-1}(\zeta),
\quad j=1,2,
\end{align}
\end{subequations}
are rational functions with poles at $\zeta=0$ and $\zeta=\infty$
only. From \eqref{eq:Yasym} we see that their explicit expressions are
\begin{subequations}
\label{eq:AB}
\begin{align}
\label{eq:A}
A(\zeta)&=\left(\frac{1}{2}\left(\frac{u_2}{\zeta^3}-\frac{u_1}{\zeta^2}\right)\hat{\Psi}_0(\zeta)\sigma_3\hat{\Psi}_0^{-1}(\zeta)\right)_{pp}
+i\sigma_3=\sum_{j=1}^3\frac{A_j}{\zeta^j}+i\sigma_3,\\
\label{eq:B1}
B_1(\zeta)&=\left(\frac{1}{2\zeta}\hat{\Psi}_0(\zeta)\sigma_3\hat{\Psi}_0^{-1}(\zeta)\right)_{pp}
=\frac{A_3}{u_2\zeta},\\
\label{eq:B2}
B_2(\zeta)&=-\left(\frac{1}{4\zeta^2}\hat{\Psi}_0(\zeta)\sigma_3\hat{\Psi}_0^{-1}(\zeta)\right)_{pp}
=-\frac{A_3}{2u_2\zeta^2}-\frac{1}{2u_2\zeta}\left(A_2+\frac{u_1}{u_2}A_3\right),
\end{align}
\end{subequations}
where $(\cdot)_{pp}$ denotes the singular part at $\zeta=0$.

We can write the matrix $A$ as follows:
\begin{equation}\label{eq:Azeta}
\begin{split}
A(\zeta)&=\begin{pmatrix}\sum_{j=2}^3a_j\zeta^{-j}&\sum_{j=2}^3b_j\zeta^{-j}\\
\sum_{j=2}^3c_j\zeta^{-j}&-\sum_{j=2}^3a_j\zeta^{-j}\end{pmatrix}+
\begin{pmatrix}0&b_1\\
               c_1&0\end{pmatrix}\zeta^{-1}+i\sigma_3.
\end{split}
\end{equation}
The coefficient of $\zeta^{-1}$ can be extracted by
expanding~\eqref{eq:lax} near $\zeta=\infty$,
\begin{equation*}
\begin{split}
\frac{\p
\Psi(\zeta)}{\p\zeta}\Psi^{-1}(\zeta)&=\frac{\p\hat{\Psi}_{\infty}(\zeta)}%
{\p\zeta}\hat{\Psi}_{\infty}^{-1}(\zeta)+
i\hat{\Psi}_{\infty}(\zeta)\sigma_3\hat{\Psi}_{\infty}^{-1}(\zeta),\\
&=i\sigma_3+i\left[\Psi^{(\infty)}_1,\sigma_3\right]\zeta^{-1}+O(\zeta^{-2}).
\end{split}
\end{equation*}
From (\ref{eq:AB}) we see that the coefficients of $\zeta^{-6}$,
$\zeta^{-5}$ and $\zeta^{-4}$ in the determinant of $A(\zeta)$
coincide with those in
\[
\det \left(\frac{1}{2}\left(\frac{u_2}{\zeta^3}-\frac{u_1}{\zeta^2}\right)\sigma_3\right).
\]
This implies
\begin{subequations}
\label{eq:coefrel}
\begin{gather}
a_3^2+b_3c_3=\frac{1}{4}u_2^2,\\
2a_2a_3+b_2c_3+c_2b_3=-\frac{1}{2}u_1u_2,\\
b_1c_3+c_1b_3+b_2c_2+a_2^2=\frac{1}{4}u_1^2.
\end{gather}
\end{subequations}
Finally, the compatibility of the linear differential systems~\eqref{eq:lax} implies
\begin{subequations}
\label{eq:compat}
\begin{gather}
\label{eq:compat1}
\p_jA(\zeta)-\p_{\zeta}B_j(\zeta)+[A(\zeta),B_j(\zeta)]=0,\quad j=1,2,\\
\p_1B_2(\zeta)-\p_2B_1(\zeta)+[B_2(\zeta),B_1(\zeta)]=0,
\end{gather}
\end{subequations}
where $\p_j$ is the derivative with respect to $u_j$ and
$\p_{\zeta}$ the derivative with respect to $\zeta$. 

\subsection{Hamilton Equations}

The compatibility conditions~\eqref{eq:compat} constitute a system of
PDEs in the variables $u_1$ and $u_2$. The theory of isomonodromic
deformations of linear ODEs~\cite{JMU81} states that such PDEs are
equivalent to an integrable Hamiltonian system of ODEs, which can be
studied much more easily. The main idea is to express the terms
involving the commutators in Eqs.~\eqref{eq:compat1} as Hamiltonian
flows on a symplectic manifold.  

Consider the Laurent polynomial\footnote{For the sake of simplicity, we use the same notation
for the Lax pair~\eqref{eq:lax} and for the generic matrices~\eqref{eq:mapsA} and~\eqref{eq:Bmat}.  It will be evident from the context which ones we are referring to.}
\begin{equation}
\label{eq:mapsA}
A(\zeta) :=\sum_{j=1}^3A_j\zeta^{-j}+i\sigma_3,
\end{equation}
as well as
\begin{equation}
  \label{eq:Bmat}
  B(\zeta) := B_1\zeta^{-1} + B_2 \zeta^{-2}.
\end{equation}
In addition, require that the coefficients of the 
expansion~\eqref{eq:mapsA} should be of the form
\begin{equation}
\label{eq:lau_coef}
  A_1 = \begin{pmatrix} 0 & b_1 \\ c_1 &0 \end{pmatrix}, \qquad A_j
  = \begin{pmatrix} a_j & b_j \\ c_j & -a_j \end{pmatrix}, \qquad j=2,3.
\end{equation}
Note that the $A_j$'s are traceless; therefore, they belong to the Lie
algebra $\mathfrak{sl}_2(\mathbb{C})$.  We shall see that the
deformations
\begin{equation}
  \label{eq:flow}
  \dot A = \left [A,B\right]
\end{equation}
define a Hamiltonian flow.

The ODEs~\eqref{eq:flow} do not identify the pair of matrices $A$
and $B$ uniquely.  If we conjugate both $A$ and $B$ by a constant
matrix $G\in \mathfrak{sl}_2(\mathbb{C})$, which is either diagonal or a scalar multiple of $A_{1}$, then the form of the
coefficients~\eqref{eq:lau_coef} and the trajectories of the
flow~\eqref{eq:flow} remain invariant.  Thus, rather than the set of
matrices~\eqref{eq:mapsA}, we look at their conjugacy classes.  This
imposes two conditions on the matrix entries of the coefficients $A_j$
and reduces the independent parameters from eight to six.  In the
particular example~\eqref{eq:A}, these extra relations
are equivalent to Eqs.~\eqref{eq:coefrel}. For simplicity, and without
loss of generality, we shall not distinguish between the set of
matrices defined in~\eqref{eq:mapsA} and their equivalent classes.

To begin with we need to equip this six-dimensional manifold with a
Poisson structure.  Let the loop algebra $\mathfrak{g}$ be the set of
smooth maps $f:S^1\rightarrow\mathfrak{sl}_2(\mathbb{C})$, where $S^1$
is the unit circle. Then, split $\mathfrak{g}$ into the subalgebras
$\mathfrak{g}_+$ and $\mathfrak{g}_-$, where $\mathfrak{g}_+$
($\mathfrak{g}_-$) is the set of maps that admits holomorphic
extension to the inside (outside) of the unit circle;  the
maps in $\mathfrak{g}_-$ should also vanish at infinity. Using the pairing
\[
<X(\zeta),Y(\zeta)>=\res_{\zeta=0}\tr\left(XY\right),
\]
we can identify $\mathfrak{g}_-$ with the dual of $\mathfrak{g}_+$.
This structure allows us to define  Poisson brackets on
$\mathfrak{g}_-$ (see, \textit{e.g.}, \cite{FT87})
\begin{equation}
\label{eq:poisson}
\left\{f,g\right\}(A):=\left<A,[df(A),dg(A)]\right>,
\end{equation}
where $f$ and $g$ are functions on
$\mathfrak{g}_-=\mathfrak{g}_+^{\ast}$. 

Using the method of the moment map~\cite{AHP88,AHH93,Har94,HR95}, this
Poisson structure can be restricted to the space of matrices of the
form~\eqref{eq:mapsA}. The Poisson brackets between the matrix entries
$a_j$, $b_j$ and $c_j$ are degenerate and are given by
\begin{subequations}
\label{eq:bracket}
\begin{alignat}{3}
\{a_j,b_k\}&=-b_{j+k-1},&\qquad &  \{a_j,c_k\} =c_{j+k-1}, \qquad & j+k\leq
4,\\
\{b_j,c_k\}&=2a_{j+k-1},&& j+k-1\leq 4. &&
\end{alignat}
\end{subequations}
All the other brackets vanish. Equations~\eqref{eq:bracket} define a
Poisson manifold.

\begin{remark} 
  Although, strictly speaking, the matrices~\eqref{eq:mapsA} do not
  belong to the algebra $\mathfrak{g}_-$, because of the term
  $i\sigma_3$, one can think of the space of matrices~\eqref{eq:mapsA}
  as being parameterised by the entries $a_j$, $b_j$ and $c_j$ and use
  Eqs.~\eqref{eq:bracket} to define the Poisson brackets. In fact, the
  original construction in \cite{AHP88,AHH93} applies to much more
  general matrices.  The interested reader is invited to consult these
  references.
\end{remark}

The symplectic leaves in this Poisson manifold are the co-adjoint
orbits of the matrix $A$, specified by $\det(A)=\mathrm{const.}$ A
set of canonical coordinates on these symplectic leaves can be found
as follows~\cite{AHH93,VN84}. Let $\psi(\zeta)$ be an eigenvector of
$A(\zeta)$ with eigenvalue $\lambda(\zeta)$. Then, the poles $q_j$ of
$\psi(\zeta)$ and the values $p_j=\lambda(q_j)$ of the eigenvalue
$\lambda(\zeta)$ at $q_j$ satisfy the relations
\begin{equation*}
\{p_j,q_k\}=\delta_{jk},\quad \{p_j,p_k\}=\{q_j,q_k\}=0.
\end{equation*}
For example, let $\psi(\zeta)$ be an eigenvector of $A(\zeta)$
normalized by
\begin{equation*}
(1,0)\cdot\psi(\zeta)=1.
\end{equation*}
The poles $q_j$ of $\psi(\zeta)$ are the zeros of the
polynomial
\begin{equation*}
b_1\zeta^2+b_2\zeta+b_3=0,
\end{equation*}
and the eigenvalues at these points are
\begin{equation*}
p_j=-\sum_{k=2}^3a_kq_j^{-k}-i.
\end{equation*}
Therefore, the symplectic form on these symplectic leaves are given
by
\begin{equation*}
\begin{split}
-\sum_{j=1}^2\sum_{k=2}^3d a_kq_j^{-k}\wedge d
q_j&=\sum_{j=1}^2\left(\frac{1}{2}d a_3\wedge d q_j^{-2}+d
a_2\wedge d q_j^{-1}\right)\\
&=d\left(-a_3\frac{b_2}{b_3}+a_2\right)\wedge
d\left(-\frac{b_2}{b_3}\right)+d a_3\wedge
d\left(-\frac{b_1}{b_3}\right).
\end{split}
\end{equation*}
Hence, 
\begin{equation}\label{eq:canon}
P_1=-a_3\frac{b_2}{b_3}+a_2,\quad Q_1=-\frac{b_2}{b_3},\quad
P_2=a_3,\quad Q_2=-\frac{b_1}{b_3},
\end{equation}
are a set of canonical coordinates on the symplectic leaves and the parameters
$P_1$, $Q_1$, $P_2$, $Q_2$, $v_1$, $v_2$, where $v_1$ and $v_2$ are functions of the matrix elements $a_j$, $b_j$ and $c_j$, constitute a set of coordinates on the
Poisson manifold.

% It turns out that the deformations $\dot{A}=[B,A]$ can be expressed as Hamiltonian flows
% in this Poisson structure 

The equivalence between the deformations~\eqref{eq:flow} and
Hamiltonian systems is provided by the following proposition, which is a
particular case of the flows studied in~\cite{AHP88,AHH90,AHH93}.
\begin{proposition}
\label{pro:ham}
Let $A(\zeta)$ belong to the set of matrices defined
in~\eqref{eq:mapsA} and $B=(\zeta^{k}A)_-$, where $(\cdot)_-$ denotes
the projection onto $\mathfrak{g}_-$ and $k=1,2$.  Then, the
ODEs~\eqref{eq:flow} are Hamiltonian with respect to the Poisson brackets
\eqref{eq:poisson}. The Hamiltonian function is
\begin{equation*}
H_B=\frac{1}{2}\res_{\zeta=0}\tr\left(A^2\zeta^{k}\right).
\end{equation*}
\end{proposition}

The matrices $A(\zeta)$, $B_1(\zeta)$ and $B_2(\zeta)$ introduced in
Eqs.~\eqref{eq:AB} belong to class of matrices in the hypothesis of
Proposition~\ref{pro:ham}.  In particular, the commutators in
Eq.~\eqref{eq:compat1} define two flows 
\begin{equation}
\label{eq:flows_LP}
\dot{A}=[B_j,A], \quad j=1,2,
\end{equation}
whose Hamiltonians are
\begin{equation}
\label{eq:flowham}
H_1=\frac{1}{2u_2}\res_{\zeta=0}\tr\left(A^2\zeta^2\right) \quad
\text{and} \quad 
H_2=-\frac{1}{2u_2}\res_{\zeta=0}\tr\left(\frac{1}{2}A^2\zeta\right)
-\frac{u_1}{2u_2}H_1.
\end{equation}

Equations~\eqref{eq:canon} and \eqref{eq:coefrel} allow us to express the
matrix entries $a_j$, $b_j$ and $c_j$ in terms of the canonical
variables $P_1$, $Q_1$, $P_2$, $Q_2$ and the parameters $u_1$, $u_2$,
which together provide a set of coordinates for the Poisson manifold.
In turn, the Hamiltonians $H_1$ and $H_2$ can be written in terms of
these coordinates.  The outcomes of this calculation are
formulae~\eqref{eq:H1} and~\eqref{eq:H2}.

The terms $\p_{\zeta}B_j$ in Eq.~\eqref{eq:compat1} imply that
deformations~\eqref{eq:flows_LP}, and hence the Hamiltonian flows
generated by $H_1$ and $H_2$, are not sufficient to describe the
compatibility conditions~\eqref{eq:compat1}. However, using a technique
developed by Mazzocco and Mo~\cite{MM07}, we can 
show that Eq.~\eqref{eq:compat1} is
equivalent to a time-dependent Hamiltonian system.  

We are now in a position to prove the following.
\begin{theorem}
\label{thm:hamuk} 
Let $H_j$ and $h_j$, $j=1,2$ be the Hamiltonians given in
Theorem~{\rm \ref{thm:main1}}, Eqs.~\eqref{eq:hamcan}.
The deformations of the matrix $A(\zeta)$ defined by
\begin{equation}
\label{eq:def_repeat}
\p_jA(\zeta)-\p_{\zeta}B_j(\zeta)+[A(\zeta),B_j(\zeta)]=0,\quad
j=1,2,
\end{equation}
are equivalent to the time-dependent Hamiltonian equations
\begin{equation}
\label{eq:heq_sh}
\p_jP_k=-\frac{\p (H_j+h_j)}{\p Q_k},\quad \p_jQ_k=\frac{\p
(H_j+h_j)}{\p P_k},\quad j,k=1,2.
\end{equation}
\end{theorem}
 \begin{proof}
   Let $\nu_1,\dotsc,\nu_4,u_1,u_2$ be a new set of coordinates on the
   Poisson manifold and denote by $\p_1^{\nu}$ and $\partial^\nu_2$
   the partial derivatives that keep the variables $\nu_1,\dotsc,\nu_4$
   fixed. We can choose these coordinates so that they are related to
   the matrix entries $a_j$, $b_j$ and $c_j$ by the equations
\begin{subequations}
\label{eq:explicitder}
\begin{alignat}{3}
\p_{1}^{\nu}a_3&=\p_{1}^{\nu}b_3=\p_{1}^{\nu}c_3=0, &\quad&
\p_1^{\nu}a_2=-\frac{a_3}{u_2}, \qquad
\p_1^{\nu}b_2=-\frac{b_3}{u_2},\\
\p_{1}^{\nu}a_1&=\p_{1}^{\nu}b_1=\p_{1}^{\nu}c_1=0,
& &\p_1^{\nu}c_2=-\frac{c_3}{u_2}, \qquad\p_2^{\nu}c_3=\frac{c_3}{u_2}, \\
\p_{2}^{\nu}a_1& =\p_{2}^{\nu}b_1=\p_{1}^{\nu}c_1=0, & & 
\p_2^{\nu}a_3 =\frac{a_3}{u_2}, \qquad \;\: \: \p_2^{\nu}b_3=\frac{b_3}{u_2},\\
\p_2^{\nu}a_2&=\frac{1}{2u_2}\left(a_2+\frac{u_1}{u_2}a_3\right), &&
\p_2^{\nu}b_2=\frac{1}{2u_2}\left(b_2+\frac{u_1}{u_2}b_3\right), \\
\p_2^{\nu}c_2&=\frac{1}{2u_2}\left(c_2+\frac{u_1}{u_2}c_3\right). && 
\end{alignat}
\end{subequations}
Using Eqs.~\eqref{eq:coefrel}, we can check that this definition is
compatible with $\p_j^{\nu}u_i=\delta_{ij}$. 

Note that Eqs.~\eqref{eq:explicitder} imply 
\begin{equation}
\label{eq:Ader}
\p_1^{\nu}A=\p_{\zeta}B_1 \quad \text{and} \quad \p_2^{\nu}A=\p_{\zeta}B_2.
\end{equation}
Therefore, we can write Eq.~\eqref{eq:def_repeat} as time-dependent
Hamiltonian equations:
\begin{equation}
\label{eq:timedep}
\p_jA=\{H_j,A\}+\p_j^{\nu}A.
\end{equation}

In order the find the evolution equations in the form
\[
\p_jP_k=-\frac{\p \mathcal{H}_j}{\p Q_k}, \qquad \p_jQ_k=\frac{\p
  \mathcal{H}_j}{\p P_k},
\]
for suitable Hamiltonians $\mathcal{H}_j$, we must replace
$\p_j^{\nu}$ by the partial differentiation $\p_j^{\mathrm{can}}$ that
keeps the canonical coordinates $P_1$, $Q_1$, $P_2$, $Q_2$ fixed.

We have
\begin{equation}
\label{eq:derrel}
\p_{j}^{\nu}=\p_{j}^{\mathrm{can}}+\sum_{k=1}^2\left(\p_{j}^{\nu}P_k
\p_{P_k}+\p_j^{\nu}Q_k\p_{Q_k}\right).
\end{equation}
From \eqref{eq:explicitder} and \eqref{eq:canon}, we can compute the
derivatives $\p_j^{\nu}P_k$ and $\p_j^{\nu}Q_k$:
\begin{subequations}
\label{eq:derPQ}
\begin{align*}
\p_1^{\nu}P_1&=\p_1^{\nu}P_2=\p_1^{\nu}Q_2=0,\quad
\p_1^{\nu}Q_1=\frac{1}{u_2},\\
\p_2^{\nu}P_2&=\frac{P_2}{u_2},\quad
\p_2^{\nu}Q_2=-\frac{Q_2}{u_2},\quad
\p_2^{\nu}P_1=\frac{1}{2u_2}P_1,\\
\p_2^{\nu}Q_1& =-\frac{1}{2u_2}Q_1-\frac{u_1}{2u_2^2}.
\end{align*}
\end{subequations}
Combining these formulae with \eqref{eq:derrel} we arrive at 
\begin{equation*}
\p_{j}^{\nu}A =\p_{j}^{\mathrm{can}}A+\{h_j,A\},
\end{equation*}
where
\begin{equation*}
h_1 =\frac{P_1}{u_2},\quad
h_2=-\frac{P_2Q_2}{u_2}-\frac{P_1Q_1}{2u_2}-\frac{u_1}{2u_2^2}P_1
\end{equation*}
and $\mathcal{H}_j= H_j + h_j$. Finally, Eq.~\eqref{eq:timedep} becomes
\begin{equation*}
\p_jA=\{\mathcal{H}_j,A\}+\p_j^{\mathrm{can}}A.
\end{equation*}
\end{proof}

\begin{remark}
  The Hamilton equations~\eqref{eq:heq_sh} are ODEs in the variables
  $u_1$ and $u_2$. A special solution of these ODEs, specified by the
  jump conditions~\eqref{eq:JY}, will give the local parametrix via
  \eqref{eq:lax}, \eqref{eq:yzeta}, \eqref{eq:phi} and
  \eqref{eq:S0P}. In particular, we will show that the
  Hamiltonians $H_1$ and $H_2$ provide the leading order
  terms of the differential identities~\eqref{eq:diffid}.
\end{remark}

\section{Final Solution of the RHP}
\label{se:final}
Here we prove that the outer and local parametrices that we constructed
in Secs.~\ref{se:lens} and \ref{se:local}, respectively, are indeed
good approximations to the solution $S(y)$ of the RHP~\eqref{eq:RHS}.

Let us define
\begin{equation}
\label{eq:Rx} 
          R(y):= \begin{cases}
         S(y)\left(S^{\left(p\right)}(y)\right)^{-1}, & y
        \in D_{\pm 2}\cup D_0, \\
         S(y)\bigr(S^{\infty}(y)\bigr)^{-1}, & y 
         \in \mathbb{C}\setminus\left(D_{\pm 2}\cup
         D_0\right),
       \end{cases}
\end{equation}
where $S^{(p)}(y)$ are the local parametrices inside the neighborhoods
$D_{\pm2}$ and $D_0$. The function $R(y)$ has jump discontinuities on
the contour $\Gamma_R$ shown in Fig.~\ref{fig:sigma}.
\begin{figure}[ht]
\centering 
\includegraphics[width=6in]{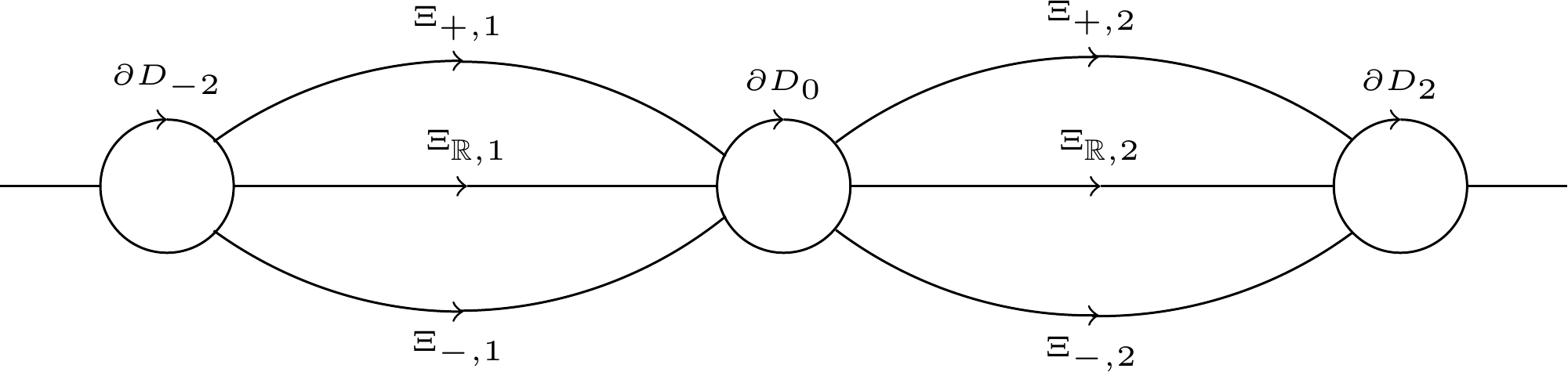}
\caption{The contour $\Gamma_R$.}\label{fig:sigma}
\end{figure}
In particular, $R(y)$ satisfies the RHP
\begin{equation}\label{eq:RHR}
\begin{aligned}
1. \quad & \text{$R(y)$ is analytic in $\mathbb{C}\setminus\Gamma_R$,} \\
2.\quad & R_+(y)=R_-(y)J_R(y),\quad y \in \Gamma_R,\\
3. \quad & R(y)=I+O(y^{-1}),  \quad y \to \infty.
\end{aligned}
\end{equation}

From the definition of $R(y)$, it follows that the
order of magnitude of the jump matrix $J_R(y)$ as $N \to \infty$ is
\begin{equation}
\label{eq:Jx} J_R(y)=
       \begin{cases}
         I+O(N^{-1}), & y\in\p D_{\pm 2}\cup \p D_0\cup\Gamma_0, \\
         I+O\left(e^{-N\eta}\right), & \text{for some fixed $\eta>0$
         and  $y \in  \Gamma_R \setminus 
         \left(\p D_{\pm 2}\cup \p D_0\cup\Gamma_0\right)$,}
       \end{cases}
\end{equation}
where $\Gamma_0$ is the subset of $[-2,2]$ that lies outside of
$D_{\pm2}$ and $D_0$. The general theory of the RHP (see,
\textit{e.g.}~\cite[Sec.~7]{DKMVZ99b}, and also Eq.~\eqref{eq:gen_th})
asserts that the third statement of the RHP~\eqref{eq:RHR}
and Eq.~\eqref{eq:Rx} lead to 
\begin{equation}
\label{eq:Rest}
\begin{split}
R(y)=I+O\left(\frac{1}{N\left(\abs{y}+1\right)}\right),
\end{split}
\end{equation}
uniformly in $\mathbb{C}$. Therefore, at leading order as $N\to
\infty$ the solution of the RHP~\eqref{eq:RHS} is given by
\begin{equation}
\label{eq:approxS} S(y)= \begin{cases}
  \bigl(I+O\left(N^{-1}\bigr)\right)S^{\left(p\right)}(y), & y \in
  D_{\pm 2}\cup D_0, \\
  \bigl(I+O\left(N^{-1}\bigr)\right)S^{\infty}(y), & y \in
  \mathbb{C} \setminus \left(D_{\pm 2}\cup D_0\right),
       \end{cases}
\end{equation}

Combining these expressions with equations~\eqref{eq:S}
and~\eqref{eq:Tg} we obtain an asymptotic formula for the solution
of the original RHP~\eqref{eq:RHP}.

\section{Asymptotics of the Hankel Determinant} 
\label{se:asy_Hank}
In this section we complete the proof of Theorem~\ref{thm:main1} by
computing the leading order terms of the differential
identities~\eqref{eq:diffid}. We shall see that their asymptotics are
determined by the behaviour of the local parametrix near the
origin. In turn, such formulae can be identified with the
Hamiltonians~\eqref{eq:H1} and~\eqref{eq:H2}.

\begin{lemma}
\label{le:diffid}
  Let $G_N(u_{1,N},u_{2,N})$ be the partition
  function~\eqref{eq:tildeB} and $\hat{\Psi}_0(\zeta)$ the 
  matrix introduced in~\eqref{Yasym0}.  At leading order as $N \to \infty$ we
  have
\begin{subequations}
\label{eq:lder}
\begin{align}
\frac{\p\log G_N}{\p u_{1,N}}&=-\frac{1}{2}\res_{\zeta=0}
 \frac{1}{\zeta}\tr\left(\hat{\Psi}_0^{-1}(\zeta)
  \hat{\Psi}^{\prime}_{0}(\zeta)\sigma_3d\zeta\right)+O(N^{-1}),\\
\frac{\p\log G_N}{\p u_{2,N}}&=\frac{1}{4}\res_{\zeta=0}
    \frac{1}{\zeta^2}\tr\left(\hat{\Psi}_0^{-1}(\zeta)
  \hat{\Psi}^{\prime}_{0}(\zeta)\sigma_3d\zeta\right)
+O(N^{-1}).
\end{align}
\end{subequations}
\end{lemma}
\begin{proof}
  The solution $Y(y)$ of the original RHP~\eqref{eq:RHP} can be
  expressed in terms of the $S(y)$ using the definitions \eqref{eq:Tg}
  and \eqref{eq:S}:
\begin{equation}
\label{eq:YS}
Y(y)=e^{\frac{Nl\sigma_3}{2}}S(y)e^{\left(Ng(y)-\frac{Nl}{2}\right)\sigma_3}.
\end{equation}
In addition, from Eq.~\eqref{eq:Rx} near the origin we have
$S(y)=R(y)S^{(0)}(y)$. Then, Eqs.~\eqref{eq:S0P} give 
\begin{subequations}
\label{eq:YP}
\begin{align}
Y_+(y)&=e^{\frac{Nl\sigma_3}{2}}R(y)S^{\infty}(y)e^{-N\tilde{g}_+(0)\sigma_3}
   \begin{pmatrix}0&-1\\1&0\end{pmatrix}\notag \\
& \quad \times P_+(\zeta) e^{N\left(\tilde{g}_+(0)+g_+(y)
  -\frac{Nl}{2}\right)\sigma_3},\quad \ipart(y)>0,\\
Y_-(y)&=e^{\frac{Nl\sigma_3}{2}}R(y)S^{\infty}(y)e^{-N\tilde{g}_+(0)\sigma_3}
  P_-(\zeta)e^{N\left(\tilde{g}_+(0)+g_-(y)-\frac{Nl}{2}\right)\sigma_3},\quad
\ipart(y)<0.
\end{align}
\end{subequations}
Now, from Eqs.~\eqref{eq:phi} and~\eqref{eq:yzeta} we have
\begin{equation*}
P_{\pm}(\zeta)=\Psi_{\pm}(\zeta)e^{\left(\mp
i\zeta+\frac{u_2}{4\zeta^2}-\frac{u_1}{2\zeta}\right)\sigma_3}.
\end{equation*}
Therefore, $P(\zeta)$ is related to
$\hat{\Psi}_0 (\zeta)$ by
\begin{equation}
\label{eq:PPsi0}
P_{\pm}(\zeta)=\hat{\Psi}_0(\zeta)e^{\mp i\zeta}.
\end{equation}
Define
\begin{equation}
\label{eq:Kpm}
K_+(\zeta): =e^{-N\tilde{g}_+(0)\sigma_3}\begin{pmatrix}0&-1\\1&0\end{pmatrix}
\hat{\Psi}_0(\zeta)e^{- i\zeta},\qquad
K_-(\zeta):=e^{-N\tilde{g}_+(0)\sigma_3}
\hat{\Psi}_0(\zeta)e^{i\zeta}.
\end{equation}
By combining Eqs.~\eqref{eq:YP} and~\eqref{eq:Kpm} we can write 
\begin{equation*}
%\label{eq:ha}
\begin{split}
\tr\left(Y_{\pm}^{-1}\frac{dY_{\pm}}{dy}\sigma_3\right)dy&=
2Ng_{\pm}^{\prime}(y)dy\pm2id\zeta
+\tr\left(\hat{\Psi}_0^{-1}(\zeta)\hat{\Psi}^{\prime}_{0}(\zeta)
\sigma_3d\zeta\right)\\
 & \quad +\tr\left(K_{\pm}(\zeta)\left(S_{\pm}^{\infty}(y)\right)^{-1}
\left(S_{\pm}^{\infty}(y)\right)^{\prime}K_{\pm}(\zeta)\sigma_3dy\right)\\
& \quad +\tr\left(K_{\pm}^{-1}(\zeta)\left(S_{\pm}^{\infty}(y)\right)^{-1}
R^{-1}(y)R^{\prime}(y)S_{\pm}^{\infty}(y)K_{\pm}(\zeta)\sigma_3dy\right).
\end{split}
\end{equation*}

The final step consists in inserting this expression into into the
differential identities~\eqref{eq:diffid}.  Note that the definitions
of the jump matrices of the RHPs~\eqref{eq:RHS}
and~\eqref{eq:localpara0} combined with Eq.~\eqref{eq:Rx} imply that
the problem~\eqref{eq:RHR} does not have any oscillatory factors.
Therefore, we can safely differentiate the right-hand side of
Eq.~\eqref{eq:Rest} and see that $R'(y)= O(1/N)$.  Finally, the
conformal map~\eqref{eq:zeta0} leads to~\eqref{eq:lder}.
\end{proof}

The leading order terms of the identities~\eqref{eq:diffid} are the
derivatives of the $\tau$-function for the equations \eqref{eq:compat}
(see, \textit{e.g.},~\cite{JMU81}).  In unpublished work, Bertola,
Harnad, Hurtubise and Putsai~\cite{BHHP,H} have shown that they are
given by the Hamiltonians \eqref{eq:flowham}.  For completeness we
repeat the proof of this statement for the system that we study.
\begin{proposition} 
\label{pro:BHHP}
  The Hamiltonians~\eqref{eq:flowham} can be expressed as
\begin{subequations}
\label{eq:speres}
\begin{align}
H_1&=\frac{1}{2}\res_{\zeta=0}\frac{1}{\zeta}
\tr\left(\hat{\Psi}_0^{-1}(\zeta)\hat{\Psi}^{\prime}_{0}(\zeta)\
\sigma_3d\zeta\right),\\
H_2&=-\frac{1}{4}\res_{\zeta=0}\frac{1}{\zeta^2}
\tr\left(\hat{\Psi}_0^{-1}(\zeta)\hat{\Psi}^{\prime}_{0}(\zeta)
\sigma_3d\zeta\right).
\end{align}
\end{subequations}
\end{proposition}
\begin{proof} 
Let $\lambda(\zeta)$ be an eigenvalue of the matrix
  $A(\zeta)$ defined in~\eqref{eq:lax_A}. The expansion of
  $\lambda(\zeta)$ near $\zeta=0$ is
\begin{equation}
\label{eq:lambdexp}
\lambda(\zeta)=\frac{1}{2}u_2\zeta^{-3}-\frac{1}{2}u_1\zeta^{-2}+H_1
-2H_2\zeta+O(\zeta^2).
\end{equation}
Then, note that from Eq.~\eqref{eq:lax_A} we have
\begin{equation*}
\hat{\Psi}^{-1}_0(\zeta)\hat{\Psi}^{\prime}_0(\zeta)+\frac{1}{2}
\left(\frac{u_2}{\zeta^{3}}-\frac{u_1}{\zeta^{2}}\right)
\sigma_3=\hat{\Psi}^{-1}_0(\zeta)A\hat{\Psi}_0(\zeta).
\end{equation*}
Set $\hat{A}:=\hat{\Psi}^{-1}_0A\hat{\Psi}_0$ and write 
\begin{equation*}
\hat{A}=\frac{1}{2}
\left(\frac{u_2}{\zeta^{3}}-\frac{u_1}{\zeta^{2}}\right)\sigma_3+\begin{pmatrix}\Pi&F_{12}\\
F_{21}&-\Pi\end{pmatrix},
\end{equation*}
where $\Pi$ and $F_{ij}$ are bounded at $\zeta=0$. In particular, we
have
\begin{equation}
\label{eq:trpi}
\tr\left(\hat{\Psi}_0^{-1}(\zeta)\hat{\Psi}^{\prime}_{0}(\zeta)\sigma_3\right)=
\frac{u_2}{\zeta^{3}}-\frac{u_1}{\zeta^{2}}+2\Pi.
\end{equation}

Let 
\begin{equation*}
\upsilon_{\pm}:=\lambda\pm\left(\frac{1}{2}
\left(\frac{u_2}{\zeta^{3}}-\frac{u_1}{\zeta^{2}}\right)+\Pi\right).
\end{equation*}
Since $\lambda$ is an eigenvalue of $A$ and $\hat{A}$ is
conjugated to $A$, we have
\begin{equation*}
\det\left(\lambda-\hat{A}\right)=\upsilon_+\upsilon_--F_{12}F_{21}=0.
\end{equation*}
From the expansion \eqref{eq:lambdexp}, we see that $\upsilon_+ =
O(\zeta^{-3})$, while $F_{12}F_{21} =O(1)$. Therefore, $\upsilon_- =
O(\zeta^3)$. Hence, $\Pi$
behaves as
\begin{equation*}
\Pi=H_1-2H_2\zeta+O(\zeta^2), \quad \zeta \to 0.
\end{equation*}
This equation and~\eqref{eq:trpi} lead to Eqs.~\eqref{eq:speres}.
\end{proof}
Lemma~\ref{le:diffid} and Proposition~\ref{pro:BHHP} complete the
proof 
of Theorem \ref{thm:main1}.

\section{Initial Conditions}
\label{se:asymu2}

The initial conditions of the Hamiltonian system of ODEs in
Theorem~\ref{thm:main1}, Eq.~\eqref{eq:hameq}, are provided by the
asymptotic limit as $u_2 \to 0$ of the solution $\Phi(\zeta)$ of the
RHP~\eqref{eq:RHPphi}. The outcome of this calculation
is Theorem~\ref{thm:init}.  When $u_1=0$ these initial conditions
simplify considerably reducing to formulae~\eqref{eq:asympcan} in
Corollary~\ref{co:asym_u2}.

The main idea of the proof is to deform the local
paramatrix~\eqref{eq:psihat} with a matrix $\tilde{R}(\zeta)$ that
behaves as $I +O(\zeta^{-1})$ as $\zeta \to \infty$ and that satisfies
a RHP with jump discontinuities of order $I + O(\sqrt{u_2})$ as $u_2
\to 0$.  Therefore, $\tilde{R}(\zeta)$ can be expressed in terms of a
Neumann series in $\sqrt{u_2}$, whose coefficients can be
computed using the technique in~\cite[Sec.~7.2]{DKMVZ99b}.  The first
four terms of this asymptotic expansion provide the initial
conditions.  

Before we discuss the proof of Theorem~\ref{thm:init} we need to
introduce some preliminary definitions and results.   Throughout this
section we shall assume that the parameters $u_1$ and $u_2$ scale as
$u_1=\tilde{u}\sqrt{u_2}$, where $\tilde{u}$ remains finite. Let 
\begin{equation}
\label{eq:phix}
\phi(x): =-\frac{1}{2\pi
i}\int_{\mathbb{R}}\frac{e^{-\frac{q^2}{2}+\tilde{u}q}}{q-\frac{1}{x}}dq
\end{equation}
and consider the coefficients of the expansion of
$\phi(x)$ as $x \rightarrow\infty$ in the upper/lower half planes, namely
\begin{equation}\label{eq:phij}
\phi_{\pm,j} :=-\frac{1}{2\pi
i}\int_{E_{\pm}}q^{-j-1}e^{-\frac{q^2}{2}+\tilde{u}q}dq.
\end{equation}
Here $E_+$ ($E_-$) consists of the union of the  two intervals
$(-\infty,-\epsilon)$ and $(\epsilon,\infty)$, oriented from $-\infty$
to $\infty$, together with the semi circle of radius $\epsilon$
around the origin in the upper (lower) half plane. The exact value
of $\epsilon$ is not important, as the integrand is analytic away
from zero; therefore, by Cauchy's theorem, integrals with different
values of $\epsilon$ will yield the same result. 

The residue theorem leads to the identity
\begin{equation}\label{eq:phires}
\sum_{k=0}^{\infty}\phi_{+,k}\frac{u_2^{\frac{k}{2}}}{\zeta^k}
=\sum_{k=0}^{\infty}\phi_{-,k}\frac{u_2^{\frac{k}{2}}}{\zeta^k}
+e^{-\frac{u_2}{2\zeta^2}+\frac{u_1}{\zeta}}.
\end{equation}
This formula gives us a set of relations between the $\phi_{+,j}$'s
and $\phi_{-,j}$'s.  For example, for the first few coefficients we have
\begin{subequations}
\label{eq:phijjump}
\begin{align}
\phi_{+,0}-\phi_{-,0}&=1, & \phi_{+,1}-\phi_{-,1}&=\tilde{u},\\
\phi_{+,2}-\phi_{-,2}&=-\frac{1}{2}+\frac{\tilde{u}^2}{2}, &
\phi_{+,3}-\phi_{-,3}& =-\frac{1}{2}\tilde{u}+\frac{1}{6}\tilde{u}^3.
\end{align}
\end{subequations}
By integrating by parts, we can express the coefficients $\phi_{\pm,j}$
in terms of $\phi_{\pm,0}$.

The integrals~\eqref{eq:phij} are rather cumbersome to compute
directly. The following lemma provides an efficient tool to this purpose.
\begin{lemma}
\label{pro:recur}
The coefficients $\phi_{\pm,j}$  satisfy the recurrence
relation
\begin{equation}
\label{eq:recur}
\phi_{\pm,j}=\frac{1}{j}\left(\tilde{u}\phi_{\pm,j-1}-\phi_{\pm,j-2}\right),
\quad j\geq 2,
\end{equation}
with initial conditions
\begin{subequations}
\label{eq:phi01}
\begin{align}
\phi_{\pm,0}& =\pm\frac{1}{2}+\frac{i}{\sqrt{2\pi}}
\int_0^{\tilde{u}}e^{\frac{q^2}{2}}dq,\\
\phi_{\pm,1}& =\pm\frac{\tilde{u}}{2}+
\frac{i}{\sqrt{2\pi}}\left(\tilde{u}\int_0^{\tilde{u}}e^{\frac{q^2}{2}}dq
-e^{\frac{\tilde{u}^2}{2}}\right).
\end{align}
\end{subequations}
\end{lemma}
\begin{proof} 
The proof of Eq.~\eqref{eq:recur} is an immediate consequence of
integration by parts:
\begin{equation*}
\begin{split}
\phi_{\pm,j}&=-\frac{1}{2\pi i}\int_{E_{\pm}}q^{-j-1}
e^{-\frac{q^2}{2}+\tilde{u}q}dq
=-\frac{1}{2\pi i j}\int_{E_{\pm}}\left(\tilde{u}-q\right)q^{-j}
e^{-\frac{q^2}{2}+\tilde{u}q}dq\\
&={j}^{-1}\left(\tilde{u}\phi_{\pm,j-1}-\phi_{\pm,j-2}\right), \quad
\text{for $j\geq 2$.} 
\end{split}
\end{equation*}
When $j=1$ this formula gives
\begin{equation}\label{eq:pm1}
\phi_{\pm,1}=\tilde{u}\phi_{\pm,0}+\frac{1}{2\pi
i}\int_{\mathbb{R}}e^{-\frac{q^2}{2}+\tilde{u}q}dq=
\tilde{u}\phi_{\pm,0}-\frac{i}{\sqrt{2\pi}}e^{\frac{\tilde{u}^2}{2}}.
\end{equation}
To compute $\phi_{\pm,0}$, we expand the factor $e^{\tilde{u}q}$ in
the integrand in~\eqref{eq:phij} in a power series:
\begin{equation}
\label{eq:pm0}
\begin{split}
\phi_{\pm,0}&=-\frac{1}{2\pi
i}\int_{E_{\pm}}q^{-1}e^{-\frac{q^2}{2}}\sum_{k=0}^{\infty}\frac{\left(\tilde{u}q\right)^k}{k!}dq\\
&=\frac{i}{2\pi}\int_{E_{\pm}}q^{-1}e^{-\frac{q^2}{2}}dq+\frac{i}{2\pi }\sum_{k=1}^{\infty}\frac{\tilde{u}^k}{k!}\int_{\mathbb{R}}q^{k-1}e^{-\frac{q^2}{2}}dq\\
&=\frac{i}{2\pi}\int_{E_{\pm}}q^{-1}e^{-\frac{q^2}{2}}dq+\frac{i}{\sqrt{2\pi}}\sum_{k=0}^{\infty}\frac{\tilde{u}^{2k+1}}{(2k+1)
2^kk!}\\
&=\frac{i}{2\pi}\int_{E_{\pm}}q^{-1}e^{-\frac{q^2}{2}}dq+\frac{i}{\sqrt{2\pi}}\int_0^{\tilde{u}}e^{\frac{q^2}{2}}dq.
\end{split}
\end{equation}
To evaluate the integral
$\int_{E_{\pm}}q^{-1}e^{-\frac{q^2}{2}}dq$, note that if we
change the integration variable from $s$ to $-s$, then the path of
integration will change from $E_{\pm}$ to $-E_{\mp}$.
Therefore, we have
\begin{equation}\label{eq:symme}
\frac{i}{2\pi}\int_{E_{\pm}}q^{-1}e^{-\frac{q^2}{2}}dq=
-\frac{i}{2\pi}\int_{E_{\mp}}q^{-1}e^{-\frac{q^2}{2}}dq.
\end{equation}
Now, the residue theorem gives
\begin{equation*}
\frac{i}{2\pi}\int_{E_{+}
}q^{-1}e^{-\frac{q^2}{2}}dq-\frac{i}{2\pi}
\int_{E_{-}}q^{-1}e^{-\frac{q^2}{2}}dq=1.
\end{equation*}
This equation and Eq.~ \eqref{eq:symme} imply that
$\frac{i}{2\pi}\int_{E_{\pm}}q^{-1}e^{-\frac{q^2}{2}}dq=\pm\frac{1}{2}$.
Inserting this value into \eqref{eq:pm0} and combing it with
\eqref{eq:phi01} completes the proof.
\end{proof}

We are now in a position to state the main result of this section.
\begin{theorem}
\label{thm:init}
  Let $u_1=\tilde{u}u_2^{\frac{1}{2}}$ and suppose $\tilde{u}$
  exists and is finite as $u_2\rightarrow0$. Let the $\phi_{-,j}$'s be
  the coefficients introduced  in Eq.~\eqref{eq:phij}.
  Furthermore, define
  \begin{equation*}
    \chi\left(\tilde{u}\right) :=\frac{i}{\sqrt{2\pi}}
    \int_{0}^{\tilde{u}} e^{-\frac{q^2}{2}}dq
  \end{equation*}
  and denote by $\phi_1$ be the coefficient of $x$ in the expansion
  of $\phi(x)$ in a neighbourhood of $x=0$.  Then, the solution
  $\Phi(\zeta)$ of the RHP~\eqref{eq:RHPphi} determines the following
  initial conditions for the Hamilton equations in Theorem~{\rm
    \ref{thm:main1}}:
\begin{align*}
P_1&=\frac{\phi_1}{2\phi_{-,0}}u_2^{\frac{1}{2}}
-i\left(\frac{\phi_1\phi_{-,1}}{\phi_{-,0}^2}+2\phi_{-,0}+1\right)u_2
+\biggl(\frac{\phi_1(\phi_{-,2}-\phi_{-,1}^2)}{\phi_{-,0}^2}
-\frac{2\phi_1\phi_{-,1}^2}{\phi_{-,0}^3}\\
&\quad -\frac{\phi_{-,1}}{\phi_{-,0}} +4\phi_{-,1}-2\phi_1
+4\phi_{-,1}\phi_{-,0}\biggr)u_2^{\frac{3}{2}}\\
& \quad +2i\biggl(\frac{2\phi_1\phi_{-,1}^3}{\phi_{-,0}^4}
+\frac{\phi_{1}\phi_{-,1}\left(9\phi_{-,1}^2+9\phi_1\phi_{-,1}
-17\phi_{-,2}\right)-\phi_1^2\phi_{-,2}}{9\phi_{-,0}^3}
+ \frac{2\phi_{-,1}\left(3\phi_{-,1}+\phi_{1}\right)}{9\phi_{-,0}^2}\\
& \quad +\frac{10\phi_{-,1}^2-5\phi_1\phi_{-,1}+\phi_1^2}{3\phi_{-,0}} 
+6\phi_{-,1}^2-4\phi_1\phi_{-,1}+\frac{1}{3}+4\phi_{-,1}^2\phi_{-,0}\biggr)u_2^2
+O\left(u_2^{\frac{5}{2}}\right),\\
Q_1&=\frac{\phi_{-,1}}{\phi_{-,0}}u_2^{-\frac{1}{2}}
-2i\left(\frac{\phi_1\phi_{-,1}}{\phi_{-,0}^2}+1\right)
+2\left(\frac{\phi_1(\phi_{-,2}-\phi_{-,1}^2)}{\phi_{-,0}^2}
-\frac{2\phi_1\phi_{-,1}^2}{\phi_{-,0}^3}
-\frac{\phi_{-,1}}{\phi_{-,0}}-2\phi_{-,1}\right)u_2^{\frac{1}{2}}\\
&\quad +i\biggl(\frac{2\phi_1\phi_{-,1}^2\left(19\phi_{-,1}+17\phi_1\right)}%
{9\phi_{-,0}^4} +\frac{4\phi_1\left(9\phi_{-,1}^3
+9\phi_1\phi_{-,1}^2-\phi_1\phi_{-,2}\right)}{9\phi_{-,0}^3}\\
& \quad
+\frac{6\phi_{-,1}\left(4\phi_{-,1}+7\phi_1\right)}{9\phi_{-,0}^2}
+8\frac{\phi_1\phi_{-,1}}{\phi_{-,0}} +\frac{4}{3}
+\frac{8\phi_{-,0}}{3}\biggr)u_2+O\left(u_2^{\frac{3}{2}}\right),\\
P_2&=\chi(\tilde{u})
u_2+i\bigl(4\chi(\tilde{u})^2\tilde{u}-\tilde{u}
+4\chi(\tilde{u})\phi_1\bigr)u_2^{\frac{3}{2}}\\
& \quad + \bigl(\phi_1\tilde{u}+\left(3\tilde{u}^2
-12\phi_1^2\right)\chi(\tilde{u}) -24\phi_1\tilde{u}\chi(\tilde{u})^2
-12\tilde{u}^2\chi(\tilde{u})^3\bigr)u_2^2\\
&  \quad +i\left(4\phi_1^2\tilde{u}-\frac{8\tilde{u}^3}{9}+
\left(\left(\frac{140\tilde{u}^2-28}{9}\right)\phi_1
-32\phi_1^3\right)\chi(\tilde{u}) \right. \\
& \quad +\left(\frac{104\tilde{u}^3}{9}
-96\phi_1^2\tilde{u}\right)\chi(\tilde{u})^2
-96\phi_1\tilde{u}^2\chi(\tilde{u})^3
-32\tilde{u}^3\chi(\tilde{u})^4\biggr)u_2^{\frac{5}{2}}+O\left(u_2^3\right),\\
Q_2&=2i\frac{\phi_{-,1}}{\phi_{-,0}}u_2^{-\frac{1}{2}}
+4\left(\frac{\phi_{-,1}^2}{\phi_{-,0}^2}
+\frac{\phi_{-,1}^2-\phi_{-,2}}{\phi_{-,0}}\right)\\
& \quad -2i\biggl(\frac{3\phi_1\phi_{-,1}^2+\phi_{-,1}^3}{\phi_{-,0}^3}
+6\frac{\phi_1\phi_{-,1}^2}{\phi_{-,0}^2}
+\frac{2\tilde{u}\phi_{-,2}
+7\phi_{-,1}}{3\phi_{-,0}} +4\phi_{-,1}\biggr)u_2^{\frac{1}{2}}\\
& \quad
-\biggl(\frac{\phi_{-,1}^2\left(33\phi_{-,1}^2+79\phi_1\phi_{-,1}
+32\phi_1^2\right)}{9\phi_{-,0}^4} +\frac{6\phi_{-,1}^2
\left(11\phi_{-,1}^2+14\phi_1\phi_{-,1}+23\phi_1^2\right)}{9\phi_{-,0}^3}\\
& \quad - \frac{2\phi_1^2\phi_{-,2}}{9\phi_{-,0}^3}  +\frac{\phi_{-,1}\left(22\phi_{-,1}
+15\phi_1+24\phi_1\phi_{-,1}^2\right)}{3\phi_{-,0}^2}
+\frac{24\phi_{-,1}^2+52\phi_1\phi_{-,1}}{3\phi_{-,0}}\\
& \quad +8\phi_{-,1}^2 +\frac{5 +10\phi_{-,0}}{3}\biggr)u_2
+O\left(u_2^{\frac{3}{2}}\right).
\end{align*}
\end{theorem}

\begin{remark}
These formulae simplify considerably at $\tilde u_1=0$ and become the
initial conditions in Corollary~\ref{co:asym_u2}.  
\end{remark}

\begin{figure}[ht]
\centering
\includegraphics[scale=0.4]{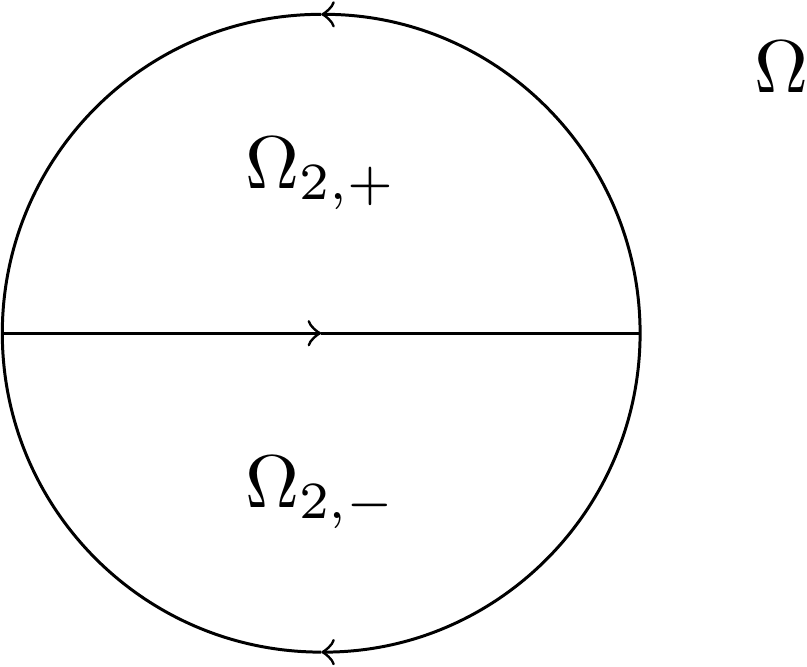}
\caption{The $\zeta$-plane.}
\label{fig:zetareg}
\end{figure}

\subsection{Proof of Theorem~\ref{thm:init}}
\label{ss:pyhminit}

Divide the $\zeta$-plane in the three regions $\Omega_{2,\pm}$ and 
$ \Omega = \mathbb{C}\setminus (\Omega_{2,+} \cup\Omega_{2,-})$ 
as in  Fig.~\ref{fig:zetareg}.   Recall that  $\partial \Omega_{2,\pm} = \Gamma_{\pm} \cup \Gamma_2$ (see Fig.~\ref{fig:local2}). Define  
\begin{equation}
\label{eq:phiout}
\Phi^{(o)}(\zeta) :=
\begin{cases}
I,&  \zeta\in \ \Omega,\\
\begin{pmatrix}1&e^{2i\zeta}\\
-e^{-2i\zeta}&0\end{pmatrix}, & \zeta\in\Omega_{2,+},\\
\begin{pmatrix}1&0\\-e^{-2i\zeta}&1\end{pmatrix}, & 
\zeta\in\Omega_{2,-}.
\end{cases}
\end{equation}
Let $\Delta_0$ be a small disc around the origin of fixed
radius and let  $\xi=u_2^{-\frac{1}{2}}\zeta$.  Then, introduce the function
\begin{equation}
\label{eq:philocal}
\Phi^{(p)}(\zeta) :=
\begin{cases}
\begin{pmatrix}1&e^{2i\zeta}\\-e^{-2i\zeta}&0\end{pmatrix}
\begin{pmatrix}1&\left(\phi(\xi)-\phi_{+,0}\right)e^{2i\zeta}\\
0&1\end{pmatrix}, & \ipart(\zeta)>0,\\
\begin{pmatrix}1&0\\-e^{-2i\zeta}&1\end{pmatrix}
\begin{pmatrix}1&\left(\phi(\xi)-\phi_{-,0}\right)e^{2i\zeta}\\0&1\end{pmatrix},
&  \ipart(\zeta)<0.
\end{cases}
\end{equation}

From Eqs.~\eqref{eq:phijjump}, \eqref{eq:philocal} and the property of
the Cauchy transform, one can check that $\Phi^{(p)}(\zeta)$ satisfies
the jump condition
\begin{equation*}
\Phi^{(p)}_+(\zeta)=\Phi^{(p)}_-(\zeta)
\begin{pmatrix}1&e^{2i\zeta-\frac{u_2}{2\zeta^2}+\frac{u_1}{\zeta}}\\
0&1\end{pmatrix}.
\end{equation*}
Hence, the matrix
\begin{equation}
\label{eq:Rphi}
R_{\Phi}(\zeta):= \begin{cases}
         \Phi(\zeta)\left(\Phi^{\left(p\right)}(\zeta)\right)^{-1}, &
\zeta \in \Delta_0, \\
 \Phi(\zeta)\left(\Phi^{(o)}(\zeta)\right)^{-1}, & 
\zeta \in \mathbb{C}\setminus\Delta_0
       \end{cases}
\end{equation}
 satisfies the RHP
\begin{equation}\label{eq:RHRphi}
\begin{aligned}
1. \quad & \text{$R_{\Phi}(\zeta)$ is analytic in 
$\mathbb{C}\setminus\Gamma_{R_{\Phi}}$,} \\
2.\quad & R_{\Phi,+}(\zeta)=R_{\Phi,-}(\zeta)J_{R_{\Phi}}(\zeta),
 \quad \zeta \in \Gamma_{R_{\Phi}},\\
3. \quad & R_{\Phi}(\zeta)=I+O(\zeta^{-1}), \quad  \zeta \to
\infty,
\end{aligned}
\end{equation}
where $\Gamma_{R_{\Phi}}$ is the contour in Fig.~\ref{fig:GammaRphi}. The jump
matrices are 
\begin{equation}
\label{eq:JRPhi}
J_{R_{\Phi}}(\zeta) :=
\begin{cases}
\begin{pmatrix}e^{-\frac{u_2}{2\zeta^2}+\frac{u_1}{\zeta}}
&e^{2i\zeta}\left(e^{-\frac{u_2}{2\zeta^2}+\frac{u_1}{\zeta}}-1\right)\\
0&e^{\frac{u_2}{2\zeta^2}-\frac{u_1}{\zeta}}\end{pmatrix}, & \zeta\in\Gamma_+,\\
\begin{pmatrix}1&0\\
e^{-2i\zeta}\left(e^{\frac{u_2}{2\zeta^2}-\frac{u_1}{\zeta}}-1\right)&1\end{pmatrix},
 & \zeta\in\Gamma_-,\\
\begin{pmatrix}e^{-\frac{u_2}{2\zeta^2}+\frac{u_1}{\zeta}}
&e^{2i\zeta}\left(e^{-\frac{u_2}{2\zeta^2}+\frac{u_1}{\zeta}}-1\right)\\
e^{-2i\zeta}\left(1-e^{-\frac{u_2}{2\zeta^2}+\frac{u_1}{\zeta}}\right)
&2-e^{-\frac{u_2}{2\zeta^2}+\frac{u_1}{\zeta}}\end{pmatrix},&
\zeta\in\Gamma_{2,-} \cup \Gamma_{2,+},\\
I+\sum_{k=1}^{\infty}\frac{\phi_{\pm,k}u_2^{\frac{k}{2}}}{\zeta^k}
\begin{pmatrix}1&e^{2i\zeta}\\
-e^{-2i\zeta}&-1\end{pmatrix}, &\zeta\in\p\Delta_0\cap\mathbb{C}_{\pm}.
\end{cases}
\end{equation}
\begin{figure}[ht]
\centering 
\includegraphics[scale=0.4]{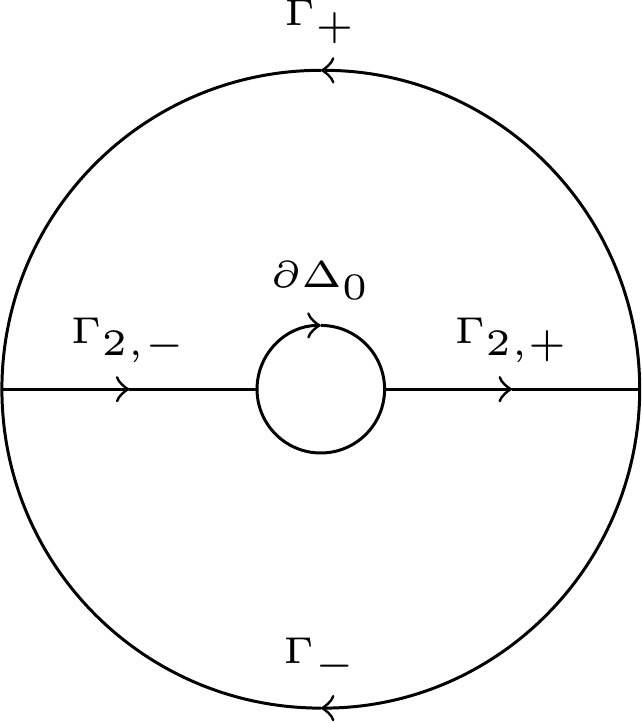}
\caption{The contour $\Gamma_{R_{\Phi}}$.}\label{fig:GammaRphi}
\end{figure}

Note that from the definition~\eqref{eq:Rphi}, \eqref{eq:yzeta}
and~\eqref{Yasym0} we have the relation
\begin{equation}
\label{eq:Phiasym}
\hat{\Phi}_0(\zeta)=R_{\Phi}(\zeta)\Phi^{(p)}(\zeta)e^{i\zeta\sigma_3}.
\end{equation}
One can check that both formulae in Eq.~\eqref{eq:philocal} provide
the correct expression for $\hat{\Phi}_0(\zeta)$.  If we can study the
behaviour of $R_{\Phi}(\zeta)$ as $u_2 \to 0 $, then expanding
$\Phi^{(p)}(\zeta)$ and using Eqs.~\eqref{eq:Phiasym} and~\eqref{eq:A}
give the asymptotic expansions for the matrix elements of $A(\zeta)$.
These, in turn, provide the initial conditions for the time evolution of
the canonical coordinates $P_1$, $Q_1$, $P_2$ and $Q_2$.  

Since $\tilde{u}$ is finite, the jump matrices $J_{R_{\Phi}}$ are of
order $I+O(\sqrt{u_2})$; hence, by the same reasoning adopted in
Secs.~\ref{ss:ptode} and~\ref{se:final}, Eqs.~\eqref{eq:gen_th}
and~\eqref{eq:Rest}, we have
\begin{equation}\label{eq:RPhi}
R_{\Phi}(\zeta)=I+O\left(\frac{\sqrt{u_2}}{\abs{\zeta}+1}\right).
\end{equation}
Indeed, a series expansion in $\sqrt{u_2}$ of the matrix
$R_{\Phi}(\zeta)$ can be computed using the method
in~\cite[Sec.~7.2]{DKMVZ99b}. 

 In order to find the necessary initial
conditions to solve the ODEs in Theorem \ref{thm:main1}, we will need
the expansion of $R_{\Phi}(\zeta)$ up to order
$O\left(u_2^{\frac{3}{2}}\right)$. To do so, let us first simplify the
RHP for  $R_{\Phi}(\zeta)$ with the transformation
\begin{equation*}
\tilde{R}(\zeta) :=
\begin{cases}
R_{\Phi}(\zeta),& \zeta\in
\Omega \cup\Delta_0,\\
R_{\Phi}(\zeta)\begin{pmatrix}e^{\frac{u_2}{2\zeta^2}-\frac{u_1}{\zeta}}&-e^{2i\zeta}\left(e^{-\frac{u_2}{2\zeta^2}+\frac{u_1}{\zeta}}-1\right)
\\0&e^{-\frac{u_2}{2\zeta^2}+\frac{u_1}{\zeta}}\end{pmatrix},&
\zeta\in\Omega_{2,+}\setminus\Delta_0,\\
R_{\Phi}(\zeta)\begin{pmatrix}1&0
\\e^{-2i\zeta}\left(e^{\frac{u_2}{2\zeta^2}-\frac{u_1}{\zeta}}-1\right)&1\end{pmatrix},&
\zeta\in\Omega_{2,-}\setminus\Delta_0.
\end{cases}
\end{equation*}
Then, the matrix $\tilde{R}(\zeta)$ satisfies the RHP
\begin{equation}\label{eq:RHRtilde}
\begin{aligned}
1. \quad & \text{$\tilde{R}(\zeta)$ is analytic in $\mathbb{C}\setminus\p\Delta_0$,} \\
2.\quad & \tilde{R}_{+}(\zeta)=\tilde{R}_{-}(\zeta)J_{\tilde{R}}(\zeta), \quad \zeta \in \p\Delta_0\\
3. \quad & \tilde{R}(\zeta)=I+O(\zeta^{-1}),  \quad \zeta \to
\infty.
\end{aligned}
\end{equation}
It follows from Eq.~\eqref{eq:phires}, that the jump matrix
$J_{\tilde{R}}(\zeta)$ is given by
\begin{equation*}
J_{\tilde{R}}(\zeta) =\begin{pmatrix}1+e^{\frac{u_2}{2\zeta^2}-\frac{\tilde{u}u_2^{\frac{1}{2}}}{\zeta}}\sum_{k=1}^{\infty}\phi_{-,k}\frac{u_2^{\frac{k}{2}}}{\zeta^k}&
e^{2i\zeta}\sum_{k=1}^{\infty}\phi_{-,k}\frac{u_2^{\frac{k}{2}}}{\zeta^k}\\
-e^{-2i\zeta+\frac{u_2}{2\zeta^2}-\frac{\tilde{u}u_2^{\frac{1}{2}}}{\zeta}}\sum_{k=1}^{\infty}\phi_{+,k}\frac{u_2^{\frac{k}{2}}}{\zeta^k}&
1-\sum_{k=1}^{\infty}\phi_{-,k}\frac{u_2^{\frac{k}{2}}}{\zeta^k}
\end{pmatrix}.
\end{equation*}
Note that $J_{\tilde{R}}(\zeta)$ is of order $I+O(\sqrt{u_2})$.
Consider the expansions
\begin{subequations}
\label{eq:ex_RJ}
\begin{align}
\label{eq:ex_R}
\tilde{R}(\zeta)& =I+\sum_{k=1}^{\infty}u_2^{\frac{k}{2}}\tilde{R}_{k}(\zeta),\\
J_{\tilde{R}}(\zeta)& =I+\sum_{k=1}^{\infty}
u_2^{\frac{k}{2}}J_{\tilde{R},k}(\zeta).
\end{align}
\end{subequations}
By comparing the coefficients of $u_2^{\frac{k}{2}}$ in the jump
conditions of $\tilde{R}(\zeta)$, we obtain linear jump
conditions for the coefficients $\tilde{R}_k(\zeta)$:
\begin{subequations}
\label{eq:jumpcoef}
\begin{align}
\tilde{R}_{+,1}&=J_{R,1}+\tilde{R}_{1,-},\\
\tilde{R}_{+,2}& =\left(\tilde{R}_{-,1}J_{\tilde{R},1}+J_{\tilde{R},2}\right)+\tilde{R}_{-,2},\\
\tilde{R}_{+,3}&=\left(\tilde{R}_{-,1}J_{\tilde{R},2}+\tilde{R}_{-,1}J_{\tilde{R},2}
+J_{\tilde{R},3}\right)+\tilde{R}_{-,3}.
\end{align}
\end{subequations}
These relations and the requirement that $\tilde{R}_k=O(\zeta^{-1})$ 
as $\zeta\rightarrow\infty$ define RHPs for the $\tilde{R}_k$'s that
we can solve.

From Eqs.~\eqref{eq:JRPhi} and \eqref{eq:phijjump}, we see that
$J_{\tilde{R},1}$ is given by
\begin{equation*}
\begin{split}
J_{\tilde{R},1}&=\frac{1}{\zeta}\begin{pmatrix}\phi_{-,1}&\phi_{-,1}e^{2i\zeta}\\-\phi_{+,1}e^{-2i\zeta}&-\phi_{-,1}\end{pmatrix},\quad
\zeta\in\p\Delta_0.
\end{split}
\end{equation*}
Let us write 
\begin{equation}
\label{eq:Ecal}
\mathcal{E}_{j}^{\pm}(\zeta): =\zeta^{-j-1}\left(\sum_{k=0}^j\frac{\left(
\pm 2i\zeta\right)^k}{k!}-e^{\pm 2i\zeta}\right)=\textrm{$O(1)$ in
$\zeta$}.
\end{equation}
By using~\eqref{eq:phijjump} and~\eqref{eq:jumpcoef}, one can 
verify that the solution to the RHP for $\tilde R_1$ is 
\begin{equation*}
%\label{eq:R1}
\tilde{R}_1=
\begin{cases}
\begin{pmatrix}0&\phi_{-,1}\mathcal{E}_{0}^+(\zeta)\\
-\phi_{+,1}\mathcal{E}_{0}^-(\zeta)&0\end{pmatrix}, &
\zeta\in\Delta_0,\\
\begin{pmatrix}\frac{\phi_{-,1}}{\zeta}&\frac{\phi_{-,1}}{\zeta}\\
-\frac{\phi_{+,1}}{\zeta}&-\frac{\phi_{-,1}}{\zeta}\end{pmatrix}, &
\zeta\in\mathbb{C}\setminus\Delta_0.
\end{cases}
\end{equation*}

Let us now denote the jump matrix of the RHP for $\tilde R_2$ by
$\hat{J}_{\tilde{R},2}:  =\tilde{R}_{-,1}J_{\tilde{R},1}+J_{\tilde{R},2}$.
Then, we have
\begin{equation*}
\begin{split}
\hat{J}_{\tilde{R},2}&=\frac{1}{\zeta^2}\begin{pmatrix}\phi_{-,1}\left(\phi_{-,1}-\phi_{+,1}e^{-2i\zeta}\right)+\phi_{-,2}
&\phi_{-,1}^2\left(e^{2i\zeta}-1\right)+\phi_{-,2}e^{2i\zeta}\\
\phi_{+,1}\left(\phi_{+,1}e^{-2i\zeta}-\phi_{-,1}\right)-\phi_{+,2}e^{-2i\zeta}&
\phi_{+,1}\phi_{-,1}\left(1-e^{2i\zeta}\right)-\phi_{-,2}
\end{pmatrix}.
\end{split}
\end{equation*}
This expression  gives
\begin{equation*}
\label{eq:R2}
\tilde{R}_2 =
\begin{cases}
\begin{pmatrix}-\phi_{+,1}\phi_{-,1}\mathcal{E}_1^-(\zeta)&
\left(\phi_{-,1}^2+\phi_{-,2}\right)\mathcal{E}_1^+(\zeta)\\
\left(\phi_{+,1}^2-\phi_{+,2}\right)\mathcal{E}_1^-(\zeta)&
-\phi_{+,1}\phi_{-,1}\mathcal{E}_1^+(\zeta)\end{pmatrix}, &
\zeta\in\Delta_0,\\
\frac{1}{\zeta^2}\begin{pmatrix}\phi_{-,2}-\tilde{u}
\phi_{-,1}+2i\phi_{+,1}\phi_{-,1}\zeta&
\phi_{-,2}+2i\left(\phi_{-,1}^2+\phi_{-,2}\right)\zeta\\
\tilde{u}\phi_{+,1}-\phi_{+,2}+2i\left(\phi_{+,2}-\phi_{+,1}^2\right)\zeta&
-\phi_{-,2}-2i\phi_{+,1}\phi_{-,1}\zeta\end{pmatrix}, & 
\zeta\in\mathbb{C}\setminus\Delta_0.
\end{cases}
\end{equation*}

Proceeding in the same way, we can compute the matrix elements of the
jump matrix 
\begin{equation*}
\hat{J}_{\tilde{R},3}=\tilde{R}_{-,1}J_{\tilde{R},2}+\tilde{R}_{-,1}J_{\tilde{R},2}
+J_{\tilde{R},3}.
\end{equation*}
We obtain
\begin{align*}
\left(\hat J_{\tilde{R},3}\right)_{11}&=\zeta^{-3}\left(\iota_{-,0}
+2i\zeta\phi_{+,1}\phi_{-,1}^2+e^{-2i\zeta}\alpha_0-2i\zeta
e^{-2i\zeta}\phi_{+,1}\alpha_{-,1}\right),\\
\left(\hat J_{\tilde{R},3}\right)_{12}&=\zeta^{-3}\Bigl(\phi_{-,1}
\left(\tilde{u}\phi_{-,1}-2\phi_{-,2}\right)
-2i\zeta\phi_{-,1}\alpha_{-,1}+e^{2i\zeta}\iota_{-,0}\\
 &\quad +2i\zeta
e^{2i\zeta}\phi_{+,1}\phi_{-,1}^2\Bigr),\\
\left(\hat J_{\tilde{R},3}\right)_{21}&=\zeta^{-3}\left(\alpha_0+
2i\zeta\phi_{-,1}\alpha_{+,1}+ e^{-2i\zeta}\iota_{+,0}
+2i\zeta e^{-2i\zeta}\phi_{+,1}^2\phi_{-,1}\right),\\
\left(\hat J_{\tilde{R},3}\right)_{22}&=\zeta^{-3}\left(-\alpha_0
-\phi_{-,3}+2i\zeta \phi_{+,1}\phi_{-,1}^2+e^{2i\zeta}\alpha_0
+2i\zeta e^{2i\zeta}\phi_{-,1}\alpha_{+,1}\right),
\end{align*}
where $\alpha_0$, $\alpha_{\pm,1}$ and $\iota_{\pm,0}$ are  given by
\begin{align*}
\alpha_0&:=\tilde{u}\phi_{+,1}\phi_{-,1}-\phi_{-,1}\phi_{+,2}
-\phi_{+,1}\phi_{-,2},&
\alpha_{\pm,1} &:=\phi_{\pm,2}\mp\phi_{\pm,1}^2,\\
\iota_{\pm,0}&:=2\phi_{\pm,1}\phi_{\pm,2}\mp\phi_{\pm,3}-\tilde{u}\phi_{\pm,1}^2.
& & 
\end{align*}
Inside $\Delta_0$ the matrix $\tilde{R}_3$ is
\begin{equation}\label{eq:R3}
\begin{split}
\tilde{R}_3=\begin{pmatrix}
\alpha_0\mathcal{E}_2^-(\zeta)-2i\phi_{+,1}\alpha_{-,1}\mathcal{E}_1^-(\zeta)&
\iota_{-,0}
\mathcal{E}_2^+(\zeta)+2i\phi_{+,1}\phi_{-,1}^2\mathcal{E}_1^+(\zeta)\\
\iota_{+,0}\mathcal{E}_2^-(\zeta)
+2i\phi_{+,1}^2\phi_{-,1}\mathcal{E}_1^-(\zeta)&
\alpha_0\mathcal{E}_2^+(\zeta)+2i\phi_{-,1}\alpha_{+,1}\mathcal{E}_1^+(\zeta)
\end{pmatrix}.
\end{split}
\end{equation}
The coefficient $\tilde{R}_3$ outside of $\Delta_0$ is obtained by
adding  $\hat J_{\tilde{R},3}$ to formula~\eqref{eq:R3}.

We need to compute one extra term in the expansion of $\tilde
R(\zeta)$, which is used to derive the asymptotics of $Q_2$.    By
iterating this algorithm we arrive at
\begin{equation*}
\tilde{R}_4 =\begin{pmatrix}\alpha_{+,4}\mathcal{E}^-_3
-2i\alpha_{+,3}\mathcal{E}^-_2
+\alpha_{+,2}\mathcal{E}^-_1&\iota_{+,2}\mathcal{E}^+_3
-2i\alpha_0\phi_{-,1}\mathcal{E}^+_2+\iota_{+,1}\mathcal{E}^+_1\\
\iota_{-,2}\mathcal{E}^-_3-2i\alpha_0\phi_{+,1}\mathcal{E}^-_2+\iota_{-,1}\mathcal{E}^-_1&
-\alpha_{-,4}\mathcal{E}^+_3+2i\alpha_{-,3}\mathcal{E}^+_2
+\alpha_{-,2}\mathcal{E}^+_1
\end{pmatrix}, \quad \zeta \in \Delta_0,
\end{equation*}
where
\begin{align*}
\alpha_{\pm,2}&:=2(2\phi_{-,1}^2\phi_{+,1}^2\pm\iota_{\mp,0}\phi_{\pm,1}),\\
\alpha_{\pm,3}&:=\alpha_{+,1}\alpha_{-,1}
+\phi_{+,1}^2\phi_{-,1}^2\pm\phi_{\pm,1}\iota_{\mp,0},\\
\alpha_{\pm,4}&:=\phi_{\mp,1}\left(\phi_{\pm,1}\phi_{\pm,2}
\mp\phi_{\pm,3}\right)-\phi_{\pm,1}(\iota_{\mp,0}
+\alpha_0)\mp\alpha_{\mp,1}\phi_{\pm,2} \\
\iota_{\pm,1}&:=\mp
2(2\alpha_{\mp,1}\phi_{\pm,1}+\alpha_0)\phi_{\mp,1},\\
\iota_{\pm,2}&:=\pm\phi_{\mp,4}+\phi_{\mp,2}
\left(\phi_{\mp,2}-\tilde{u_1}\phi_{\mp,1}\right)+
\phi_{\mp,1}\bigl(\phi_{\mp,3}\pm(\iota_{\mp,0}+\alpha_0)\bigr).
\end{align*}
The expression of $\tilde{R}_4$ outside $\Delta_0$ is not needed and
will not be computed. 

By inserting the coefficients $\tilde R_1$, $\tilde R_2$, $\tilde R_3$
and $\tilde R_4$ into the expansion~\eqref{eq:ex_R} and then using
formula~\eqref{eq:Phiasym}, we arrive after long and tedious
calculations to the asymptotics of $P_1$, $Q_1$, $P_2$ and $Q_2$
at the desired order in $u_2$.  Such calculations can be performed
using a computer algebra package like \texttt{MAPLE}.

\section{Reduction to PIII}
\label{se:red}

This final section is devoted to express the ensemble average
$G_N(u_{1,N},u_{2,N})$ in \eqref{eq:tildeB} in terms of a special
solution to the PIII equation, thereby proving
Theorem~\ref{thm:main2}.

In Sec.~\ref{se:st_res} we pointed out that since the partition
function $E_N(z,t)$ defined in~\eqref{eq:main_average} is real
analytic in a neighbourhood of $t=0$, its analysis is equivalent to
study the coefficients of its Taylor expansion.  This is tantamount to
solve the projection of the Hamiltonian system~\eqref{eq:hameq} at
$u_1=0$, which in turn can be reduced to the second order system of
ODEs~\eqref{eq:2ndorder}. The purpose of this section is to prove that
such a system is equivalent to a special case of the Painlev\'e III
equation. In order to achieve this goal we need a result by Chen and
Its~\cite{CI10}, which we now present in some detail.

Consider the system of monic orthogonal polynomials $P_n(x)$  with
weight
\begin{equation}
\label{eq:its_w2}
W_{\alpha}(x)=\exp\left(-\frac{s}{x}-x\right)x^{\alpha}, \qquad x\in \mathbb{R}_+,
\end{equation}
where $\alpha >-1$ and $s>0$.  They satisfy the orthogonality conditions
\begin{equation*}
\int_{\mathbb{R}_+}P_n(x)P_m(x)W_{\alpha}(x)dx=h_n^{(p)}\delta_{mn}.
\end{equation*}
Then, introduce the matrix (c.f. \eqref{eq:RH_sol})
\begin{equation}
 \label{eq:RHX}
   X(x) :=
   \begin{pmatrix}
   P_n(x) & \frac{1}{2\pi i}
    \int_{-\infty}^\infty\frac{P_n(q)W_{\alpha}(q)}{q-x}dq  \\
     \kappa_{n-1}^{(p)}P_{n-1}(x) &
     \frac{\kappa_{n-1}^{(p)}}{2\pi i}
      \int_{-\infty}^\infty\frac{P_{n-1}(q)W_{\alpha}(q)}{q-x}dq
\end{pmatrix},
\end{equation}
where $\kappa_{n-1}^{(p)}=-2\pi i/h_{n-1}^{(p)}$. The function $X(x)$
solves the RHP
\begin{equation}
  \label{eq:RHPX}
\begin{aligned}
1. \quad  & \text{$X(x)$ is analytic in $\mathbb{C}/\mathbb{R}_+$},\\
2. \quad & X_+(x)=X_-(x)\begin{pmatrix} 1 & W_{\alpha}(x)  \\ 0 & 1
\end{pmatrix},\quad x\in\mathbb{R}_+,\\
3. \quad  & X(x)=\left(I+O(x^{-1})\right)\begin{pmatrix} x^n & 0
\\ 0 & x^{-n} \end{pmatrix}, \quad x\rightarrow\infty,\\
4. \quad &X(x)=O(1),\quad  x\rightarrow 0.
\end{aligned}
\end{equation}

Using standard arguments in the theory of integrable systems, similar
to those that we outlined is Sec.~\ref{ss:ptode}, Chen and
Its~\cite[Sec.~5]{CI10} showed that the Lax pair associated to $X(x)$
is that for the PIII equation. More precisely, in our
notation their result is the following.
\begin{theorem}[Chen and Its~\cite{CI10}]
\label{pro:CI} 
Define the function
% \footnote{For simplicity, and in order to adhere 
%   to well established conventions, in this section only, we use 
%    $t$ to denote the argument
%   of the Painlev\'e transcendent.  It should not be confused with the
%   name of the parameter that appear in the partition function
%   $E_N(z,t)$ in Eq.~\eqref{eq:main_average}}
\begin{equation}
\label{eq:CIPIII}
v_X(\omega) := \frac{h_n^{(p)}}{2\pi i \omega d_n}, \qquad 
\omega = \sqrt{s},\qquad d_n = X_{11}(0)X_{12}(0).
\end{equation}
Then, $v_X(\omega)$ satisfies the PIII equation
\begin{equation}
\label{eq:PIII_IC}
v_X^{\prime\prime}=\frac{\left(v_X^{\prime}\right)^2}{v_X}
-\frac{v^{\prime}}{\omega}-\frac{1}{\omega}\left(4\alpha
v^2_X+4(2n+1+\alpha)\right)+4 v^3_X-\frac{4}{v_X}.
\end{equation}
\end{theorem}

When $u_{1,N}=0$ the system of monic polynomials orthogonal with
respect to the weight~\eqref{eq:rescaled_weight} can be mapped to the
polynomials $P_n(x)$ with $\alpha=\pm1/2$ by a change of variables.  In
what follows, we shall use this observation to relate $X(x)$ to the
solution $Y(y)$ of the RHP~\eqref{eq:RHP}. Upon taking the double scaling 
limit, the matrix $Y(y)$ can be expressed in terms of a solution of the Hamiltonian
 system~\eqref{eq:2ndorder}.  
\begin{proposition}
\label{pro:PIII}
Set $r=\sqrt{u_2}$ and let $v_Y(r)$ be a solution of the PIII equation
\begin{equation}
\label{eq:PIII_vy}
v_Y^{\prime\prime} =  \frac{(v_Y^\prime)^2}{v_Y} -
\frac{v_Y^\prime}{r} +(-1)^N\frac{v_Y^2}{r}-\frac{2}{r}+ v_Y^3,
\end{equation}
with initial conditions
\begin{equation}
\label{eq:vasymeven2}
v_Y(r)=\sqrt{\frac{\pi}{2}}
+\frac{\pi-4}{2}r+\frac{\pi^2-4\pi+4}{2^{\frac{3}{2}}
\sqrt{\pi}}r^2+O\left(r^2\right),\quad
r\rightarrow 0,
\end{equation}
when $N$ is even, and 
\begin{equation}
\label{eq:vasymoddr2}
v_Y(r)=\frac{1}{r}+
\sqrt{\frac{2}{\pi}}-\frac{2\pi-6}{3\pi}r+O(r^2),\quad
r\rightarrow 0,
\end{equation}
if $N$ is odd. Then, the trajectory
\begin{equation*}
\begin{split}
P_1& =  \frac{ir}{ 2v_Y(r)} +
\frac{iu_2v_Y^2(r)}{4} -(-1)^N \frac{iu_2 
v_Y^\prime(r)}{4},  \\
Q_1& = \frac{(-1)^Ni}{rv_Y(r)} -(-1)^N
\frac{iv_Y^2(r)}{2} + \frac{iv_Y^\prime(r)}{2},
\end{split}
\end{equation*}
solve the Hamilton equations~\eqref{eq:2ndorder}.
\end{proposition}
This is the main result of this section and completes the proof of Theorem \ref{thm:main2}.  Proposition~\ref{pro:PIII} is proved in Secs. \ref{ss:relXY} and \ref{sse:rel_PIII}.

Note that, although both ODEs in Theorem \ref{pro:CI} and Proposition \ref{pro:PIII} belong
 to the family of PIII equations, they are different in
that they have different parameters.

\subsection{Relation Between $X(x)$ and $Y(y)$}
\label{ss:relXY}
Here we derive the relations between the entries of the matrix
$X(x)$ in \eqref{eq:RHX} and those of $Y(y)$ in \eqref{eq:RH_sol}.
We will consider the case where $N$ is odd and even separately.

\subsubsection{$N$ Even}

Let $\pi_j(y)$ be the orthogonal polynomials for the weight $w_N(y)$
in~\eqref{eq:rescaled_weight} and $P_n(x)$ be the orthogonal
polynomials for the weight~\eqref{eq:its_w2} with $\alpha =
-\frac{1}{2}$.

Notice that $w_N(y) = w_N(-y)$ so by the orthogonality conditions,
\[
\int_\mathbb{R}  \pi_j(y)y^j w_N(y) \, dy = (-1)^j \int_\mathbb{R} 
\pi_j(-y)y^j w_N(y) \, dy \neq 0;
\]
therefore,
\begin{equation}\label{eq:sym}
\pi_j(-y)=(-1)^j\pi_j(y).
\end{equation}

From \eqref{eq:sym}, we see that $\pi_N(y)=p_{\frac{N}{2}}(y^2)$, where $p_{\frac{N}{2}}$ is a polynomial of degree $N/2$. Let us
make the change of variables
\begin{equation}
\label{eq:transformation}
y^2 = \frac{2x}{N},\qquad s = \frac{u_{2,N}}{4N},\qquad n=\frac{N}{2}.
\end{equation}
As the difference between $u_{2,N}$ and its limit $u_2$ does not
affect the asymptotic results, we shall replace $u_{2,N}$ by $u_2$.

By the orthogonality of $\pi_N(y)$, we can see that (omitting the constant)%$w_N(x) \, dx = w(z) \,\frac{dz}{\sqrt{2N}}$
\begin{equation*}
\int_{\mathbb{R}_+}p_{\frac{N}{2}}(x)x^{j}W_{-\frac12}(x)dx=0,\quad j=0,\dots
,\frac{N}{2}-1.
\end{equation*}
So in fact, $p_{\frac{N}{2}}(x)$ is a degree $\frac{N}{2}$
polynomial orthogonal with respect to the weight $W_{-\frac12}(x)$, with
leading coefficient $\left( \frac{2}{N}\right)^{\frac{N}{2}}$. Hence,
the monic polynomials orthogonal with respect to the weight $W_{-\frac12}(x)$
are
$P_{\frac{N}{2}}(x)=\left(\frac{N}{2}\right)^{\frac{N}{2}}p_{\frac{N}{2}}(x)$.
Therefore,
\begin{equation}\label{Y11 X11}
X_{11}(x) = \left(\frac{N}{2}\right)^\frac{N}{2}Y_{11}(y).
\end{equation}

Now, let us look at $X_{12}(0)$ and $Y_{12}(0)$. As $N$ is even,
\begin{equation*} 
Y_{12}(0) = \frac{1}{2\pi i} \int_{\mathbb{R}} \frac{\pi_N(q)w_N(q)}{q} dq = 0,
\end{equation*}
since it is simply the integration of an odd function over a
symmetric interval. But if we expand $Y_{12}(y)$ near the origin, we see that
\begin{equation*}
Y_{12}(y)=\frac{1}{2\pi
i}\int_{\mathbb{R}}\frac{\pi_{N}(q)w_N(q)}{q-y}dq= \frac{1}{2\pi
i}\sum_{j=1}^{\infty}y^{j-1}\int_{\mathbb{R}}\frac{\pi_{N}(q)w_N(q)}{q^j}dq.
\end{equation*}
The term with $j=2$ in the above is given by
\begin{equation*}
\frac{y}{2\pi i}\int_{\mathbb{R}}\frac{\pi_{N}(q)w_N(q)}{q^2}dq.
\end{equation*}
We can make the same change of variables, $q^2 = 2q_1/N$, to
obtain
\begin{equation*}
\frac{1}{2\pi i}\int_{\mathbb{R}}\frac{\pi_{N}(q)w_N(q)}{q^2}dq =
\frac{1}{2}\left(\frac{N}{2}\right)^{-\frac{(N-1)}{2}}X_{12}(0).
\end{equation*}
Expanding $Y(y)$ near the origin, we can express $X_{12}(0)$ in terms of
$Y_{12}^{\prime}(0)$, \textit{i.e.}
\begin{equation}\label{Y12 X12 even}
X_{12}(0) = 2 \left( \frac{N}{2}\right)^\frac{N-1}{2}
Y^\prime_{12}(0).
\end{equation}

In order to complete the mapping $Y(y) \mapsto X(x)$, we are only left to determine
the relation between the orthogonality constants.  We have
\begin{equation*}
\int_{\mathbb{R}_+}P^2_{N/2}(x)W_{-\frac12}(x)dx=h_{N/2}^{(p)},\qquad
\int_{\mathbb{R}}\pi_N^2(y)w_N(y)dy=h_N .
\end{equation*}
Finally, the change of variable $y^2=2x/N$ gives
\[
\int_{\mathbb{R}}\pi_N^2(y)w_N(y)dy = \left( \frac{N}{2}\right)^{-N}
\frac{1}{\sqrt{2N}} \int_{\mathbb{R}_+} P^2_{\frac{N}{2}}(x)W_{-\frac12}(x)dx,
\] 
and so
\begin{equation}\label{hn kn even}
h_{N/2}^{(p)} = \left( \frac{N}{2}\right)^N\sqrt{2N}h_N.
\end{equation}

\subsubsection{$N$ Odd} 

When $N$ is odd, we set $\alpha=1/2$ and replace $n= (N-1)/2$ in
the transformation~\eqref{eq:transformation}. Thus, we have
\begin{equation}
\label{piPodd}
\pi_N(y) = yp_{\frac{N-1}{2}}(y^2),
\end{equation}
where $p_{\frac{N-1}{2}}$ is a polynomial of degree $(N-1)/2$. The
orthogonality conditions for $\pi_N(y)$ now give
\[
\int_{\mathbb{R}_+}  p_{\frac{N-1}{2}}(x) x^{j}
W_{\frac12}(x)\,dx = 0 
\qquad j = 0, \dots, \frac{N-3}{2}.
\]
As previously, the monic polynomials orthogonal with respect to $W_{\frac12}(x)$ are
\[
P_{\frac{N-1}{2}}(x)=\left(\frac{N}{2}\right)^{\frac{N-1}{2}}
p_{\frac{N-1}{2}}(x).
\]

The relation between $Y_{11}(y)$ and $X_{11}(x)$ follows immediately:
\begin{equation}\label{eq:Y X odd}
Y_{11}(y) =
y\left(\frac{N}{2}\right)^{-\frac{N-1}{2}}, \qquad P_{\frac{N-1}{2}}(x) =
y\left(\frac{N}{2}\right)^{-\frac{N-1}{2}}X_{11}(x).
\end{equation}
Differentiating both sides with respect to $y$ gives 
\begin{equation}\label{Y11 X11 odd1}
Y_{11}^\prime(y) =
\left(\frac{N}{2}\right)^{-\frac{N-1}{2}}X_{11}(x) +
\frac{Ny^2}{2}\left(\frac{N}{2}\right)^{-\frac{N-1}{2}}\frac{d}{dx}X_{11}(x).
\end{equation}
Therefore,
\begin{equation}\label{Y11 X11 odd}
X_{11}(0) = \left(\frac{N}{2}\right)^\frac{N-1}{2}Y_{11}^\prime(0).
\end{equation}

When $N$ is odd the integral
\begin{equation*}
 Y_{12}(0) = \frac{1}{2\pi
i}\int_\mathbb{R} \frac{\pi_N(q)w_N(q)}{q}dq
\end{equation*}
is different from zero. Equation~\eqref{piPodd} and the change of
variables $q^2 = 2q_1/N$ give
\begin{equation}\label{eq:Y12 X12 odd}
X_{12}(0) =\sqrt{2N} \left(\frac{N}{2}\right)^{ \frac{N-1}{2}}Y_{12}(0).
\end{equation}

The relation between the orthogonality constants can be found in the
same way as for $N$ even; we obtain
\begin{equation} \label{hn kn odd}
h_{(N-1)/2}^{(p)} = \left( \frac{N}{2}\right)^N \sqrt{2N}h_N.
\end{equation}
\subsection{Relation to PIII}
\label{sse:rel_PIII}
Now that we have obtained the relations between entries in $X(x)$ and
$Y(y)$, we can express the elements of $Y(y)$ in terms of $v_X(\omega)$.  
In turn,  as $N \to \infty$ and $N^{\frac12}z\to \sqrt{u_2}$ such relations will allow us to 
derive formulae for the canonical coordinates $P_1$ and $Q_1$ involving the PIII transcendent $v_Y(r)$.
This completes the proof of Theorem~\ref{thm:main2}.
% Before we discuss the proof of Proposition~\ref{pro:PIII} we need some
% preliminary results.  

From the relations between the entries of $X(x)$ and $Y(y)$ in
Eqs.~\eqref{Y11 X11}, \eqref{Y12 X12 even}, \eqref{hn kn even},
\eqref{eq:Y X odd}, \eqref{eq:Y12 X12 odd} and \eqref{hn kn odd}, we
obtain  
\begin{equation} 
\label{v in Y 1}
v_X(\omega) = \begin{cases}  \frac{\sqrt{N}h_N}{2\pi i \sqrt{u_2} 
Y_{11}\p_{\zeta}Y_{12}}\bigl(1+O(N^{-1})\bigr), &  \text{$N$ even,}\\
\frac{\sqrt{N}h_N}{2\pi i \sqrt{u_2} \p_{\zeta}Y_{11}Y_{12}}
\bigl(1+O(N^{-1})\bigr), & \text{$N$ odd,}
        \end{cases}
\end{equation}
where 
\begin{equation}
\label{eq:tu2_rel}
\omega = \frac{1}{2}\sqrt{\frac{u_2}{N}}.
\end{equation}
The entries $Y_{11}$, $Y_{12}$ and $\p_{\zeta}Y_{12}$ in Eq.~\eqref{v
  in Y 1} are evaluated at $u_1=y=0$ and are functions of $\omega$
through~\eqref{eq:tu2_rel}.  

Since, we are interested in the large $N$ behaviour of $v_X(\omega)$, it is more
convenient to replace $h_N$ with its asymptotic limit. From
Eqs.~\eqref{eq:para}, \eqref{eq:approxS}, \eqref{eq:S} and \eqref{eq:Tg}, we see
that $S^\infty_{12}(y) = i/y + O(y^{-3})$, which leads to
\[
Y_{-1,12} = ie^{Nl}\left(1 + O(N^{-1})\right), \qquad N \to \infty,
\]
where $Y_{-1,12}$ is the coefficient of $y^{{-N-1}}$ in $Y_{12}(y)$ and $l$ is the constant appearing in the
inequalitiess~\eqref{eq:ineq}. Therefore, we have
\begin{equation*}%x\label{kN in N}
h_{N} = 2\pi e^{Nl}\left(1 + O(N^{-1})\right), \qquad N \to \infty.
\end{equation*}
Inserting this formula into~\eqref{v in Y 1} gives
\begin{equation}
\label{v in Y 2}
v_X(\omega) = \begin{cases}
\frac{\sqrt{N}e^{Nl}}{ i \sqrt{u_2} Y_{11}\p_{\zeta}Y_{12}}
\bigl(1 + O(N^{-\frac{1}{2}})\bigr), 
&  \text{ $N$ even},\\
\frac{\sqrt{N}e^{Nl}}{ i \sqrt{u_2}
\p_{\zeta}Y_{11}Y_{12}}\bigl(1 + O(N^{-\frac{1}{2}})\bigr), & \text{ $N$  odd}.
\end{cases}
\end{equation}
Thus, we see that as $N \to \infty$,  $v_X(\omega) = O\left(\sqrt{N}\right)$ uniformly.

Equation~\eqref{v in Y 2} expresses a solution of the ODE~\eqref{eq:PIII_IC} in
terms of the matrix elements of $Y(y)$.  In order to complete the proof of
Proposition \ref{pro:PIII}, we need to go further and study how the connection
between $Y(y)$ and the PIII equation changes as $N^{\frac{1}{2}}z \to
\sqrt{u_2}$.  This double scaling limit is encoded in the behaviour near the origin of
the function $\hat{\Psi}_0(\zeta)$.  Recall that $\hat{\Psi}_0(\zeta)$ was
defined in~\eqref{Yasym0}, while its expansions at $\zeta=0$ and $\zeta=\infty$
in Eqs.~\eqref{eq:psihat}.  The rest of the proof consists of three parts:
firstly, we will express the matrix elements $Y_{11}(0)$, $Y_{12}(0)$ and their
derivatives $\p_{\zeta}Y_{11}(0)$, $\p_\zeta Y_{12}(0)$ in terms of the first two
coefficients, $\Psi_0^{(0)}$ and $\Psi_1^{(0)}$, of the Taylor series of
$\hat{\Psi}_0(\zeta)$ at the origin; secondly, we will write the canonical
coordinates $P_1$ and $Q_1$ as functions of $\Psi_0^{(0)}$ and $\Psi_1^{(0)}$;
finally, eliminating these coefficients from such formulae and Eq.~\eqref{v in Y
  2} lead to the statement of Proposition~\ref{pro:PIII}. The first two steps
are summarized in the next lemma.

\begin{lemma}
\label{lem:can coords}
  Consider the expansion 
\[
\hat{\Psi}_0(\zeta) = \Psi_0^{(0)} +
  \Psi_1^{(0)} \zeta + O(\zeta^2).
\]
When $u_1=0$, the quantities $Y_{11}(0)$, $Y_{12}(0)$, $\p_\zeta Y_{11}(0)$,
$\p_{\zeta}Y_{12}(0)$ have the representations
\begin{subequations} 
\label{Y11 Y12 and derivs}
\begin{align}
Y_{11}(0) &= \frac{e^\frac{Nl}{2}}{\sqrt{2}}\left( \Psi_{0,11}^{(0)} e^{- i \frac{N\pi}{2}} - \Psi_{0,21}^{(0)}e^{i \frac{N\pi}{2}}\right)+O(N^{-1}), \\
\p_{\zeta}Y_{11}(0) &= \frac{e^\frac{Nl}{2}}{\sqrt{2}} \left(\Psi_{1,11}^{(0)}e^{-i \frac{N\pi}{2}} -\Psi_{1,21}^{(0)}e^{i \frac{N\pi}{2}}\right)+O(N^{-1}), \\
Y_{12}(0) &= \frac{e^{\frac{Nl}{2}}}{\sqrt{2}}\left( \Psi_{0,12}^{(0)}e^{- i \frac{N\pi}{2}} - \Psi_{0,22}^{(0)}e^{i \frac{N\pi}{2}}\right)+O(N^{-1}), \\
\p_{\zeta}Y_{12}(0) &=
\frac{e^{\frac{Nl}{2}}}{\sqrt{2}}\left(\Psi_{1,12}^{(0)}e^{-i
\frac{N\pi}{2}} - \Psi_{1,22}^{(0)}e^{i
\frac{N\pi}{2}}\right)+O(N^{-1}).
\end{align}
\end{subequations}
In addition, the canonical coordinates $P_1$
  and $Q_1$ can be written as
\begin{subequations}
\label{eq:can_Psi}
\begin{align}
P_1 &= \frac{u_2}{2}\left(\frac{\Psi_{1,11}^{(0)}}{2\Psi_{0,11}^{(0)}} - \frac{\Psi_{1,12}^{(0)}}{2\Psi_{0,12}^{(0)}}+ \Psi_{1,21}^{(0)}\Psi_{0,12}^{(0)} + \Psi_{1,22}^{(0)}\Psi_{0,11}^{(0)}\right), \\
Q_1 &= -\frac{\Psi_{1,11}^{(0)}}{2\Psi_{0,11}^{(0)}} -
\frac{\Psi_{1,12}^{(0)}}{2\Psi_{0,12}^{(0)}} -
\Psi_{1,21}^{(0)}\Psi_{0,12}^{(0)} +
\Psi_{1,22}^{(0)}\Psi_{0,11}^{(0)}.
\end{align}
\end{subequations}
\end{lemma}
\begin{proof}
 % We can
% then use \eqref{Y11 Y12 and derivs} and \eqref{v in Y 2} to obtain
% relations between the canonical coordinates and the function $v$.
Equations~\eqref{eq:YP}, \eqref{eq:PPsi0} and~\eqref{eq:zeta0} lead to 
\begin{equation}\label{eq:YPsi}
Y(y)=e^{\frac{Nl\sigma_3}{2}}S^{\infty}(y)e^{N\tilde{g}_+(0)\sigma_3}\hat{\Psi}_0(\zeta)
e^{\frac{Ny^2}{4}\sigma_3},\qquad \zeta\rightarrow 0.
\end{equation}
Inverting the conformal map~\eqref{eq:zeta0} gives
\[
e^{\pm\frac{Ny^2}{4}} = \left(1 \pm \frac{\zeta^2}{4N} +
\frac{\zeta^4}{32N^2}+ \cdots\right)\left(1+O(N^{-1})\right).
\]
When $\ipart(\zeta)<0$, formula~\eqref{eq:para} allows write the outer
parametrix $S^\infty(y)$ in terms of $\zeta$:
\begin{equation}\label{Sinf zeta}
S^\infty(\zeta) = \frac{1}{\sqrt{2}}
\begin{pmatrix}
  1 + \frac{i\zeta}{4N} + \frac{\zeta^2}{32N^2} & -\left(1 - \frac{i\zeta}{4N} + 
\frac{\zeta^2}{32N^2}\right)  \\
  1 - \frac{i\zeta}{4N} + \frac{\zeta^2}{32N^2} & 1 + \frac{i\zeta}{4N} +
  \frac{\zeta^2}{32N^2}
\end{pmatrix}
\left(1 + O(N^{-1})\right),
\end{equation}
where we have used the expansion
\begin{equation*}%\label{gamma exp}
\gamma^{\pm 1} = e^{\mp i \frac{\pi}{4}}\left(1 \mp \frac{\zeta}{4N}
+ \frac{\zeta^2}{32N^2} + \cdots \right) \left( 1 + O(N^{-1})
\right).
\end{equation*}
Furthermore, Eq.~\eqref{eq:YPsi} implies
\begin{subequations}
\label{Y Psi symmetries}
\begin{align}
  Y_{11}(\zeta) &=
  e^{\frac{N(2l+y^2)}{4}}\left(S^\infty_{11}(\zeta)\left(\hat{\Psi}_0(\zeta)\right)_{11}
    e^{-i\frac{N\pi}{2}}+
    S^\infty_{12}(\zeta)\left(\hat{\Psi}_0(\zeta)\right)_{21}
e^{i\frac{N\pi}{2}}\right), \\
  Y_{12}(\zeta) &=
  e^{\frac{N(2l-y^2)}{4}}\left(S^\infty_{11}(\zeta)\left(\hat{\Psi}_0(\zeta)\right)_{12}e^{-i\frac{N\pi}{2}}+
    S^\infty_{12}(\zeta)\left(\hat{\Psi}_0(\zeta)\right)_{22}e^{i\frac{N\pi}{2}}\right).
\end{align}
\end{subequations}
Combining~\eqref{Sinf zeta} and~\eqref{Y Psi symmetries} give formulae~\eqref{Y11 Y12 and derivs}.

In order to prove Eqs.~\eqref{eq:can_Psi}, note that from~\eqref{eq:A}
and~\eqref{eq:Azeta} we have 
\[
A_3 = \frac{u_2}{2}\Psi_0^{(0)} \sigma_3
\left(\Psi_0^{(0)}\right)^{-1}.
\]
Since $u_1=0$, $P_2=a_3=0$; therefore, we can write
\[
A_3 = \begin{pmatrix} 0 & b_3 \\ c_3 & 0 \end{pmatrix} =
\frac{u_2}{2}
\begin{pmatrix} \Psi_{0,11}^{(0)}\Psi_{0,22}^{(0)} + \Psi_{0,12}^{(0)}\Psi_{0,21}^{(0)} & -2\Psi_{0,11}^{(0)}\Psi_{0,12}^{(0)} \\ 2\Psi_{0,21}^{(0)}\Psi_{0,22}^{(0)} & -\Psi_{0,11}^{(0)}\Psi_{0,22}^{(0)} - \Psi_{0,12}^{(0)}\Psi_{0,21}^{(0)} \end{pmatrix}.
\]
The condition $\det(\Psi_0^{(0)}) = 1$ implies 
\begin{subequations}
\label{eq:b3c3}
\begin{alignat}{2}
\Psi_{0,22}^{(0)}& = \frac{1}{2\Psi_{0,11}^{(0)}}, &\qquad
\Psi_{0,21}^{(0)} & = -\frac{1}{2\Psi_{0,12}^{(0)}},\\
\label{eq:b3c3b} 
b_3& =-u_2\Psi_{0,11}^{(0)}\Psi_{0,12}^{(0)},  & c_3 & = -
\frac{u_2}{4\Psi_{0,11}^{(0)}\Psi_{0,12}^{(0)}}.
\end{alignat}
\end{subequations}
These relations together with  Eqs.~\eqref{eq:A}, \eqref{eq:Azeta} and
\eqref{eq:canon} lead to formulae~\eqref{eq:can_Psi}.
\end{proof}
%Lemma 81 was here

Using Eqs.~\eqref{v in Y 2} and~\eqref{Y11 Y12 and derivs} we can write
$v_X(\omega)$ in terms of the elements of the matrices $\Psi_0^{(0)}$ and
$\Psi_1^{(0)}$. Recall that $\omega$, $u_2$ and $N$ are related
by~\eqref{eq:tu2_rel} and that $v_X=O(\sqrt{N})$. Let us also set $r
=\sqrt{u_2}$. We have
%Recall that,
%\begin{equation*}
%v(t) = \begin{cases}
%\frac{\sqrt{N}e^{Nl}}{ i \sqrt{u_2} Y_{11}(0,u_2)\p_{\zeta}Y_{12,0}(0,u_2)}(1 + O(N^{-\frac{1}{2}})), & N \text{ even}, \vspace{0.5ex} \\
%\frac{\sqrt{N}e^{Nl}}{ i \sqrt{u_2}
%\p_{\zeta}Y_{11}(0,u_2)Y_{12}(0,u_2)}(1 + O(N^{-\frac{1}{2}})), & N
%\text{ odd}.
%\end{cases}\end{equation*}
%From \eqref{v in Y 2} and \eqref{Y11 Y12 and derivs}, we have
\begin{equation}
 \label{eq:vx_vy}
  v_X(\omega)=\sqrt{N}v_Y(r)\Bigl(1+O\bigl(N^{-1/2}\bigr)\Bigr),
\end{equation}
where
\begin{equation}
\label{v in psi}
v_Y^{-1}(r)= 
\begin{cases}
\frac{ ir}{2 } 
\left(
  \Psi_{0,11}^{(0)}-\Psi_{0,21}^{(0)}\right)\left(\Psi_{1,12}^{(0)}  
- \Psi_{1,22}^{(0)}\right), & \text{$N$ even}, \\
-\frac{ir}{2}\left(\Psi_{1,11}^{(0)}+\Psi_{1,21}^{(0)}
\right)\left( \Psi_{0,12}^{(0)}+ \Psi_{0,22}^{(0)}\right), & \text{$N$ odd}.
\end{cases}
\end{equation}
By inserting Eq.~\eqref{eq:vx_vy} into the ODE~\eqref{eq:PIII_IC} and taking the
leading order term, we see that $v_Y(r)$ satisfies the PIII
equation~\eqref{eq:PIII_vy} in Proposition~\ref{pro:PIII}.

%  Before we carry on, let us
% show that $v_Y$ will satisfy the Painlev\'e III equation stated in
% Theorem \ref{thm:main2}.

% Recall that from \cite{CI10}, $v_X(\omega)$ solves the PIII
% equation~\eqref{eq:PIII_IC}  

% \[
% v^{\prime\prime} = \frac{(v^\prime)^2}{v} - \frac{v^\prime}{t} -
% \frac{1}{t}\left(4\alpha v^2+4(2n+1+\alpha)\right) +4v^3 -
% \frac{4}{v}.
% \]
% In our variables, $v$ is of order $O(\sqrt{N})$, so by writing
% $v(t)=\sqrt{N}v_Y(r)\left(1+O\left(N^{-\frac{1}{2}}\right)\right)$,
% and take the scaling $t=r/2\sqrt{N}$, $r=\sqrt{u_2}$ and then
% consider the leading order term in the above equation, we see that
% $v_Y(r)$ satisfies the following PIII equation
% \[
% v_Y^{\prime\prime} =  \frac{(v_Y^\prime)^2}{v_Y} -
% \frac{v_Y^\prime}{r} -\frac{2\alpha v_Y^2}{r}-\frac{2}{r}+ v_Y^3.
% \]

The last step of the proof of Proposition~\ref{pro:PIII} consists in
eliminating the entries of $\Psi_0^{(0)}$ and $\Psi^{(0)}_1$ from
Eqs.~\eqref{eq:can_Psi} and~\eqref{v in psi}, thereby exprissing the
canonical coordinates $P_1$ and $Q_1$ in terms of $v_Y(r)$.

% We will now further simplify the relations
% \eqref{v in psi} by eliminating some of the entries in
% $\Psi_{0}^{(0)}$ and $\Psi_{1}^{(0)}$.

From~\eqref{eq:A} and the compatibility condition
\[
\partial_{u_2} A(\zeta) - \partial_\zeta B(\zeta) +
[A(\zeta),B(\zeta)]=0
\]
 we obtain
\begin{equation} 
\label{b3}
\begin{pmatrix} 0 &\partial_{u_2} b_3 \\ \partial_{u_2} c_3 & 0 \end{pmatrix} 
= \begin{pmatrix} \frac{1}{2u_2}\left( b_1c_3-c_1b_3 \right) & \frac{b_3}{u_2}
  \\ 
\frac{c_3}{u_2} & \frac{1}{2u_2}\left( c_1b_3-b_1c_3\right)
\end{pmatrix}.
\end{equation}
Hence, $b_3 = cu_2$ for some $c \in \mathbb{C}$. To determine the
constant $c$, we will use the asymptotic expansion for
$\hat{\Psi}_0(\zeta)$ in $u_2$ derived in Sec.~\ref{se:asymu2}.
Recall that  
\begin{equation}
\label{eq:Rphi2}
\hat{\Psi}_0(\zeta)=R_{\Phi}(\zeta)\Phi^{(p)}(\zeta)e^{i\zeta\sigma_3},
\end{equation}
where $R_{\Phi}(\zeta)$ was defined in~\eqref{eq:Rphi}.  Setting $\zeta=0$
in~\eqref{eq:Rphi2}, we find $\Psi_0^{(0)}$. Then, by \eqref{eq:RPhi} we have
\begin{equation*}
\hat{\Psi}_0(0)=\Bigl(I+O\bigl(\sqrt{u_2}\bigr)\Bigr)\Phi^{(p)}(0).
\end{equation*}
From the definition of $\Phi^{(p)}$ in \eqref{eq:philocal} we see that
\begin{equation*}
\Phi^{(p)}(0)=\begin{pmatrix}1&1\\
-1&0\end{pmatrix}\begin{pmatrix}1&\left(\phi(0)-\phi_{+,0}\right)\\
0&1\end{pmatrix}.
\end{equation*}
Lemma~\ref{pro:recur} gives $\phi(0) = 0$ and $\phi_{+,0} = \frac{1}{2}$. Thus,
the first relation in Eq.~\eqref{eq:b3c3b} becomes
\begin{equation}\label{c}
- \Psi_{0,11}^{(0)}\Psi_{0,12}^{(0)} = -\frac{1}{2} = c.
\end{equation}

Next, the constraint $\det\left(\hat{\Psi}_0\right)=1$ gives the following
identities among the entries of $\Psi_0^{(0)}$ and $\Psi_1^{(0)}$: 
\begin{subequations}
\label{eq:detid}
\begin{gather}
\Psi_{0,11}^{(0)}\Psi_{0,22}^{(0)} - \Psi_{0,12}^{(0)}\Psi_{0,21}^{(0)}=1,\\
\label{eq:detidb}
\Psi_{1,11}^{(0)}\Psi_{0,22}^{(0)} +
\Psi_{1,22}^{(0)}\Psi_{0,11}^{(0)} -
\Psi_{1,21}^{(0)}\Psi_{0,12}^{(0)} -
\Psi_{1,12}^{(0)}\Psi_{0,21}^{(0)} = 0.
\end{gather}
\end{subequations}
Definition~\eqref{eq:lax_A} leads to 
\begin{equation*}
\begin{split}
A(\zeta)& =\frac{\partial \hat{\Psi}_0(\zeta)}{\partial
\zeta}\hat{\Psi}_0^{-1}(\zeta) + 
\frac{u_2}{2\zeta^3}\hat{\Psi}_0(\zeta) \sigma_3
\hat{\Psi}_0^{-1}(\zeta)\\
& = \hat{\Psi}_0(\zeta)\left(  \frac{u_2}{2\zeta^3}\sigma_3 + 
\hat{\Psi}_0^{-1}(\zeta) \Psi_1^{(0)}\bigl(I + O(\zeta)\bigr)\right)\hat{\Psi}_0^{-1}(\zeta)  \\
&= \hat{\Psi}_0(\zeta)\left( \frac{u_2}{2\zeta^3}\sigma_3 +
\left(\Psi_0^{(0)}\right)^{-1} \Psi_1^{(0)}\bigl(I +
O(\zeta)\bigr)\right)\hat{\Psi}_0^{-1}(\zeta).
\end{split}
\end{equation*}
By taking the determinant of the right-hand side and looking at the
coefficient of $\zeta^{-3}$ we obtain
\begin{equation}\label{Psi1 id2}
-2\left(b_2c_1+b_1c_2\right) =
\left(\Psi_{0,11}^{(0)}\Psi_{1,22}^{(0)}+\Psi_{0,12}^{(0)}\Psi_{1,21}^{(0)}-\Psi_{0,21}^{(0)}\Psi_{1,12}^{(0)}
-\Psi_{0,22}^{(0)}\Psi_{1,11}^{(0)}\right)u_2.
\end{equation}
By using \eqref{eq:coefrel},  \eqref{eq:red1} and the fact
that $u_1=0$, we see that
\begin{equation*}
c_1=\frac{u_2^2}{4}\frac{b_1}{b_3^2},\qquad
c_2=\frac{u_2^2}{4}\frac{b_2}{b_3^2}.
\end{equation*}
Therefore, we also have
\begin{equation}\label{eq:b1c3}
b_2c_1+b_1c_2=0.
\end{equation}
We can now insert \eqref{eq:b1c3} into~\eqref{Psi1 id2} and use
\eqref{eq:detidb} to  arrive at
\begin{equation}
\label{eq:psi1 psi0}
\Psi_{0,11}^{(0)}\Psi_{1,22}^{(0)}=\Psi_{0,21}^{(0)}\Psi_{1,12}^{(0)},\qquad
\Psi_{0,12}^{(0)}\Psi_{1,21}^{(0)}=\Psi_{0,22}^{(0)}\Psi_{1,11}^{(0)}.
\end{equation}
These relations together with Eqs.\eqref{eq:b3c3} can be used to
simplify formulae~\eqref{eq:can_Psi} in Lemma \ref{lem:can
  coords}. Finally, we obtain
\begin{equation}
\label{P Q simp}
\begin{split}
P_1 = \frac{u_2}{2}\left(
\frac{\Psi_{1,21}^{(0)}}{\Psi_{0,11}^{(0)}} +
2\Psi_{1,22}^{(0)}\Psi_{0,11}^{(0)} \right), \qquad Q_1 = -
\frac{\Psi_{1,21}^{(0)}}{\Psi_{0,11}^{(0)}} +
2\Psi_{1,22}^{(0)}\Psi_{0,11}^{(0)}.
\end{split}
\end{equation}

We are now in a position to write $P_1$ and $Q_1$ in terms of the
transcendent $v_Y(r)$. We shall outline the calculation for $N$ even; the
derivation when $N$ is odd is almost identical. 

Using \eqref{c}, \eqref{eq:b3c3}, \eqref{eq:b1c3} and \eqref{v in psi} leads to
\[
v_Y^{-1}(r) =
-2ir\Psi_{0,11}^{(0)}\Psi_{1,22}^{(0)}.
\] 
Thus, by~\eqref{P Q simp} we have
\begin{equation}
\label{v0 in P Q even}
v_Y^{-1}(r) = -ir\left(\frac{P_1(r^2)}{r^2}+
\frac{Q_1(r^2)}{2}\right).
\end{equation}
The initial
conditions~\eqref{eq:asympcan_a} and~\eqref{eq:asympcan_b} provide the
asymptotic behaviour of $v_Y^{-1}(r)$ as $r\rightarrow 0$ and
prove~\eqref{eq:vasymeven2}. Higher order terms can be computed;
however, they are rather complicated and we shall not write them down
explicitly. % Similarly, for $N$ odd we have
% \begin{equation}
% \label{v0 in P Q odd}
% v_Y^{-1}(r) = -ir\left(\frac{P_1(r^2)}{r^2} - \frac{Q_1(r^2)}{2}
% \right),
% \end{equation}
% with initial conditions given by \eqref{eq:vasymodd2}.

Differentiating $v_Y^{-1}(r)$ and using the relations~\eqref{eq:red1}
between $P_1$, $Q_1$ and the Hamilton equations~\eqref{eq:2ndorder} gives
\begin{equation}
\label{eq:dvy}
\begin{split}
\frac{dv_Y}{dr}
&= i v_Y^2(r)\left( \frac{\p_{r}P_1}{r} - \frac{P_1}{r^2} + \frac{r\p_{r}Q_1}{2} + \frac{Q_1}{2}\right)\\
%&= -i\sqrt{u_2}\left(\frac{P_1}{u_2}+\frac{Q_1}{2}\right)\frac{u_2^2Q_1^2-4P_1^2}{2u_2^\frac{3}{2}} - 1 \\
&= v_Y(r)\frac{4P_1^2(r^2)-r^4Q_1^2(r^2)}{2r^3} + v_Y^2(r).
%&= \frac{ \sqrt{u}}{2i} \left(\frac{u_2^2Q^2-4P^2}{4u_2^2}\left(\frac{P}{2u_2}+ \frac{Q}{4}\right) - \frac{i}{4u_2}\right) \\
%&= \frac{u_2^2Q^2-4P^2}{4v_Yu_2^2N^3} - \frac{i}{8Y_{-1,0}\left(
%\frac{u}{N^3}\right)\sqrt{u}}
\end{split}
\end{equation}
Equation~\eqref{v0 in P Q even} can be rearranged as 
\begin{equation*}
%\label{P Q in v even}
P_1 = \frac{ir}{v_Y(r)} - \frac{r^2 Q_1}{2}, \qquad
Q_1 = \frac{2i}{v_Y(r)r}- \frac{2P_1}{r^2}.
\end{equation*}
Inserting these expressions into Eq.~\eqref{eq:dvy} we find
$P_1$ and $Q_1$ in terms of $v_Y^{-1}$, $v_Y^\prime$ and $r$.
Namely, we have
\begin{subequations}
\label{P in v even}
\begin{align}
P_1& =  \frac{ir}{ 2v_Y(r)} +
\frac{iu_2}{4}\left(v_Y^2(r) -
v_Y^\prime(r)\right) ,\\ Q_1&=
\frac{i}{rv_Y(r)}+\frac{i}{2}\left(v_Y^{\prime}(r)-v_Y^2(r)\right).
\end{align}
\end{subequations}
% And for $N$ odd,
% \begin{align*}
% \frac{dv_Y}{dr} &= v_Y\frac{r^4Q_1^2(r^2)-4P_1^2(r^2)}{2r^3}-v_Y^2,
% \end{align*}
% in this case,
% \begin{equation} \label{P in v odd}
% \begin{split}
% %P = \frac{4Y_{-1,0}\left( \frac{u}{N^3}\right)\sqrt{u}}{5N^3 v_Y} - \frac{5u^\frac{3}{2}v_Y^2(v_Y^{-1})^\prime}{16Y_{-1,0}\left( \frac{u}{N^3}\right)} + \frac{25iv_Y^2u}{128Y_{-1,0}^2\left( \frac{u}{N^3}\right)}
% P_1& =
% \frac{iu_2}{4}\left(v_Y^2(\sqrt{u_2})+v_Y^{\prime}(\sqrt{u_2})\right)+
% \frac{i\sqrt{u_2}}{2v_Y(\sqrt{u_2})} ,\\ Q_1& =
% -\frac{i}{v_Y(\sqrt{u_2})\sqrt{u_2}}+
% \frac{i}{2}\left(v_Y^2(\sqrt{u_2})+v_Y^{\prime}(\sqrt{u_2})\right).
% %Q= \frac{5N^3v_Y^2\sqrt{u}(v_Y^{-1})^\prime}{8Y_{-1,0}\left( \frac{u}{N^3}\right)} - \frac{8Y_{-1,0}\left( \frac{u}{N^3}\right)}{5v_Y\sqrt{u}} - \frac{25iN^3v_Y^2}{64Y_{-1,0}^2\left( \frac{u}{N^3}\right)}
% \end{split}
% \end{equation}
This completes the proof of Proposition~\ref{pro:PIII}.
\bibliographystyle{amsplain}
%\bibliography{bib_database}
\providecommand{\bysame}{\leavevmode\hbox to3em{\hrulefill}\thinspace}
\providecommand{\MR}{\relax\ifhmode\unskip\space\fi MR }
% \MRhref is called by the amsart/book/proc definition of \MR.
\providecommand{\MRhref}[2]{%
  \href{http://www.ams.org/mathscinet-getitem?mr=#1}{#2}
}
\providecommand{\href}[2]{#2}

\noindent\rule{16.2cm}{.5pt}

\vspace{.25cm}

{\small \noindent {\sl School of Mathematics \\
                       University of Bristol\\
                       Bristol BS8 1TW, UK  \\
                       Email: {\tt l.brightmore@bristol.ac.uk}\\
                       Email: {\tt f.mezzadri@bristol.ac.uk}\\
                       Email: {\tt m.mo@bristol.ac.uk}

                       \vspace{.25cm}

                       \noindent 20 September 2013}

\end{document}